\documentclass[prx,aps,superscriptaddress,groupedaddress,nofootinbib, twocolumn ]{revtex4-1}

\usepackage{amsmath,amsfonts,amsthm,amssymb,bbm,graphicx,color,enumerate,hyperref,xcolor,tikz,pgfplots,ulem,subfigure}
\usepackage{microtype}
\usepackage[noend]{algorithm2e}
\usepackage[colorinlistoftodos]{todonotes}
\usepackage{MnSymbol}
\usepackage{tikz}
\usetikzlibrary{quantikz}

\hypersetup{
	colorlinks=true,  
	linkcolor=blue,   
	citecolor=blue,   
	urlcolor=blue     
}
\def\C{\mathcal{C}}
\def\K{ {\mathcal K} }
\def\E{ {\mathcal E} }
\def\P{ {\mathcal P} }
\def\h{ {\mathcal H} }
\def\H{ {\mathcal H} }

\def\U{ {\mathcal U} }
\def\R{ {\mathcal R} }

\def\N{ {\mathcal N} }

\def\A{ {\mathcal A} }
\def\B{ {\mathcal B} }
\def\O{ {\mathcal O} }
\def\P{ {\mathcal P} }
\def\D{ {\mathcal D} }	
\def\S{ {\mathcal S} }
\def\T{ {\mathcal T} }

\def\I{ {\mathcal I} }

\def\W{{\bf{W}}}

\def\g{\boldsymbol{g}}
\def\r{\boldsymbol{r}}
\def\s{\boldsymbol{s}}
\def\p{\boldsymbol{p}}

\def\x{\boldsymbol{x}}

\def\z{\boldsymbol{z}}

\def\Tr{{\rm{tr}}}
\def\tr{ \mbox{tr} }

\def\>{\rangle}
\def\<{\langle}
\def\kv{\rangle\rangle}
\def\bv{\langle\langle}

\def\hc{^{\dagger}}

\renewcommand{\emph}{\textit}

\newcommand{\iden}{\mathbb{I}}

\newcommand{\clebsch}[6]{\langle #1,#2;#3,#4\hspace{0.5px} | \hspace{0.5px}#5,#6\rangle}

\newcommand{\cc}[1]{\textcolor{magenta}{#1}}

\newtheorem{theorem}{Theorem}

\newtheorem{lemma}[theorem]{Lemma}

\newtheorem{definition}[theorem]{Definition}

\newtheorem{corollary}[theorem]{Corollary}

\usetikzlibrary{calc}

\definecolor{block}{RGB}{0,162,232}

\newenvironment{blockmatrix}{%
	\left(%
	\vcenter\bgroup\hbox\bgroup
	\tikzpicture[
	x=1.5\baselineskip,
	y=1.5\baselineskip,
	]%
}{%
	\endtikzpicture
	\egroup
	\egroup
	\right)%
}

\newcommand*{\block}[1][block]{%
	\blockaux{#1}%
}
\def\blockaux#1(#2,#3)#4(#5,#6){%
	\draw[fill={#1}]
	let \p1=(#2,#3),
	\p2=(#5,#6),
	\p3=(#2+#5,#3+#6),
	\p4=(#2+#5/2,#3+#6/2),
	\p5 = (#2+#5+#5/2,#3+#6/2)
	in
	(\p1) rectangle (\p3)
	(\p4) node {$#4$};
}

\usepackage{ulem}

\usepackage{scalerel,stackengine}
\stackMath
\newcommand\reallywidehat[1]{%
	\savestack{\tmpbox}{\stretchto{%
			\scaleto{%
				\scalerel*[\widthof{\ensuremath{#1}}]{\kern-.6pt\bigwedge\kern-.6pt}%
				{\rule[-\textheight/2]{1ex}{\textheight}}
			}{\textheight}%
		}{0.5ex}}%
	\stackon[1pt]{#1}{\tmpbox}%
}
\begin{document}

\title{A Fourier analysis framework for approximate classical simulations of quantum circuits}

\author{Cristina C\^irstoiu}
\affiliation{Quantinuum, Terrington House, 13-15 Hills Road, Cambridge CB2 1NL, UK}

\begin{abstract}

	What makes a class of quantum circuits efficiently classically simulable on average? 
	I present a framework that applies harmonic analysis of groups to circuits with a structure encoded by group parameters. Expanding the circuits in a suitable truncated multi-path operator basis gives algorithms to evaluate the Fourier coefficients of output distributions or expectation values that are viewed as functions on the group. Under certain conditions, a truncated Fourier series can be efficiently estimated with guaranteed mean-square convergence. For classes of noisy circuits, it leads to algorithms for sampling and mean value estimation under error models with a spectral gap, where the complexity increases exponentially with the gap's inverse and polynomially with the circuit's size. This approach unifies and extends existing algorithms for noisy parametrised or random circuits using Pauli basis paths. For classes of noiseless circuits, mean values satisfying Lipschitz continuity can be on average approximated using efficient sparse Fourier decompositions.  I also discuss generalisations to homogeneous spaces,  qudit systems and a way to analyse random circuits via matrix coefficients of irreducible representations.
\end{abstract}
	\maketitle
\section{Introduction}

Quantum algorithms with guaranteed speed-ups over classical computation often exhibit a delicate interplay between structure and randomness.  For highly structured problems such as factoring,  quantum computations give an exponential advantage.  Random circuit sampling \cite{boixo2018characterizing,hangleiter2023computational} has been used as a basis for experimental demonstration of quantum computations beyond classical methods \cite{morvan2023phase,decross2024computational}.



A complexity separation between classical and quantum remains in the presence of noise for particular algorithms solving structured problems like Simon's and Bernstein-Vazirani \cite{chen2023complexity}, but in many cases noisy computations undermine asymptotic quantum speed-ups \cite{francca2021efficient, stilck2021limitations}. For example, the works of Aharanov et. al  \cite{aharonov2023polynomial} and Gao et. al \cite{gao2018efficient}  have led to a polynomial-time classical algorithm to approximate the task of sampling from  a random quantum circuit under a finite depolarising error. Noise effectively truncates the number of significant contributions to the outcome probabilities of noisy random circuits, thus enabling a classical simulation that is efficient in the number of qubits $n$. While the complexity is polynomial in $n$, its degree scales inversely with the error rate.  In \cite{fontana2023classical} we used similar ideas based on truncated evolution of observables to construct an efficient classical algorithm (LOWESA - low weight simulation algorithm) that approximates the expectation values of parametrised quantum circuits under incoherent Pauli noise. Different truncation schemes were subsequently used in \cite{shao2023simulating}. Non-unital noise leads to a similar behaviour where only the last log-depth number of layers in a random circuit contribute to noisy expectation values of Pauli operators \cite{mele2024noise}.

Qualitatively, given that emergence of classicality can be attributed to decoherence from interactions with an environment \cite{zurek2003decoherence} it is not surprising that ``de-quantising" quantum algorithms may be possible once enough errors occur. In an asymptotic sense, fault tolerance permits arbitrarily accurate quantum computations at the expense of  poly-logarithmic overheads. Quantum computers have reached sizes and levels of physical noise to allow first demonstrations of quantum error correction components implemented on hardware \cite{ryan2024high,wang2024fault, google2023suppressing,bluvstein2024logical}. From a practical sense, a quantum device is a finite-sized system, so a trade-off occurs between size of the logical computation and overheads in suppressing noise to a given level. 

\emph{When do quantum computations become easier to simulate classically, in the presence of a finite level of errors?} 


We develop a framework that allows to separate randomness from the fixed structure in an ensemble of circuits by leveraging harmonic analysis of groups - a generalisation of Fourier analysis. In doing so, we can address the above question by constructing classical algorithms that approximate noisy quantum tasks such as mean value estimation and sampling for circuits subject to finite (fixed) error rate.  As particular examples, we also recover several existing results  \cite{aharonov2023polynomial, fontana2023classical,shao2023simulating} under a unified framework.

We considers classes of unitaries, $\mathfrak{C}_n$, on $n$ qudit systems that decompose into a particular "circuit" structure where each unitary is either fixed or specified by a choice of independent elements of a group. The first part of the framework concerns ensembles of \emph{noiseless circuits} and we focus on two types of quantum tasks:  (i) sampling from the output distribution of these unitaries and (ii) computing expectation values of observables. 

Given the group parametrisation of the circuit ensembles, then mean values and probability distributions can be viewed as real-valued functions on the group. Therefore, they can be expressed as linear combinations of orthonormal basis functions with a natural choice being the matrix coefficients of irreducible representations.  In classical Fourier analysis this basis corresponds to trigonometric functions. For general groups the Fourier coefficient is a matrix of dimension equal to the irreducible representation it is associated with.  Generally, the resulting Fourier series is finite but can contain a number of terms that scale exponentially with the number of qudits, which is classically intractable. Therefore, we want to understand under what conditions can an appropriately truncated Fourier series (which we can efficiently evaluate classically) provide a good approximation of mean values or probability distributions for the circuit ensemble. 
The framework involves two ingredients. First, we need to determine a truncation set and provide a classical algorithm to evaluate the corresponding Fourier series. We show that for \emph{sparse} Fourier decompositions then the low-dimensional Fourier coefficients can be efficiently evaluated from the circuit's structure by decomposing it in terms of an appropriately truncated sum over paths with respect to a symmetry-adapted operator basis. We can employ tree structure to describe this decomposition - where the height of the trees directly correspond to the truncation parameter in the Fourier series.  Second, we need to quantify the approximation error. For this we will make use of a general form of Parseval's theorem, which will control an $L_2$ average approximation error over the group parameters. Figure 1 provides the high-level overview of the framework.

For general circuits, one shouldn't expect such sparse Fourier series to provide good approximations of mean values (or outcome distributions), however leveraging results from \cite{daher2019titchmarsh} we can guarantee that in certain cases the average mean square error converges with increasing the truncation parameter. This assumes that viewed as functions over the group parameter space the expectation values satisfy a regularity Lipschitz-type condition. It remains an interesting open question to classify what type of circuits can result in such conditions being met and explore the role of additional regularity constraints. 

The second part of the framework looks at conditions that permit de-quantisation of noisy computations (e.g computing mean values or output probabilities) with noise models that have a fixed spectral gap $\gamma$.  For simple depolarising models that have appeared in previous works, the spectral gap can be straightforwardly related to the error rate. The main idea is that particular noise models  induce an effective truncation where the most significant contributions come from a sparse Fourier series.  Under certain assumptions that we carefully highlight, we explicitly construct classical algorithms that have a polynomial runtime in the number of qubits and exponential with the inverse spectral gap of the error models. With probability at least $1-\delta$ over the  class of unitary circuits $\mathfrak{C}_n$,  we can compute  noisy expectation values  of an observable $O$ with an average approximation error $\epsilon$ using a classical algorithm with runtime $poly(n)\left(\frac{\O(1)||O||_2}{\epsilon \sqrt{\delta}}\right)^{\O(1/\gamma)}$. For the sampling task, the runtime for obtaining a sample from a probability distribution that is at most $\epsilon$-away in total variation distance from the noisy probability distribution is given by  $poly(n)\left(\frac{\O(1)}{\epsilon \sqrt{\delta}}\right)^{\O(1/\gamma)}. $  This complexity is obtained under the assumption that the residual error distribution in the approximation anti-concentrates. Intuitively, this requires that the truncated classical algorithm is a good approximation for the peaks in the noiseless probability distribution.   The result also holds when the noise-less probability distribution anti-concentrates, but this is a stronger requirement that is not fully necessary.

The harmonic analysis perspective naturally separates randomness of the ensembles (captured by the basis of functions which only depend on the group) from the fixed circuit structure (captured by the Fourier coefficients).  For example, we can view random circuit sampling broken down into functions dealing with randomness of the gates (given by matrix coefficients of irreducible representations of $SU(4)$ for two qubit gate unitaries) and the fixed structure given by the circuit geometry. We describe certain constraints on the circuit architecture may lead to improved classical simulations of noisy random circuits - a point also observed in \cite{decross2024computational}. While the strong-noise regime can be classically spoofable (and provably so in the asymptotic regimes), it remains an open question to explore the extent to which the algorithmic constructions presented here can be used as heuristics for finite-size systems with reasonable computational resources. 

We emphasise that these results are asymptotic by nature, and do not necessarily guarantee that the classical algorithms will produce a good approximation of the noisy quantum task outcome in the \emph{finite-size regime}. However, for fixed number of qubits the methods can be employed as \emph{heuristic classical simulations}.  For example, \cite{rudolph2023classical} and \cite{beguvsic2024fast} use multi-path Pauli basis decompositions with different types of truncations to classically estimate mean values of observables under a quantum quench dynamics of an Ising model.  ZX-calculus \cite{coecke2011interacting} and parametric graphical re-write rules have been used in \cite{sutcliffe2024fast} to simplify and reduce computational resources in these multi-path methods and in \cite{koch2023speedy} to perform faster contractions. 

Another important aspect is the fact that obtaining noise-induced de-quantisation in this framework relies on approximation errors that are $\emph{averaged over ensembles of circuits}$.  This is rather unsatisfactory, but also common aspect in prior literature using classical simulations based on Pauli-path decompositions.  When the target computation is an expectation value for a $\emph{fixed}$ circuit it can be misleading to infer it is ``classically simulable'' from the construction of a polynomial-time algorithm that achieves an $\epsilon$ average error over a class of circuits even when including the fixed one. This is less of an issue when discussing \emph{sampling from a class of circuits}, for which average measures might be suitable when using the same distributions. Recent work \cite{schuster2024polynomial} provide polynomial-time algorithms for noisy circuits using Pauli-path decompositions and worst-case approximation error, but rely on other assumptions that inputs form a mixed ensemble.  It remains an interesting question if such an analysis can be incorporated into this framework to guarantee pointwise convergence of the truncated Fourier series rather than the $L_2$ convergence provided by Parseval's theorem.

\section{Overview of classical simulation algorithms and path decompositions}

What do we mean when we say a \emph{quantum computation can be efficiently simulated by a classical algorithm} ? To define what a simulation of a quantum computation is one must first establish (i) what type of quantum computational task is the output of the classical algorithm approximating and (ii) what type of precision is allowed in the approximation of that task.  Typically, efficiency is claimed when the time and space resources required for the classical algorithm grow at most poly-logarithmically with the dimension of the space in which the quantum computation lives. Such scaling statements assume an asymptotic behaviour and that requires a description of the quantum computation for each system size.  

Several classes of circuits such as Clifford (or stabilizer) \cite{aaronson2004improved, gottesman1998heisenberg}, matchgates (or free-fermionic)  \cite{terhal2002classical, jozsa2008matchgates} are well known to be efficient to classically simulate in polynomial time. 

Given a quantum system of $n$ qudits described by the Hilbert space $\h_n \approx (\mathbb{C}^d)^{\otimes n} $ of $n$ we consider families $\mathfrak{C}_n$ of unitaries (polynomial-sized quantum circuits) $C\in \mathfrak{C}_n$ acting on the $n$-qudit computational input state $|{\bf{0} }\> = |0\>^{\otimes n}$. Generally, quantum applications will involve one of the following quantum computational tasks.

\emph{Born probability estimation task}  Given a circuit $C\in \mathfrak{C}_n$ compute any marginal outcome probability distribution of the quantum state $|\psi\> =C|{\bf{0}}\>$ over subsets of qudits 
\begin{align}
p(x_1,..., x_k) = \Tr( |x_1...x_k\>\<x_1... x_k| \otimes \mathbb{I} \  |\psi\>\<\psi|),
\end{align}
where $|x_1,... x_k\>$ is a computational basis state with each $x_i \in \mathbb{Z}_d$ and $\mathbb{I}$ acts as identity on the remaining qudits.

\emph{Sampling task} Produce a sample ${\bf{x}} = (x_1,..., x_n) \in \mathbb{Z}_d^{\times n}$ drawn according to the output probability distribution $p(\bf{x}) =|\<\bf{x}|\psi\>|^2$, of the state $\psi$  over measurement outcomes in the computational basis.

\emph{Mean value estimation task} Given an observable $O$, compute its expectation value 	$\<O\> = \Tr(O |\psi\>\<\psi|)$ with respect to the target state $|\psi\>$.

Classical algorithms that output Born probability estimates within a target accuracy  $\epsilon$  and runtime $poly(1/\epsilon, n)$ are said to perform an efficient \emph{strong} simulation.  By invoking a sampling to compute reduction \cite{bremner2017achieving}, strong simulation can be used to perform an efficient \emph{weak simulation} that produces samples drawn from a probability distribution $\epsilon$-close to the target distribution of $\psi$ over measurement outcomes.  However, a classical sampling algorithm can also be constructed \cite{bravyi2022simulate} without necessarily requiring estimates for all the marginal probability distributions.

Consider an $n$ qudit system and a basis $\B = \{T_i\}_{i=1}^{d^{2n}}$ for operators on $\h_n\approx (\mathbb{C}^d)^{\otimes n}$ that is orthonormal with respect to the Hilbert-Schmidt inner product $\<T, T'\> = \Tr(T\hc T')$.  
Given a state $|\psi\>$ we can express its density matrix as $ |\psi\>\<\psi| = \sum_i \psi_{i_0} T_{i_0}$ for some coefficients $\psi_{i_0} = Tr(T_{i_0}\hc |\psi\>\<\psi|)$, with analogous expansions for any operator. Similarly we can express the action of a unitary $V$ with respect to this operator basis as 
\begin{equation}
	VT_i V\hc = \sum_j   v_{ij} T_j
\end{equation}
where $v_{ij} = Tr (T_j\hc  VT_i V\hc)$ are complex coefficients. More generally, we have that the Liouville matrix representation $\mathbf{E}$ of a quantum channel $\E$ has matrix entries $e_{ij} = Tr(T_j\hc \mathcal{E}(T_i) )$ and the action of the channel on each basis element expands into $\E(T_i) = \sum_j e_{ij} T_j$. We will also use a vectorised notation $|A \>\> \in \h \otimes \h^{*}$ to represent a vector with entries given by $\Tr(T_i\hc A)$. This satisfies $U\otimes U^{*}  |A\>\> = |UA U\hc\>\>$, which implies that $U\otimes U^{*}$ corresponds to the Liouville matrix representation of the unitary channel $U$. Furthermore $\{|T_i\>\>\}$ will similarly form an orthonormal basis for $\h\otimes \h^{*}$ with inner product $\<\<A|B\>\> = \Tr (A\hc B)$.

Suppose an observable $O$ undergoes a unitary evolution by $V$, so the evolved operator is given by $V\hc O V$ in the Heisenberg picture. The expectation value of the observable with respect to the state $\psi$ can be written as $Tr(O V\psi V\hc) = \<\<O|V\otimes V^{*} |\psi\>\>$.  In terms of the basis expansion we can write  $\<\<O| = \sum_{i_0} O_{i_0} \<\<T_{i_0}|$, where $O_{i_0} = Tr( T_{i_0} O)$.  Then
\begin{align}
	V\hc OV = \sum_{i_0, i_1} O_{i_0}  v_{i_1 i_0}T_{i_1}\hc.
\end{align}

More generally if $V = V_1... V_D$ is a sequence of $D$ unitary operations  then one can expand the evolved operator sequentially to get
\begin{align}
	V \hc O V   = \sum_{i_0, i_1,..., i_D}  O_{i_0}   v^{(1)}_{i_1 i_0} ... v^{(D)}_{i_{D} i_{D-1}}  T_{i_D}\hc,
	\label{eqn:pathdecomp}
\end{align}
where $v_{i_{k} i_{k-1}}^{(k)} :=  Tr(T_{i_{k-1}}\hc V_k T_{i_{k}} V\hc)$ are the matrix coefficients for each unitary $V_i$ in the basis $\B$.   Each term corresponds to a particular path of basis operators $T_{i_0} \rightarrow  T_{i_1} \rightarrow .... \rightarrow T_{i_D}$ that we denote in short by ${\mathbf{i} }
:=  i_0 \rightarrow i_1 \rightarrow ... \rightarrow i_D$.  The number of all such possible paths is given by $d^{2nD}$, and computing all terms in the decomposition is intractable for generic states $\psi$ and unitaries $V$.  However, there are specific classes of states/unitaries for which there is a choice of basis $\{T_i\}$  so that the number of terms only scales polynomially with $n$.  
Examples of such efficient sub-theories include Clifford-based decomposition with respect to Pauli operator basis \cite{gottesman1998heisenberg}, matchgates with respect to spinors/Majorana basis \cite{jozsa2008matchgates}. Furthermore, circuits with constrained number of universality enabling resource gates (e.g T-gates for Clifford or SWAP gates for matchgates) can have a controlled number of multi-paths that can be simulated classically \cite{mocherla2023extending, bravyi2016improved, ermakov2024unified}.

\begin{figure*}[t!]
	\includegraphics[width=\textwidth]{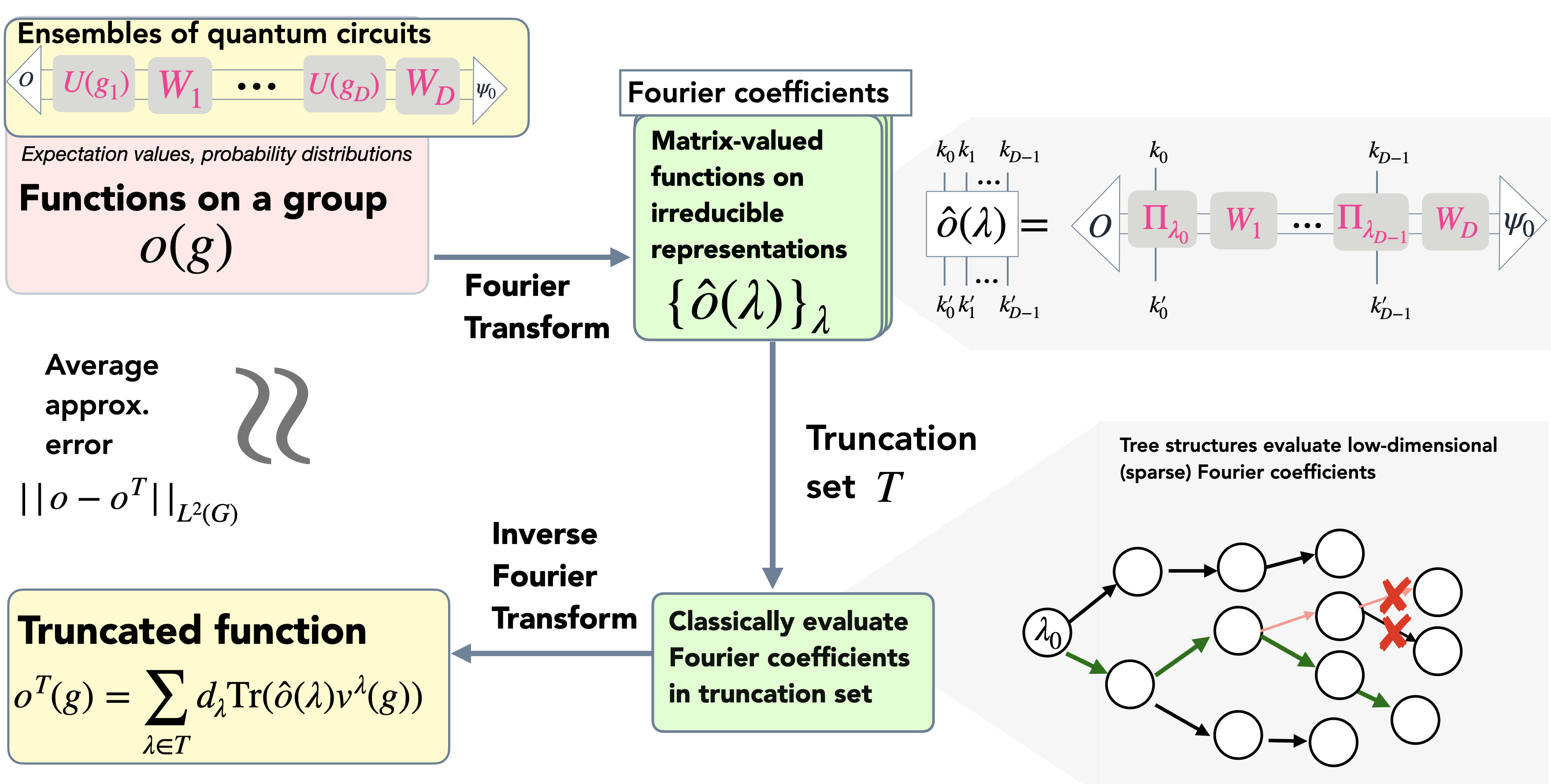}
	\label{fig:overview}
	\caption{Overview of the framework for expectation values of observable $O$. Ensembles of circuits produced by independent group elements $g=(g_1,..,g_D)$ with otherwise fixed structure results into expectation values viewed as real-valued functions  $o(g)$ on the product group $G^{\times D}$.  Applying a Fourier transform gives the Fourier coefficients which are matrix-valued functions indexed by the irreducible representations. An operator-truncated simulation will select a subset of the Fourier coefficients such that all can be efficiently evaluated. Applying the inverse Fourier transform to  the Fourier coefficients $\hat{o}(\lambda)$ in the truncation with the coefficients of matrix irreducible representations $v^{\lambda}_{kk'}(g)$ constructs a truncated function $o^T(g)$.  The $L_2$ average approximation error between $o^T$ and $o$ is given by Parseval's theorem in terms of the Hilbert-Schmidt norms of the Fourier coefficients that were not included in the truncation set. A truncation set that contains the Fourier coefficients given by the low-dimensional matrices can be evaluated by a truncated path decomposition in a symmetry-adapted basis. In a noiseless setting, smoothness will ensure the truncated function converges to $o(g)$.  In a noisy setting, we get an exponential convergence with the truncation parameter $L$ and spectral gap $\gamma$ of the noise model. An analogous set-up occurs for computing probability distributions and marginals, with the main difference being the use of bi-directional trees resulting from contractions within circuits of a particular structure.} 
\end{figure*}

\section{Main Framework and Results }
\subsection{Circuit structures  - setup}
\label{sec:defs}
We will consider classes of unitaries on $n$ qudits $\h_n$ with a particular  structure parametrised by a compact or finite group $G$, and more generally homogeneous spaces. For independent group elements $g_1,..., g_D \in G$ and fixed unitaries  $W_i$ on $\h_n$ we construct families of circuits of the form
\begin{align}
	C(\g) = U_1(g_1) W_1 U_2(g_2) W_2 \, ...  \, U_D(g_D) W_D
\end{align}
where $\g: = (g_1,..., g_D)$ is an element of a direct group product $\mathbb{G} = G^{\times D}$ and where  the unitaries $U(g)$ are parametrised by group elements $g\in G$ in such a way that they form a representation of $G$.  We denote the family of all unitaries with such a fixed structure by $\mathfrak{C}_n = \{C(\g): \g\in \mathbb{G}\}$.

\emph{Example 1: Clifford + $Z(\theta)$ }  Consider a circuit family consisting of Clifford unitaries $W_i$ and single qubit Z-rotations. 
The group of interest is $U(1)$ in this case and  $Z(\theta) = e^{iZ\theta}$ forms a representation of $U(1)$ on single qubits. With respect to the $n$-qubit Pauli basis consisting of tensor products of single qubit Pauli operators $\{I,X,Y,Z\}$ each basis operator transforms under the adjoint action of $Z(\theta)$ into a linear combination of at most two $n$-qubit Pauli operators. This follows from the fact that $Z(\theta) X Z^{\dagger}(\theta) = \cos{2\theta} X+\sin{2\theta} Y$ and $Z(\theta) Y Z^{\dagger}(\theta) = \cos{2\theta} \, Y - \sin{2\theta} \, X$.  This splits the single-qubit operator space into three components  $span\{I\}, span\{Z\}, span \{X, Y\}$ - which can be seen to be irreducible over reals (i.e not possible to split further into smaller subspaces with a linear span over the reals). 

\emph{Example 2: Random circuits} For random circuits consisting of uniform Haar-random two-qubit gates the underlying group is $G = U(4)$ and the circuit has a layered structure given by the defining representation e.g $U(g) = g$ for $g\in U(4)$.  Furthermore, $U(g) P U^{\dagger}(g)$ will involve a linear combination of at most $15$ different Pauli operators, and will fix the local two-qubit identity $I\otimes I$.
The fixed $W$'s can be thought of as either identity or a fixed permutation of the qubits. Permutation of the qubits will also corresponds to a permutation of the Pauli basis. For instance, we can choose $U_i = U$ to act on the first two qubits and any $W$ will involve a transposition of at most two elements.  Alternatively, we can pick a layered description containing tensor products of two-qubit $U(g)$ and $W$ an $n$-qubit permutation.

These examples illustrate that there is flexibility in how we realise a circuit description to fit the above structure, and that it is a  general enough framework to incorporate universal computations. This flexibility does not restrict the generality of the results, but additional care must be taken to be consistent in how we describe circuit depth or the effect of noise. For example, we may consider $W$'s to include a fixed permutation depending on which qubits the parametrised unitaries act on. In this case, unless restricted by the device's architecture, the permutations are really a mathematical convenience rather than a restriction on the circuit, so they would not be considered noisy or add to the depth of the circuit. In general, unless the scaling is explicitly stated, we will consider $D$ to be polynomial in the number of qudits, and the unitaries $U_i$ and $W_i$ may act non-trivially on the entire register. While our results will assume the group parameters in different layers $i$ to be independent, within a layer, synthesising each unitary $U_i$ may involve multiple gates that will carry a dependency on $g_i$ .

\subsection{Path decomposition in a symmetry-adapted basis}
\label{sec:multipath}

Given that $\mathfrak{C}_n$ has a well-defined structure constructed from a unitary representation $U$ of $G$, then it is natural to consider an
 \emph{Irreducible Tensor Operator (ITO) basis}, that will block-diagonalise the Liouville matrix $U(g)\otimes U^{*}(g)$ for the unitary channel of $U(g)$ 
into a block-diagonal structure independent of $g\in G$. This relies on a general result that every finite-dimensional unitary representation of a group $G$ decomposes into a direct sum of irreducible representations (for short, irreps).  A representation of the group $G$ given by a linear map $V(g):\mathcal{K} \longrightarrow \mathcal{K}$ is irreducible if $\forall |\psi\>$ and $|\phi\>$ in $\K$ with $\<\phi|\psi\>=0$  there is an element $g_0\in G$ such that  $\<\phi| V(g_0) |\psi\> \neq 0$, or equivalently if there's no strict non-trivial subspace of $\mathcal{K}$ that remains invariant under the action of $V$ for all group elements. For compact groups, one can construct all inequivalent irreps and each will be associated with a unique label $\lambda$. The irreducible matrix coefficients $v^{\lambda}_{k',k} (g) = \<k| V^{\lambda}(g) |k'\>$ are the entries of $V^{\lambda}$ with respect to a fixed orthonormal basis for the carrier space $\mathcal{K}$.

We define the ITOs as sets of operators $T^{\lambda, \alpha}_k\in \B(\h)$ that are labelled by an irreducible representation $\lambda$ with a vector component $k\in\{1,..,d_\lambda\}$ where $d_\lambda$ is the dimension of the $\lambda$-irrep. and transform according to
\begin{equation}
	U(g) T^{\lambda,\alpha}_k U^{\dagger}(g) = \sum_{k'} v^{\lambda}_{k'k}(g) T^{\lambda,\alpha}_{k'}.
\end{equation}

The $\lambda$ labels the inequivalent irreps in the decomposition of the representation $U\otimes U^{*}$. As each irrep may occur multiple times, we use $\alpha$ to label different multiplicities, but emphasise that the matrix coefficients are not dependent on the  multiplicity label. As we range over all such $\lambda$-irreps and multiplicities we can construct for $\B(\h_n)$ a complete operator basis $\{T^{\lambda,\alpha}_k\}_{\lambda,k, \alpha}$ that is orthonormal under the Hilbert-Schmidt inner product,  $\Tr[(T^{\lambda,\alpha}_{k})^{\dagger} \, T^{\mu,\beta}_{k'}]= \delta_{\mu,\lambda} \delta_{k,k'} \delta_{\alpha,\beta}$.  Recall that in terms of the vectorised notation we have $|T^{\lambda,\alpha}_k\>\> \in \h_n \otimes \h_n^{*}$ and these form a complete orthonormal basis for this space. In particular we have a resolution of identity as $\mathbb{I} = \sum_{\lambda, \alpha, k}  |T^{\lambda,\alpha}_k\>\> \<\< T^{\lambda,\alpha}_k|$.

Note that we will always have at least a trivial irrep in the decomposition of $U\otimes U^{*}$ into irreducible representations and we can always choose the identity $I$ as an element of the ITO basis, however there may generally be additional basis operators also transforming trivially if we have multiplicities. 
Of particular interest for us will be the subspace that is spanned by the $1$-dimensional trivial irreps given by the projector
\begin{align}
	\Pi_0 = \sum_{ \alpha} |T^{\lambda=0,\alpha}_{0} \>\> \<\<T^{\lambda=0,\alpha}_{0} |,
\end{align}
where the $\lambda=0$ corresponds to the trivial irrep that acts invariantly as 
$U(g) \otimes U^{*}(g) |T^{\lambda=0,\alpha}_0\>\> =  |T^{\lambda=0,\alpha}_0\>\>$ such that $|v^{\lambda=0}_{00}(g)|=1$. More generally 1-dimensional  irreps will have a proportionality coefficient given by the character $\chi_{\lambda}$ of the irrep $\lambda$ so that  $v^{\lambda}_{00}(g) = \chi_{\lambda}(g)$. The character represents the trace of a representation $\chi_V(g) = \Tr(V(g)) $ and is a basis-invariant property. In particular, for the trivial representation  $\lambda=0$ then $\chi_0(\g) = 1 \  \ \forall g$.  In our case, the identity operator $|\mathbb{I}\>\>$ is one particular example that always transforms trivially under $U\otimes U^{*}$

In the definition of the circuits in $\mathfrak{C}_n$ there may be distinct representations $U_i$, and where there is no mutually block-diagonalising operator basis then we will use different ITO basis at different ``time-slices" $i$. This requires an additional label, and that the transfer matrices for each $W_i$ involving  different the basis
\begin{align}
 \Tr((\tilde{T}_{k_i}^{\lambda_i, \alpha_i} )\hc W_i  T_{k_{i+1}}^{\lambda_{i+1}, \alpha_{i+1}}  W_i\hc)   =: w^{(i)} _{(\lambda_{j-1},k_{i-1}'), (\lambda_i, k_i)} 
\end{align}
can be computed. To simplify the notation, we will incorporate the different position-dependent basis label and the multiplicity label into the irrep label e.g $\lambda_i = (\lambda_i, \alpha_i)$.  We note that the transfer matrices will generally depend on the multiplicity label. We will make the multiplicities explicit only where it needs to be.

Because of completeness and orthonormality of the operator basis, then any observable $O\in \B(\h_n)$ decomposes uniquely into $O = \sum_{\lambda_0, \alpha_0 k_0} O_{\lambda_0,\alpha_0,k_0} (T^{\lambda_0,\alpha_0}_{k_0})\hc$ for a number of $r(O)$ non-zero complex coefficients given by $O_{\lambda_0,\alpha_0,k_0} =  \Tr(O T^{\lambda_0,\alpha_0}_{k_0} ) \neq 0 $.
We refer to $r(O)$ as the \emph{rank of operator $O$} with respect to the ITO basis. 

In the Heisenberg picture, the evolved operator will be given by
\begin{equation}
	C\hc(\g) O C(\g) = \sum_{ (\boldsymbol{\lambda,k})}  O_{\lambda_0,k_0} v^{\boldsymbol{\lambda}}_{\bf{k'},\bf{k}} (\g^{-1})   w^{\boldsymbol{\lambda}, \lambda_D}_{{\bf{k'},\bf{k}}, k_D}  T^{\lambda_D}_{k_D},
\end{equation}
where each term corresponds to the path of basis operators $T^{\lambda_0,\alpha_0}_{k_0} \rightarrow T^{\lambda_0,\alpha_0}_{k_0'} \rightarrow T^{\lambda_1,\alpha_1}_{k_1} ... \rightarrow T^{\lambda_{D-1}, \alpha_{D-1}}_{k_{D-1}'} \rightarrow T^{\lambda_D, \alpha_D}_{k_D} =: (\boldsymbol{\lambda}, \boldsymbol{k}, \boldsymbol{k}')$.  We also used a shorthand notation for the products of irreducible matrix coefficients given by $v^{\boldsymbol{\lambda}}_{\boldsymbol{k'}, \boldsymbol{k}} (\g^{-1}) :=  v^{\lambda_0}_{k_0',k_0}(g_1^{-1}) v^{\lambda_1}_{k_1',k_1}(g_2^{-1}) ... v^{\lambda_{D-1}} _{k_{D-1}', k_{D-1}}(g_D^{-1}) $  and where $w^{\boldsymbol{\lambda}, \lambda_D}_{\boldsymbol{k'},\boldsymbol{k}, k_D} :=  w^{(1)} _{(\lambda_0,k_0'), (\lambda_1, k_1)} ... w^{(D)} _{(\lambda_{D-1},k_{D-1}'), (\lambda_D, k_D)} $  correspond to products of coefficients for transfer matrices involving only the fixed unitaries $W$. 

 We can equivalently see this decomposition in the Liouville matrix representation, where it arises by introducing a resolution of identity  $I = \sum_{{\lambda,\alpha, k}} |T^{\lambda,\alpha}_{k} \> \> \< \< T^{\lambda,\alpha}_k|$ after every unitary channel to get
\begin{align}
	C(\g) \otimes C^{*}(\g) =   \sum_{ (\boldsymbol{\lambda,k,k'})}  v^{\boldsymbol{\lambda}}_{\bf{k'},\bf{k}} (\g^{-1})   w^{\boldsymbol{\lambda}, \lambda_D}_{{\bf{k'},\bf{k}}, k_D}  |T^{\lambda_0}_{k_0}\>\>\<T^{\lambda_D}_{k_D}|.
	\label{eq:pathdecomposition}
\end{align}

		Details of these decompositions are expanded on in Appendix ~\ref{app:pathdecomp}. In particular, for an input state $|\bf{0}\>$ and its vectorisation $|\bf{0}\>\>$ the expectation value of $O$ is 
		\begin{equation}
			\<\<O|C(\g)\otimes C^{*}(\g) |{\bf{0}}\>\> = \<{\bf{0}}|C\hc(\g) O C(\g) |{\bf{0}}\> .
			\label{eq:expectation}
		\end{equation}

\subsection{ Harmonic analysis of quantum circuits}

We describe how the structured unitaries in the class $\mathfrak{C}_n$ can be analysed using general harmonic analysis on compact or finite groups (or homogeneous spaces) and show this naturally relates to the symmetry-adapted path decomposition. For a full review of the topic we refer to \cite{goodman2009symmetry} and some preliminaries in Appendix \ref{app:ITO}.

Given a complex-valued function on a compact (or finite) group $\mathbb{G}$ as $f: \mathbb{G} \rightarrow \mathbb{C}$,  its Fourier transform will map irreducible representations to complex matrices 
\begin{equation}
	\hat{f}(\boldsymbol{\lambda}) = \int  f(\g) \boldsymbol{\lambda}^{*}(\g)  d\, \g,
\end{equation}
where the integration is with respect to the Haar measure on $\mathbb{G}$, $\boldsymbol{\lambda}:\mathbb{G}\rightarrow GL(d_{\boldsymbol{\lambda}})$ denotes an irreducible (unitary) representation of $\mathbb{G}$ with dimension $d_{\boldsymbol{\lambda}}$ and where $\boldsymbol{\lambda}^{*}$ is its conjugate matrix. In the finite group case the integration is replaced by an average over the entire group (i.e $\frac{1}{|G|} \sum_{\g\in G}$).  Notably, Fourier coefficients $\hat{f}(\boldsymbol{\lambda})$ carry a basis dependency that arises from the choice of basis in the carrier space of the irrep.

For circuits $\mathfrak{C}_n$ the group is given by $\mathbb{G} = G^{\times D}$,  a direct product of groups and therefore all of its (complex) irreps arise as tensor products of the irreps for the underlying group $G$.  In other words $\boldsymbol{\lambda} = \lambda_0 \otimes... \otimes \lambda_{D-1}$ for $\lambda_i$ irreps of $G$. 
Given an observable $O$, its expectation value relative to the state $C(\g) | {\bf{0}}\>$ can be viewed as a function on the  group  $o : \mathbb{G} \longrightarrow \mathbb{C}$ with $o(\g^{-1}) := \<{\bf{0}} |  C\hc(\g) O C(\g) |{\bf{0}} \>$. The Fourier inversion formula ensures that we have the following Fourier series decomposition  of $o$
\begin{equation}
	o(\g) =  \sum_{\boldsymbol{\lambda}} d_{\boldsymbol{\lambda}} \Tr(\hat{o}(\boldsymbol{\lambda})  \, \boldsymbol{\lambda}^{*}(\g^{-1})),
\end{equation} 
where the summation is over all inequivalent irreps of $\mathbb{G}$.

A general form of Plancharel's identity also holds for the inner product  $\<o, \tilde{o}\>_{L^2(\mathbb{G})} := \int o(\g) \tilde{o}^{*}(\g) d\, \g$ on $L^2(\mathbb{G})$.  This implies a generalised Parseval's identity
\begin{equation}
	||o||_{L^2(\mathbb{G})}^2  = \sum_{\boldsymbol{\lambda}} d_{\boldsymbol{\lambda}}  || \hat{o}({\boldsymbol{\lambda}}) ||_{HS}^2,
\end{equation}
where the Hilbert-Schmidt norm is $||\hat{o}(\boldsymbol{\lambda}) ||_{HS}^2 = \tr(\hat{o}\hc(\boldsymbol{\lambda}) \hat{o}(\boldsymbol{\lambda}))$.

To connect with the path-decomposition we observe that if we view $\hat{o}(\boldsymbol{\lambda})$ as a matrix in $GL(d_{\boldsymbol{\lambda}})$ with entries given by $\hat{o} (\boldsymbol{\lambda}) _{\bf{k',k}}$ then 
\begin{align}
o(\g) = \sum_{\boldsymbol{\lambda}} d_{\boldsymbol{\lambda}}  \sum_{{\bf{k}', \bf{k}}}[\hat{o}(\boldsymbol{\lambda})] _{\bf{k',k}}  \left(v^{\boldsymbol{\lambda}}_{{\bf{k,k'}}}({\bf{g}}^{-1})\right)^{*}.
\end{align}
 This assumes that the Fourier transform is taken with respect to the same basis for the carrier space that results in the $v^{\lambda}_{kk'}$ matrix coefficients. If we substitute Eqn.~\ref{eq:expectation} and Eqn.~\ref{eq:pathdecomposition} into the Fourier transform of $\hat{o}$ we obtain that the matrix of Fourier coefficients at the irrrep $\boldsymbol{\lambda}$ is given in terms of
\begin{equation}
	[\hat{o}(\boldsymbol{\lambda})]_{\bf{k',k}} = \frac{1}{d_{\boldsymbol{\lambda}}}  O_{\lambda_0,k_0}  \sum_{\lambda_D, k_D}   w^{\boldsymbol{\lambda}, \lambda_D}_{{\bf{k'},\bf{k},} k_D}  \<{\bf{0}}|T^{\lambda_D}_{k_D}|{\bf{0}}\>.
\end{equation}
Note that this follows from Schur orthogonality of irreducible matrix coefficients, which is also captured in a sense by Parseval's theorem and the fact that the irreducible matrix coefficients form an orthogonal basis for complex-valued functions on a group.

Similarly, one can define Fourier transforms for matrix-valued functions component-wise for each matrix entry.  This allows to consider operators of the form $C(\g)\otimes C^{*}(\g)$ that represent the unitary in the Liouville matrix representation. In that way, one can obtain similar connection with the path decompositions without including the path contraction with the target observable or the input state (further details and examples in Appendix~\ref{app:FTgroup} ). This allows to analyse the entire family of unitaries $\mathfrak{C}_n$ in terms of the Fourier transform 
\begin{equation}
	\widehat{C\otimes C^{*}} (\boldsymbol{\lambda}) = \int C(\g) \otimes C^{*}(\g) \otimes \boldsymbol{\lambda}^{*}(\g)  d\,\g
\end{equation}

which is an operator on  $\widehat{C\otimes C^{*}} (\boldsymbol{\lambda}) \in  \B(\h_n\otimes \h_n^{*}) \otimes GL(d_{\boldsymbol{\lambda}})$. This definition and more broadly the harmonic analysis of quantum circuits $C(\g)$ can be applied to any circuit structure with a group dependency, including the situation where unitaries $U_i$ at different positions are parametrised by the same group parameters. To evaluate the Fourier coefficients in such a case where the group parameters in different circuit layers are not independent requires to apply higher-moment integration tools such as Wiengarten calculus \cite{collins2022weingarten}.

However, for  $\mathfrak{C}_n$ as the group elements are independently chosen for the different unitaries $U_i$ then performing this integration in the Fourier transform is straightforward and can be related to the symmetry-adapted path decomposition, provided we are consistent in the choice of basis for carrier space of irreps that determine the irreducible matrix coefficients.  Specifically,  we can independently perform the integration in-situ where 
\begin{equation}
	\Pi_{\lambda_i} = \int U_{i+1}(g_{i+1}) \otimes U^{*}_{i+1}(g_{i+1}) \otimes \lambda_i^{*}(g_{i+1}) d\, g_{i+1}
\end{equation}
denotes a projector onto the $\lambda_i$ isotypical component of $U_{i+1}\otimes U^{*}_{i+1}$. Notably, if the $\lambda_i$ irrep is not found in the decomposition of $U_{i+1}\otimes U_{i+1}^{*}$ then the above quantity is zero.  In terms of the symmetry-adapted basis we have that
\begin{align}
	\Pi_{\lambda_i} = \sum_{\alpha_i, k,k'} |T^{\lambda_i,\alpha_i}_{k}\>\>\<\<  T^{\lambda_i,\alpha_i}_{k'}| \otimes |k\>\<k'|,
\end{align}
where we recall that the auxiliary system represents $GL(d_{\lambda_i})$ and  $|k\>$ labels the basis elements of the carrier irrep. 
The Fourier transform can then be viewed compactly as 
\begin{align}
	\widehat{C\otimes C^{*}} (\boldsymbol{\lambda}) = \Pi_{\lambda_0} W_1\otimes W_1^{*} ...  \Pi_{\lambda_{D-1}} W_D\otimes W_D^{*}.
\end{align}

Then we get that the Fourier coefficients, by which we mean the components of the Fourier transform at a specific irrep $\boldsymbol{\lambda}$ of $\mathbb{G}$ are
\begin{align}
	[\widehat{C\otimes C^{*}} (\boldsymbol{\lambda}) ]_{{\bf{k',k}}} = \frac{1}{d_{\boldsymbol{\lambda}}} \sum_{\lambda_D,k_D}  w^{\boldsymbol{\lambda}, \lambda_D}_{\bf{k'},\bf{k}, k_D}  |T^{\lambda_0}_{k_0}\>\> \<\< T^{\lambda_D}_{k_D}| .
\end{align}
Contracting the Fourier coefficients against an input state $|\boldsymbol{0}\>\>$ and output observable $\<\<O|$ we obtain precisely $[\hat{o}(\boldsymbol{\lambda})]_{\bf{k',k}}$ as before.
Finally, from the Fourier inversion theorem we have the Fourier series decomposition
\begin{align}
	C(\g)\otimes C^{*}(\g) = \sum_{\boldsymbol{\lambda}} d_{\boldsymbol{\lambda}} \tr_{2}[\widehat{C\otimes C^{*}} (\boldsymbol{\lambda})  \mathbb{I} \otimes \boldsymbol{\lambda}^{*} (\g^{-1})],
\end{align}
where the trace is taken over the second system representing the auxiliary carrier space of the representation $\boldsymbol{\lambda}$ whose basis elements are $|{\bf{k}}\> = |k_{0},..,k_{{D-1}}\>$. 

This connection that path decomposition can be reformulated as a Fourier decomposition is fairly natural considering that Peter-Weyl theorem ensures that the set all matrix coefficients for all irreducible representations form an orthonormal basis of the space $L^2(G)$ of complex-valued functions on the group $G$. 

\subsection{Operator truncated classical simulations} 

 An operator truncated classical simulation will involve the evaluation of only a restricted set of Fourier coefficients using the connection with the multi-path basis decomposition. The harmonic analysis of circuits we discussed in the previous section gives us a principled way to construct such a truncated set $\T$.   We will discuss two computational tasks for a given class of circuits $\mathfrak{C}_n$. 
 
 \emph{Mean value simulation} For a given observable $O$ and a truncation set $\T$ consisting of a restricted set of irreps $\boldsymbol{\lambda} = (\lambda_0,...\lambda_{D-1})$ in the Fourier decomposition of $o$, then applying the inverse Fourier transform will produce  a truncated series
 \begin{align}
 	o^{\T} (\g) = \sum_{\boldsymbol{\lambda}\in \T}  d_{\boldsymbol{\lambda}} \tr(\hat{o}(\boldsymbol{\lambda}) {\boldsymbol{\lambda}^{*}}(\g^{-1})).
 \end{align}

 \emph{Strong simulation} Given an input state $|{\bf{0}}\>$ and a truncation set $\T$ then the truncated Fourier series for the probability distribution of obtaining outcome ${\bf{x}}$ will be given  by
 \begin{align}
p^{\T} ( {\bf{x}}, \g)  = \sum_{\boldsymbol{\lambda}\in \T} d_{\boldsymbol{\lambda}} \tr [ \<\< {\bf{x}}| \widehat{C\otimes C^{*}} (\boldsymbol{\lambda})|{\bf{0}}\>\>   \boldsymbol{\lambda}^{*} (\g^{-1})].
\end{align}
In this case, the simulation requires to represent the operator $\widehat{C\otimes C^{*} }(\boldsymbol{\lambda})$ from which the truncated distributions can be obtained for \emph{all} the outcomes. 

Different truncation strategies will influence both the complexity of evaluating the Fourier coefficients for the set $\T$ and the approximation error between  $o$ and $o^{\T}$ or $p$ and $p^{\T}$. Generally, we select a number of irreps of  $\mathbb{G}$, $\boldsymbol{\lambda} = (\lambda_0,..., \lambda_{D-1})$ that scale polynomially with the number of qubits so that  $|\T| = poly(n)$ and for which the Fourier coefficients are non-zero.  We then construct algorithms to evaluate the  truncated series either exactly or approximately leading to two potential strategies

1. {\bf{Truncated Fourier series}}  - where for all  $\boldsymbol{\lambda}\in \T$ the Fourier coefficients  $[\widehat{C\otimes C^{*}}(\boldsymbol{\lambda})]_{\bf{k},{\bf{k'}}}$ (or  $[\hat{o}(\boldsymbol{\lambda})]_{\bf{k},{\bf{k'}}}$ ) can be efficiently represented  and computed in $poly(n)$ time and memory. This requires the dimensions of the irreps $\lambda_i$ to scale at most polynomially such that $d_{\boldsymbol{\lambda}} = poly(n)$.

2. {\bf{Truncated Fourier series with low-rank approximation coefficients}} - where for all $\boldsymbol{\lambda}\in \T$ the Fourier coefficients  $[\widehat{C\otimes C^{*}}(\boldsymbol{\lambda})]_{\bf{k},{\bf{k'}}}$ (or  $[\hat{o}(\boldsymbol{\lambda})]_{\bf{k},{\bf{k'}}}$ ) can be efficiently represented in $poly(n)$ time by a low-rank approximation that can be efficiently computed in $poly(n) $ time.

Importantly, we emphasise again that the Fourier coefficients depend on the choice of basis for the carrier space of the irreducible representation $\boldsymbol{\lambda}$ but the terms in the Fourier series do not.  Furthermore, we have seen how the Fourier coefficients can be expressed in terms of multi-paths decompositions in a symmetry adapted operator basis.  The choice of that operator basis $\{T^{\lambda_i}_{k_i}\}$ must be consistent with the choice of basis for the irrep  which gives the irrep's matrix coefficients $v^{\lambda}_{kk'}(\cdot)$.

In the remainder of this section we will discuss in Sec.~\ref{sec:noiselessconverge} conditions on selecting the truncation set $\T$ so that it ensures convergence of $o^\T$ to $o$. Then we consider \emph{sparse Fourier series} - where the truncation set $\T$ contains irreps $\boldsymbol{\lambda}  = (\lambda_0,..., \lambda_{D-1})$ for which at most a number of $L$ of the $\lambda_i$ are non-trivial irreps.  In this case, we construct algorithms in Sec.~\ref{sec:sparsealg} to evaluate these sparse Fourier coefficients for $\hat{o}(\boldsymbol{\lambda})$ from a path decomposition using tree structures with $L$ levels that start at basis elements $T^{\lambda_0}_{k_0}$ that have non-zero overlap with the observable $O$.  Then we extend this construction in Sec.~\ref{sec:bidirectional}  to obtain $\widehat{C\otimes C^{*}}(\boldsymbol{\lambda})$.

\subsubsection{Convergence of truncated Fourier expansions}
\label{sec:noiselessconverge}
For now, we discuss under what conditions we could construct sequences of increasingly larger truncation sets in a way that ensures convergence of $o^{\T}(\g)$ to $o(\g)$.  While point-wise convergence remains challenging, we can show  $L^2(\mathbb{G})$ approximation error bounds whenever the expectation value $o(\g) \in L^2(\mathbb{G})$ satisfies the following additional Lipschitz condition.  Given a connected Lie group $G$ we say $o$ is $\alpha$-Lipschitz ($0<\alpha \leq 1$) if  as a function of  $\x\in \mathbb{G}$, we have $ ||o(\g\cdot \x) - o(\x)||_{L^2(G)} = \O(|h|^{\alpha})$ as $|h|\rightarrow 0$, where $|h|$ is the geodesic distance on $G$ from the identity element. 
If this condition is satisfied then we can employ several of the results in \cite{daher2019titchmarsh} to determine a truncation set $\T$ that allows us to control the approximation error. In particular Theorem 3.5 from \cite{daher2019titchmarsh} ensures that 
\begin{align}
	\sum_{\boldsymbol{\lambda},  e(\boldsymbol{\lambda}) \geq N} d_{\boldsymbol{\lambda}}  ||\hat{o}(\boldsymbol{\lambda})||^2_{HS}  = \O(N^{-2\alpha})  \hspace{0.5cm} {\rm{as}} \hspace{0.5cm} N\rightarrow \infty
\end{align}
where in the above $e(\boldsymbol{\lambda})$ is related to the eigenvalue of the Casimir operator (or minus of the Laplacian $\mathcal{L}_{\mathbb{G}}$) for the irrep given by $\boldsymbol{\lambda}$ and satisfying $-\mathcal{L}_{\mathbb{G}} \boldsymbol{\lambda}(\g)=  (e(\boldsymbol{\lambda})^2-1) \boldsymbol{\lambda}(\g)$. Furthermore, as inequivalent irreps $\boldsymbol{\lambda}$ can be uniquely labelled by using highest weights then  the corresponding Laplace eigenvalue will be determined by a quadratic form on the space of weights.  In our situation, these series of results allow us to construct a truncation set $\T_{\alpha-Lip, N} = \{ \boldsymbol{\lambda} : e(\boldsymbol{\lambda})\leq N\}$  such that using Parseval's theorem and the above convergence rate result we have the following $L^2$ approximation error
\begin{align}
	||o(\g) - o^{\T_{\alpha-Lip,N}}(\g)||^{2}_{L^2(\mathbb{G})} = 	
	\O(N^{-2\alpha}) 
\end{align}

The eigenvalue of the Casimir operator for the trivial irrep  is zero. Furthermore, as we consider direct product groups here then we would have that $e(\boldsymbol{\lambda})^2 -1  =  c(\lambda_0) + ... + c(\lambda_{D-1})$, where $c(\lambda_i)$ is the eigenvalue of the Casimir operator for the irrep $\lambda_i$. (For example, in $SU(2)$ the Casimir operator is half of the total angular momentum operator $ \bf{J}^2/2$ that has eigenvalue $\frac{j(j+1)}{2}$ for ireps of dimension $2j+1$). From this we can see that the truncation set will include what we refer to as `sparse' elements $\boldsymbol{\lambda}$ where at most $N$ irreps $\lambda_0,..., \lambda_{D-1}$ will be non-trivial. We refer to this as a sparse representation and discuss in the next section an algorithm to efficiently evaluate the function $o^{\T}(\g)$, provided a series of assumptions on the circuit structure are met. In an informal way, this analysis can be summarised as
\emph{ If the expectation value over group parameters of the circuit ensemble is $\alpha$-Lipschitz then there is a truncation in terms of a sparse Fourier decomposition with guaranteed convergence in terms of the mean square error over the parameter space. 
}
Finally, a full classification of observables and circuit ensembles for which the Lipschitz property is satisfied will enable classical simulations of noiseless circuits with provably converged average approximation error over the parameter space.

\subsubsection{Algorithm for sparse truncated Fourier series} 

\label{sec:sparsealg}
Here we construct an algorithm that uses the multi-path decomposition to evaluate a sparse, truncated Fourier series of an observable that contains all the terms corresponding to irreps  $\boldsymbol{\lambda} = (\lambda_0,...,\lambda_{D-1})$ where at most a number of $L$ locations have non-trivial irreps $\lambda_i \neq 0$. Formally, we denote this truncation rule as $\texttt{truncate}[(\boldsymbol{\lambda},{\bf{k}, k'}), L] = True$  if and  only if 
$\#|\{ \lambda_i \neq 0,  i\in \{0,1,...D-1\} \}| \leq L $. In this case, given a particular $\boldsymbol{\lambda}$ in the truncation set, then the paths $(\boldsymbol{\lambda}, {\bf{k,k'}})$ that we consider will include all possible components ${\bf{k,k'}}$ (forming a number of at most  $d_{\boldsymbol{\lambda}}^2$ paths). 
We denote the set of all valid paths $(\boldsymbol{\lambda, {\bf{k}}, {\bf{k}'}})$ that satisfy a given truncation rule with a  truncation parameter $L$ by 
\begin{align}
	\T &= \{ (\boldsymbol{\lambda}, {\bf{k,k'}}) : \texttt{validate}(\boldsymbol{\lambda, {\bf{k}}, {\bf{k}'}}) = True \  \    \nonumber \\  \&  & \  \texttt{truncate}[(\boldsymbol{\lambda},{\bf{k,k'}}), L] = True\}. 
\end{align}
The $\texttt{validate}$ function ensures that only paths with non-zero contributions are included.  Notably, when talking about paths some choice of operator basis has already been made to decompose the circuit as described in Sec.~\ref{sec:multipath}. This will depend on the circuit ensemble's structure and requires an efficient labelling scheme to represent the basis $\{T^{\lambda}_{k} \}$, which we discuss in more detail at the end of the section.  Furthermore, the following assumption ensures that we can efficiently compute the action of the $W$-type operators in this basis.  Its purpose is to explicitly state the conditions required to select the operator basis and/or construct circuit ensembles that will ensure a sparse Fourier series approximation of observables is efficient to compute (and later on, of probability distributions).

{\bf{\emph{Assumption 1} }}{\bf{(Controlled tree growth) }} \emph{For all $i\in \{1, D\}$, each unitary $W_i$ maps any given basis operator $|T^{\lambda}_{k}\>\>$ entirely into the trivial subspace or is orthogonal to the trivial subspace.  Equivalently, 
\begin{equation}
   \Pi_0 W_{i}\otimes W_i^{*} |T^{\lambda}_k \>\> = \begin{cases}
		0 &  \ \  or \\
		W_{i} \otimes W_{i} ^{*}|T^{\lambda}_k \>\> .
	\end{cases}
\end{equation}
{\bf{(Path contractions and evaluation) }}The number of basis operators $(\mu, m)$ for which the contraction is non-zero $\<\< T^{\mu}_m|W_{i}\otimes W_{i}^{*} ... W_{i+J}\otimes W_{i+J}^{*}|T^{\lambda}_k\>\> \neq 0$ is at most $s = poly(n)$ and each term can be evaluated efficiently in time at most $cost(\texttt{\textup{contract}})$.}

\emph{If we require to construct the tree by propagating the basis in the Heisenberg picture the above holds for the adjoint matrices $W_i\hc$ too. 
}

We will see that this assumption is a-priori met for large classes of unitaries, so in certain situations it will not impose any additional constraints on the $W$'s.

\emph{Example:} \emph{For particular families $\mathfrak{C}_n$ the complexity growth due to the $W$'s  can be easily kept under control, and Assumption 1 always holds. One such situation is when the fixed unitaries $W_j$ are permutations of the basis elements so that $s=1$. In such a case, the unitaries $W_j$'s impose "time-local" constraints on the valid paths so that $\tt{validate}(\boldsymbol{\lambda, {\bf{k,k'}}}) = True$ iff $W_j\hc T^{\lambda_{j-1}}_{k'_{j-1}} W_j = T^{\lambda_j}_{k_j}$ for all $j$,  which is determined by a function $w_j$ permuting the basis labels $(\lambda_j,k_j) = w_j(\lambda_{j-1},k'_{j-1})$. For all valid paths we also have $[\widehat{C\otimes C^{*}} (\boldsymbol{\lambda}) ]_{{\bf{k',k}}} = \frac{1}{d_{\boldsymbol{\lambda}}} |T^{\lambda_D}_{k_D}\>\> \<\< T^{\lambda_0}_{k_0}| $ and zero otherwise. Whilst a restricted class, nonetheless it can still apply to many interesting circuit classes, and includes the situation when the basis is given by the $n$-qubit Pauli operators and the $W_j$'s are any Clifford operations. In this case $cost(\texttt{\textup{contract}})  = O(n^2)$.
}

\SetKwComment{Comment}{$\triangleright$ }{ }

\begin{algorithm}[h!]
	\caption{\small{Generate Truncated Tree \texttt{GTT}. Truncated Fourier approximation of the observable is obtained from evaluation of paths in the generated trees.   \label{alg:two}}} 
	\hrule

	\hrule\hrule
	\vspace{0.1cm}
	{\bf{Truncated Fourier approximation of observables}}
	\vspace{0.1cm}
	\hrule\hrule
	\vspace{0.1cm}
	\KwIn{$O$,  sequence of unitaries  $U_1(g_1),W_1, ... W_{D}$, truncation parameter $L$}
	$o^\T(\g) = 0$\;
	$roots =  \{T^{\lambda_0}_{k_0} :  O_{\lambda_0,k_0}= \<\<O|T^{\lambda_0}_{k_0}\>\> \neq 0 \}$\;
	
	\For{$T^{\lambda_0}_{k_0} \in roots$}{
		\For{$path \in \texttt{\textup{GTT}}(L, T^{\lambda_0}_{k_0} , W_1,...,W_D)$}{
			$o^\T(\g) += O_{\lambda_0,k_0}\texttt{evaluate}(path)$ }
	}
	
	\KwOut{$o^\T(\g)$}
	\vspace{0.1cm}
	\hrule \hrule
	\vspace{0.1cm}
	{\bf{Generate Truncated Tree \texttt{GTT}}}
	\vspace{0.1cm}
	\hrule
	\hrule
	\vspace{0.1cm}
	\SetKwFunction{FTree}{GTT}
	\SetKwProg{Pn}{Function}{:}{\KwRet{ $ tree$} }
	\Pn{\FTree{$level, root, W_{start},..., W_D$}}{
		$tree  = \{nodes = (root,start), edges = \{\}\}$\;
		\If{$level > 1$}{
			$	children,edge\_data = 	 \texttt{generate-children}(root,start)$\;
			$   level = level-1$\;
			$tree.append(children, edge\_data )$\;
			\For{$child, index  \in children$}{
				$tree.append(\texttt{GTT}(level, child, W_{index}, ... W_D)$)
			}
		}
		\Else{\KwRet{ tree}}
		
	}	
	\KwRet{tree}

	\vspace{0.1cm}
	\hrule
	
	\vspace{0.1 cm}
	\SetKwFunction{Fchild}{generate-children}
	\SetKwProg{Pn}{Function}{:}{\KwRet{ $children, edge_{data}$} }
	\Pn{\Fchild{$root,start$}}{
		$j=0$\;
		\While{$ \Pi_0 \W_{start+j} ... \Pi_0 \W_{start}|root\>\> \neq 0$ \  \hspace{1cm} {\bf{and}} $start+j<D$ }{$j=:j+1$}

		\For{$child, value\in $\texttt{\textup{contract}}$[\<\<root| \W_{start} ... \W_{start+j}]$}{
			\If{$start+j<D$}{      
				\For{$child_{E} \in \texttt{\textup{expand}}(child)$}{
					$children +=:\{(child_{E}, start+j+1)  \}$\;
					$edge_{data} +=:\{ (root, child_{E}) :  value\times \texttt{matrix-coeff}(child, child_{E}, start+j ) \}$
				}
			}
			\If{$start+j =D$ }{
				$children+ =:\{child, None\}; \ edge_{data}+ =:\{(root, child):value\}$
			}
		}
	}\KwRet{ $children, edge_{data}$} 
	\vspace{0.2cm}
	
	\SetKwFunction{Fcontract}{expand}
	\SetKwProg{Fn}{Function}{:}{\KwRet{$\{ T^{\lambda}_{k'} :  k' \in \{1,..., dim(\lambda)\} \}$}}
	\Fn{\Fcontract{$T^{\lambda}_k$}}{
	}
	\KwRet{$\{ T^{\lambda}_{k'} :  k' \in \{1,..., dim(\lambda)\} \}$}
	\vspace{0.2cm}
	
	\SetKwFunction{Fcontract}{contract}
	\SetKwProg{Fn}{Function}{:}{\KwRet{$\{ (|T^{\lambda}_k\>\>, value) \, {\rm{s.t}} \,  \, value=\<\< T^{\lambda}_k|\W_{start+j}  ... \Pi_0 \W_{start} |root\>\> \neq 0 \}$}}
	\Fn{\Fcontract{$\<\<root| \W_{start} ...  \W_{start+j} $}}{
	}
	\KwRet{$\{ (\<\<T^{\lambda}_k|, value) \, {\rm{s.t}} \,  \, value=\<\< root|\W_{start} \W_{start+1} ...  \W_{start+j} |T^{\lambda}_k\>\> \neq 0 \}$}
	\vspace{0.2cm}
	
	\SetKwFunction{Fmatrix}{matrix-coeff}
	\SetKwProg{Fn}{Function}{:}{\KwRet{$v^{\lambda}_{kk'}(g_{index})$}}
	\Fn{\Fmatrix{$T^{\lambda}_k, T^{\lambda}_{k'}, index$}}{
	}
	\KwRet{$v^{\lambda}_{kk'}(g_{index})$}
	\vspace{0.2cm}
	
	\SetKwFunction{Fmatrix}{evaluate}
	\SetKwProg{Fn}{Function}{:}{\KwRet{$ $}}
	\Fn{\Fmatrix{$path, state = {|\bf{0}\>}$} }{
		$nodes, edges  =: path $\;
		\Comment{Get last node in path }
		$(root\_last, index\_last) =: nodes[last]$\;
		$overlap = \<\< root\_last | \W_{index\_last} \Pi_0 ... \Pi_0 \W_D | {\bf{0}}\>\> $\;
	}
	\KwRet{$overlap\prod_{i=1}^{last-1} edges[(nodes[i], nodes[i+1])]$}

\vspace{0.1cm}
\hrule
\vspace{0.1cm}

\end{algorithm}

The Algorithm  \ref{alg:two} describes how to construct the truncated, sparse Fourier series of an observable $O$ corresponding to the truncation set $\T$. It assumes an operator basis has been fixed and that $O$ decomposes with respect to this basis into at most $r(O) =poly(n)$ terms. The main construction is given by the Generate Truncated Tree $\texttt{GTT}$ function that takes as input a truncation parameter, an operator $|T^{\lambda_i}_{k_i}\>\>$, an ordered sequence of unitaries  $W_i,...,W_{i+J}$. It returns a tree where nodes hold operator basis elements (or their labels), and a position label (e.g telling us which unitary we apply next). The edges contain the values of matrix coefficients for the unitaries applied with respect to the fixed basis.  This tree structure should be viewed as representing the path expansion of 
$ U^{\hc}_{i+J}(g_{i+J})W^{\hc}_{i+J} ... U^{\hc}_i(g_i)W^{\hc}_i  T^{\lambda_i}_{k_i}W_{i} U_i(g_i) ... W_{i+J} U_{i+J}(g_{i+J})$ truncated to include all of the Fourier terms $\boldsymbol{\lambda} = (\lambda_0,...\lambda_{D-1})$ that contain non-trivial irreps in at most $L$ locations. The number of paths in the tree are going to be at most $(sd_{max})^L$, where $d_{max}$ is the maximal dimension of irreps $\lambda_i$ at any location. 

Applying a $W$ unitary expands each node to at most $s$ other nodes since $W_i T^{\lambda}_k W_i\hc$  involves a linear combination of at most $s$ orthonormal basis operators. Due to \emph{Assumption} 1, all of these $s$ nodes either live entirely into the trivial space $\Pi_0$ projects into or its orthogonal complement. The importance of this is that in the former case applying a group-parametrised unitary is a trivial step since $U(g) \otimes U^{*}(g) \Pi_0 = \Pi_0$ and the tree does not branch out. In the latter case applying the unitary $U$ to each incoming node $T^{\lambda}_k$ expands into $d_{\lambda}$ nodes corresponding to every component $k$. These edges will carry the value of irreducible matrix coefficient $v^{\lambda}_{kk'}(g)$ and values of the contractions due to products of $W$-type unitaries. In this situation, the truncation parameter increases as we want to penalise the growth in complexity due to the unitaries $U$. 

Along each branch there will be at most $L$ contractions and the function $\texttt{contract}$ determines the nodes and produces the values of those contractions between the input nodes and output nodes. Therefore evaluating each path will depend on the complexity of these two sub-functions and overall it takes $L \times cost(\texttt{contract})\times cost(\texttt{matrix-coeff})$ for each path. We will generally denote by $cost(\texttt{evaluate})$ to capture this time complexity to evaluate each path, which may also include the cost of computing the final overlap with the computational basis state as described by the function $\texttt{evaluate}$.  This leads to an overall time complexity of $r(O)(sd_{max})^L cost(\texttt{evaluate})$ to compute a truncated Fourier approximation of the observable $O$.  Notably the \texttt{matrix-coeff} function may be used \emph{parametrically} so that once the truncated trees are generated, then the expectation value for \emph{different} circuits in the ensemble can be determined by changing parameters for the group elements $g_i$ in the path evaluation, without incurring cost overhead.  

The complexity cost of determining matrix coefficients for the irreducible representations comes from evaluating $v^{\lambda}_{kk'}$ at different group parameters $g\in G$. We would only need to determine these for the irreps $\lambda$ of $G$ that appear in the truncation set at any of the locations labelled by $i$.  Many such explicit constructions have been developed, particularly for the case of unitary groups and from special functions and homogenous polynomials \cite{klimyk1995representations}. For example, irreducible representations of complex general linear groups $GL_m$ with respect to a Gelfand-Tsetlin basis for the carrier representation space can be evaluated with fast classical algorithms of $O(m^2 d_\lambda)$, where $d_{\lambda}$ is the dimension of the irrep \cite{burgisser2000computational}. In general, the labelling scheme used can have an impact on the memory requirements for storing the tree data structures. For our symmetry-adapted operator basis we can use a labelling scheme for the irrep alone,  but we will need additional labelling for different multiplicities. In general, $\O(nD)$ bits of memory will suffice to uniquely represent the basis at all locations, however to keep track of the basis transformation under the group action we will employ a labelling scheme for the poly-sized irreducible representations as used for instance by computational group theory packages such as GAP or LieArt that can evaluate matrix irreducible coefficients for various groups.

\begin{figure*}[t!]
	\includegraphics[width=\textwidth]{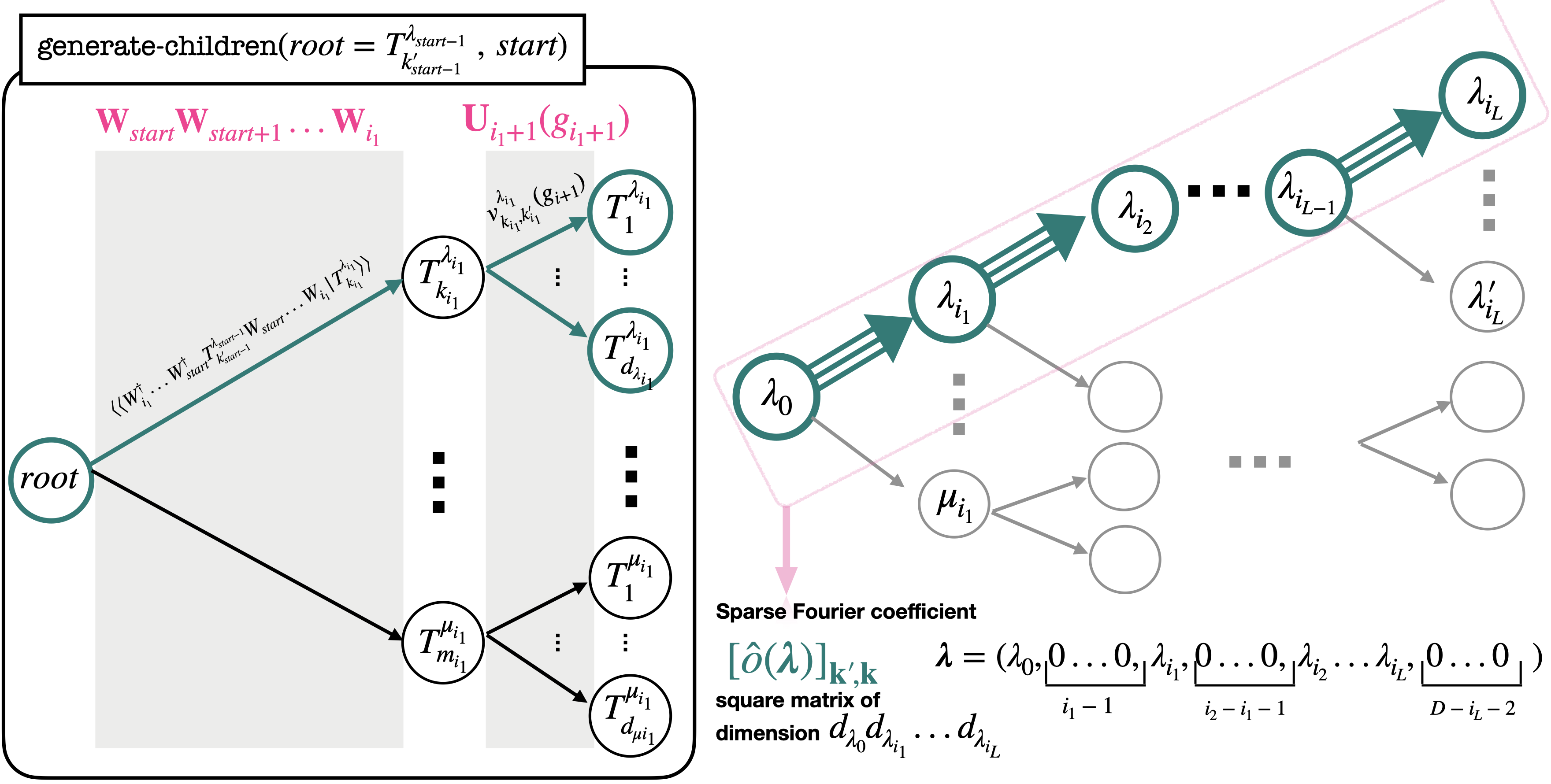}
	\label{paths}
	\caption{Visualisation of the main algorithm to evaluate the sparse Fourier coefficients from truncated path decompositions with respect to a symmetry-adapted operator basis (i.e $\{T^{\lambda}_{k}\}_{k,\lambda}$). \emph{(Left) }The function \texttt{generate-children} is the main rule used to recursively generate trees with a total of $L$ levels via $\texttt{GTT}$. The nodes contain information about the basis element (i.e $root$) and the index $start$ which labels the next unitary to be applied. Edges carry the overlap value (inner product) between the output node basis and the input node basis evolved through a sequence of one or more unitaries. \emph{ (Right) } Trees with $L$ levels generated by $\texttt{GTT}$ starting from fixed basis operators $\{T^{\lambda_0}_{k_0}\}_{k_0=1}^{dim(\lambda_0)} $ for an irreducible representation $\lambda_0$ will contain subtrees that correspond to the (sparse) Fourier coefficients evaluated at $\boldsymbol{\lambda} = (\lambda_0, \lambda_1, ..., \lambda_{D-1})$ where at most $L$ irreps $\lambda_i$ are non-trivial.  } 
\end{figure*}

 \subsubsection{Bi-directional trees for truncated probability distributions}

 \begin{algorithm}[h!]
 	\caption{\small{Algorithm to output truncated Fourier approximation of  probability distributions using bi-directional trees constructed via \texttt{GTT}.  The truncated trees are independent of the outcome ${\bf{z}}$, which is only used for the path evaluation via \texttt{evaluate}.  }}
 	\hrule
 	
 	\hrule\hrule
 	\vspace{0.1cm}
 	{\bf{Truncated Fourier approximation of output probability distributions}}
 	\vspace{0.1cm}
 	\hrule\hrule
 	\vspace{0.1cm}
 	\KwIn{Input state $|{\bf{0}}\>$,  sequence of unitaries  $W_1, ... W_{D}$, $U_1(g_1),.. U_D(g_D)$, truncation parameter $L$. }
 	$p^T(\g)({\bf{z}}) = 0$\;
 	$len= [D/L]$\; \Comment{Min length of consecutive trivial irreps}
 	\For{$start \in\{1,...D-len\}$}{
 		$end = start+len$\;
 		\For{ $left, right, value \in \texttt{\textup{contract-full}}(W_{start}, ... W_{end})$}{
 			$left\_tree = \texttt{\textup{check}}( \texttt{GTT}(L, left, W_{start-1}\hc, ... W_{1}\hc), len)$\;
 			
 			\Comment{Checks the paths have not been included for previous start values}
 		}	
 		\If{$left\_tree$}{	
 			$right\_tree =  \texttt{GTT}(L, right, W_{end+1}, ... W_{D})$\;
 			\For{$left\_path \in left\_tree, right\_path\in right\_tree$}{
 				\If{$\texttt{\textup{length}}(left\_path)\hspace{-0.04cm}+\texttt{\textup{length}}(right\_path) \hspace{-0.1cm}\leq L$}{
 					$p^T(\g)(\z) += value\cdot (\texttt{evaluate}(left\_path, |{\bf{z}}\>))^{*}\cdot \texttt{evaluate}(right\_path, |{\bf{0}}\>)$
 				}
 			}
 		}
 	}

 	\KwOut{$p^T(\g)(\bf{z})$}
 	\vspace{0.1cm}
 	\hrule
 	\vspace{0.2cm}
 	\SetKwFunction{Fcontract}{contract-full}
 	\SetKwProg{Fn}{Function}{:}{\KwRet{}}
 	\Fn{\Fcontract{$W_{start}, ...  ,W_{end}$}}{
 	}
 	\KwRet{$\{ (|T^{\mu}_0\>\>, \<\<T^{\mu'}_0|, value) \, {\rm{s.t}} \,  \, value=\<\< T^{\mu}_0| \Pi_0 \W_{start} \Pi_0 \W_{start+1} ... \W_{end} \Pi_0  |T^{\mu'}_0\>\> \neq 0 \}$}
 	\vspace{0.2cm}
 	
 	\SetKwFunction{Fcheck}{check}
 	\SetKwProg{Fn}{Function}{:}{\KwRet{}}
 	\Fn{\Fcheck{$tree, len$}}{
 		\If{$tree$ {\rm{contains $branch$ with contractions}} $\W_{start-j}\Pi_0 ... \Pi_0\W_{start}|node\>\> \neq 0$ with $j\geq len$ }{
 			$delete(branch)$\;
 			\Comment{This condition should be checked dynamically as the tree is generated.} 
 		}
 	}	\KwRet{ tree}\;

 	\vspace{0.2cm}
 	\hrule
 	\vspace{0.2cm}
 	\label{alg:prob}
 \end{algorithm}

 \label{sec:bidirectional}
For certain applications such as  computing outcome probabilities the observables that need to be measured can have a decomposition into a number of basis operators that is not polynomial (or manageable with fixed resources).  To deal with that situation, we observe that imposing a truncation $L$ representing the maximal number of irreps that are not trivial within every valid path, implies that there exists at least a value $i \in \{1, D\}$ such that projectors onto the truncation set $\Pi_0$ are applied sequentially after unitaries $W_{i},... ,W_{i+J}$ for at least $J\geq D/L$ consecutive terms. This suggests a strategy where we scan the structure of a circuit and identify the minimal value of $L$ such that for all $i\in \{1, D\}$ we have that $\Pi_0 W_i \otimes W_i^{*}\Pi_0 ...  \Pi_0  W_{i + [D/L]} \otimes W_{i+[D/L]}\Pi_0$ has support on $poly(n)$  basis operators of the form $\{|T^\mu_0\>\>\<\< T^{\mu'}_0\}_{\mu,\mu'}$, where $\mu, \mu'$ label trivial irreps including multiplicities.

 For all $i\in \{1, D\}$ we perform the contraction to the above polynomially-sized basis and take each of those as roots for constructing truncated trees of height $L$ in both directions. We have the forward action given by expanding out in the Heisenberg picture
\begin{align}
\<\< T^{\mu'}_0|   W_{i+[D/L]+1} \otimes W_{i+[D/L]+1}^{*} ...   (U \otimes U^{*})(g_{D}) W_{D} \otimes W_{D}^{*}  \nonumber
\end{align}
and the backward action given by expanding out in the Sch\o dinger picture)
\begin{align}
	(U \otimes U^{*})(g_{1}) W_{1} \otimes W_{1}^{*} ...   (U \otimes U^{*})(g_{i-1}) W_{i-1} \otimes W_{i-1}^{*}  | T^{\mu}_0 \>\>. \nonumber
\end{align}
This essentially allows to cut the circuit into two pieces  $U_1W_1 ... U_{i-1}W_{i-1}$ and $W_{i+ [D/L]+ 1} U_{i+ [D/L]+ 1} ... U_DW_{D}$, and evaluate the contributions to $C(\g) \otimes C^{*}(\g)$ from all the valid paths by evaluating the trees $\texttt{GTT}(L, T^{\mu}_0, W_i\hc,...W_1\hc)$ and $ \texttt{GTT}(L, T^{\mu'}_0, W_{i+[D/L]+1} ,...W_D ) $.

The Algorithm \ref{alg:prob} describes how to use this procedure to evaluate sparse Fourier series decompositions for output probability distributions. Note that the same bi-directional trees generated for a given truncation level can be used to determine the probability distribution for \emph{different outcomes}.
The complexity of this algorithm  is going to be at most $D\times cost(\texttt{rank}) \times cost(\texttt{contract-full}) \times (sd_{max})^{2L}  \times cost(\texttt{evaluate})$, where $cost(\texttt{\texttt{rank}})$ represents the number of different root operators resulting from the contraction of $\Pi_0 W_i \otimes W_i^{*} \Pi_0 ,... \Pi_0 {W}_{i+[D/L]}\otimes {W}_{i+[D/L]}^{*} \Pi_0$ that incurs $cost(\texttt{contract-full})$ time complexity. In order to be efficient, it relies on the \emph{Assumption 3} described in the next section that ensures the contractions are possible and reduce to a polynomial number of terms.

\subsection{Noise-induced truncation of path decompositions}
We discuss how introducing certain types of noise into the quantum circuits can exponentially suppress the magnitude of the higher dimensional Fourier coefficients.  This effectively induces a truncation where the dominant terms give a sparse Fourier series that is amenable to the algorithms described previously.

Here we assume that the unitaries are subject to a finite amount of noise. We represent this with a quantum channel $\E:\mathcal{\B(\h_n)} \longrightarrow \B(\h_n)$ that acts on each layer in the circuit and denote by $\mathbf{E}$ its Liouville representation. For simplicity, we assume the noise to be gate and time independent channel that follows the same local structure as $U(g)$ so that the noisy unitary $C(\g) \otimes C^*(\g) $ is given in the Liouville representation by
\begin{equation}
\tilde{\mathbf{C}}(\g)= \mathbf{E} U(g_1) W_1 \otimes U^{*}(g_1) W_1^{*} ... \, \mathbf{E}U(g_D)W_D \otimes U^{*}(g_D)W_D^{*}.
\end{equation}

The requirement that the channel is trace-preserving implies that $ \<\<\mathbb{I}|\mathbf{E}|T^{\lambda}_k\>\> = 1$ for all basis vectors.  Furthermore completely positive trace preserving maps can be analysed in terms of their eigenvalues $e$ for which $\mathbf{E}|X\>\>= e|X\>\>$ for some operator $X$. The eigenvalues are real or come in complex conjugates pairs and are restricted to the unit disk. Since every channel has at least one fixed point, we denote the spectral gap by
\begin{equation}
\gamma := 1 - max_{|e|\neq 1} |e|.
\end{equation}

We next analyse how noise affects the Fourier decomposition, and to make it easier to track indices for now we will assume in the discussion that follows that the noise is diagonal in the ITO basis, that is $\mathbf{E}|T^{\lambda}_k\>\> = e_{\lambda,k} |T^{\lambda}_k\>\>$.  We discuss how this restriction can be lifted in the Appendix~\ref{app:noise}.  We note however, that for example a uniform noise model such as depolarising channel will satisfy this independently of the basis. Similarly, if we work with the Pauli basis the Pauli channels satisfy these constraints. Under this restriction that the noise model is diagonal, the noisy Fourier coefficients are given by
\begin{equation}
[\widehat{\tilde{\mathbf{C}}}(\boldsymbol{\lambda}) ]_{{\bf{k',k}}}=   e_{\lambda_0,k_0} ... e_{\lambda_{D-1},k_{D-1}}  [ \widehat{C\otimes C^{*}}(\boldsymbol{\lambda})]_{\bf{k',k}} .
\end{equation}



In general $|e_{\lambda,k}| \leq 1$  and a channel can have fixed points for which $e=1$ or rotating points for which the eigenvalues are on the periphery of the unit disk such that $|e|=1$.  When considering multiple applications of the channel that results in an asymptotic projection onto the unit disk eigenspace \cite{albert2019asymptotics}. 

In particular, for us this means that the Fourier coefficients will be contracted by a factor of  
\begin{equation}
|e_{\lambda_0,k_0} e_{\lambda_1, k_1} ... e_{\lambda_{D-1}, k_{D-1}} | \leq   (1-\gamma)^{L(\boldsymbol{\lambda},{\bf{k}})}
\end{equation} 
where  $L(\boldsymbol{\lambda},{\bf{k}}) $  denotes  the number of eigenvalues $e_{\lambda_i, k_i}$  not on the periphery of the unit disk within the path. 

If the unitary $U(\cdot)$ had no local structure and acts non-trivially on the entire Hilbert space $\H \cong (\mathbb{C}^d)^n$ then the previous bound will be trivially 1. This happens because every unitary fixes the identity on $\H$, so any path that encounters a full identity on $\H$ is trivial (e.g starts at $I$ and ends at the $I$ operator and has  weight coefficients $1$). Let's expand on then notation and meaning of  "non-trivial irreps within a path" when $U(\cdot)$ has a local structure, which will be the situation encountered in the examples considered. When $U(g)$ acts trivially on a subset of qudits, then any operator of the form $I\otimes B$  will transform trivially as  $U(g) I\otimes B U(g)^{\dagger} = I\otimes B$ for any operator $B$ acting only on the subset on which $U$ is trivial.  The noise channel is unital and has the same local structure as $U$ so in the Heisenberg picture it will also fix this operator $\E^{\dagger}(I \otimes B) = I\otimes B$, and as a result the noise in this layer has no effect on the particular path considered.  As we start from an operator $T^{\lambda}_k$ and propagate through the circuit, choosing particular paths when a branching occurs, then the number of times when a propagated operator of the form $I\otimes B$ occurs before applying a new layer in the circuit will correspond to the number of trivial irreps within that particular path. 

These observations motivate using the following condition for constructing a truncated set of paths 
\begin{align}	\label{eqn:truncation}
	\mathcal{T} :&= \{   (\boldsymbol{\lambda}, \boldsymbol{k}) : \texttt{validate}[ (\boldsymbol{\lambda}, \boldsymbol{k})] = True \\ \nonumber &  \& \ \texttt{truncate}[ (\boldsymbol{\lambda}, \boldsymbol{k})] = True    \} 
\end{align}
for a truncation rule defined for a fixed cut-off parameter $L$ by $\texttt{truncate}[ (\boldsymbol{\lambda}, \boldsymbol{k})] = False $  if and only if $ |e_{\lambda_0,k_0} ... e_{\lambda_{D-1}, k_{D-1}} | \leq   (1-\gamma)^{L}$. In other words, all the valid paths outside of the truncation set are suppressed by a factor of $(1-\gamma)^L$ or more.
This is just one particular way of building the truncation set $\mathcal{T}$, which is quite natural for the unital, covariant noise model considered, and relies on the fact that the fixed-points of the channel lie in the subspace that transforms trivially under $U(\cdot)$.

Finally we note that any noise that may incur on the $W$ operators will result in Fourier coefficients whose norms is further dampened, and therefore the approximation bounds may in fact be strengthened. 

\subsection{Estimating mean values of noisy circuits by an operator truncated classical simulation} 
Consider a class of circuits $\mathfrak{C}_n$ as in Sec.~\ref{sec:defs} with the task of estimating expectation values of observables $o(\g)  = Tr(O \, C(\g)\psi_0 C^{\dagger}(\g))$ for an input state $\psi_0$.  We will next see what are the general sufficient conditions that ensure these classes of circuits can be simulated efficiently when subject to noise. 

\emph{{\bf{Assumption 2 (Noise penalisation)}}. For the noise model, we assume every layer is subject to a channel $\E$ with spectral gap $\gamma$ and that the only fixed or rotating points correspond to the trivial irreps which means that  $|e_{\lambda,k}| <(1-\gamma)$ whenever $\lambda\neq 0$. }

The noisy expectation value is denoted by  $\tilde{o}(\g)$ and the truncated noisy expectation value by  $\tilde{o}^{\T}(\g)$, with respect to the truncation rule in the previous section.



Under the Assumption 1 and 2 above, we have the following informal statement  for $(sd_{max}) =\O(1)$ that is proved in Sec.~\ref{sec:results} in more general terms.

{{\bf{Result 1 (Theorem.~\ref{thm:mainmean})}}} \emph{If the observable $O$ decomposes into $r(O) = poly(n)$ number of terms with respect to a suitable basis, then there exists an explicit classical algorithm with runtime $poly(n) \left(\frac{||O||_{HS}}{\epsilon} \right)^{\O(1/\gamma)}$ which outputs a function $\tilde{o}^{\T}(\g)$ that approximates the noisy expectation value $\tilde{o}(\g)$ with average error at most $\epsilon$ over the entire parameter space.}

The core ideas are the following: 
\begin{enumerate}
	\item The form of the noise model determines a truncation set $\T$ such that every noisy Fourier coefficient corresponding to paths outside of the truncation set contracts by a factor of at most $(1-\gamma)^{L}$, for some fixed truncation parameter. This ensures Parseval's theorem will give us control over the average approximation error
	\begin{align}
		||\tilde{o}^{\T} (\g) - \tilde{o}(\g) ||_{L^2(\mathbb{G})} \leq  (1-\gamma)^{L}||O||_{HS}.
	\end{align}
	\item The assumption that $|e_{\lambda,k}| <(1-\gamma)$ for all $d_{\lambda} >1$ ensures that the growth in complexity always gets penalised by the truncation rule. In other words, as we go through the circuit each path branches out at a number of locations at most $L$.
	\item The number of \emph{valid paths} in the truncation set is at most $r(O) (sd_{max} )^{L}$ and each of them can pe efficiently estimated in polynomial time due to the efficiently contractible $W's$. 
	\item The classical algorithm relies on the tree data structure that explicitly enumerates all the valid paths and is given by Algorithm \ref{alg:two}. It  will consist of $sd_{max}$ -ary trees of height $L$ and there are at most $r(O)$ different trees corresponding to all $T^{\lambda_0}_{k_0}$ in the decomposition of the observable $O$. The only modification comes in the evaluation of the paths that will need to incorporate the error parameters $e_{\lambda,k}$. This can be achieved by introducing the additional factor in the function $\texttt{matrix-coeff}$ (or when evaluating the contractions if we consider further noise on the $W$ unitaries).
\end{enumerate}


\subsection{Approximate sampling of noisy circuits via an operator truncated classical simulation} 
Consider a class of circuits $\mathfrak{C}_n$ as in Sec.~\ref{sec:defs} with the task of sampling, that is producing a sample drawn from the output probability distribution $p({\bf{x}}) : =  Tr(|{\bf{x}}\>\<{\bf{x}}| \, C(\g)\psi_0 C^{\dagger}(\g))$ for input state $\psi_0$ and with respect to a computational basis ${\bf{x} }\in \mathbb{Z}_d^{\times n}$. 
A known result from \cite{bremner2017achieving} (see Lemma~\ref{lemma:samplingtocompute}) shows that sampling can be reduced to computing probability distributions and its marginals, within an $l_1$-norm approximation error.  One can use this as a reduction from the task of sampling from the noisy distribution $\tilde{p}({\bf{x}})$ to that of computing a function $\tilde{p}^{\T}({\bf{x}})$ and all its marginals $\tilde{p}^{\T}(x_1,...,x_k) = \sum_{x_{k+1},...,x_{n} \in \mathbb{Z}_d} \tilde{p}^{\T}(x_1,...,x_k) $ such that $||\tilde{p} -\tilde{p}^{\T} ||_1 = \sum_{{\bf{x}}} |\tilde{p}({\bf{x}}) - \tilde{p}^{\T}({\bf{x}})|\leq \epsilon $. 


In addition to the assumptions we've had for  noisy expectation value  we can  also consider the situations where the sub-circuits formed out of the $W$'s are efficiently contractible over the trivial subspaces induced by the unitaries $U(\cdot)$.

\emph{{\bf{Assumption 3 (Efficient contractions) }}  There exists an $M$ such that for all $i \in \{1, D-M\}$ the following operator
	\begin{align}
		\Pi_0 W_i\otimes W_i^{*} \Pi_0 ... W_{i+M}\otimes W_{i+M}^{*} \Pi_0 
	\end{align}
can be efficiently contracted to a support containing a polynomial number of  terms of the form $|T^{\lambda_0} \>\> \<\< T^{\lambda_0'}|$. 
}

Under this more restrictive set of conditions we have the following informal result obtaining a polynomial-time complexity, proved in Sec.~\ref{sec:results}.

{{\bf{Result 2 (Theorem.~\ref{thm:mainsampling})}}} \emph{
		Suppose the 
	probability distribution marginals $\sum_{x_{i_1},... x_{i_j}} \<{\bf{x}}|T^{\lambda}_{k} |\bf{x}\>$ can be efficiently evaluated for all basis operators. 
	Given Assumptions 1-3 then there exists an explicit classical algorithm which outputs a function $\tilde{p}^{\T}({\bf{x}})$ with $l_1$-approximation error $ ||\tilde{p}- \tilde{p}^{\T} ||_1 \leq \epsilon$, provided that the noiseless residual (quasi)-distribution $p -p^{\T}$ anti-concentrates and with runtime either
	\begin{enumerate}
		\item  constant-degree polynomial scaling in $n$ and exponential scaling in the inverse spectral gap $\gamma^{-1}$ with  $\O(poly(n)) \left(\frac{\O(1)}{\epsilon} \right)^{\O(1/\gamma)}$ if $s, d_{max} = \O(1)$ 
		\item quasi-polynomial in $n$ (for inverse polynomial precision $\epsilon$) and exponential in the inverse spectral gap $\gamma^{-1}$ with  $\O( n^{\O(\log{\epsilon^{-1}}/\gamma)})$ if $sd_{max} = poly(n)$  and  $\O(poly(n) \log{n}^{\O(\log{\epsilon^{-1}}/\gamma)})$ if $sd_{max} = \O(log(n))$ 
		\end{enumerate}
}


The definition for anti-concentration we use is to say $q(\bf{x}, \g)$ satisfies
\begin{align}
	\sum_{\bf{x} \in \mathbb{Z}_d^{\times n}} \int_{\g\in G}|q({\bf{x}}, \g)|^2 d\, \g   = \frac{O(1)}{d^n}.
\end{align}

As a corollary, the result also holds for probability distributions $p$ that anti-concentrate without any additional requirement on $p-p^{\T}$.  The more general result we have involving anti-concentration of $p-p^{\T}$ is stronger in the sense that it could also be applied to circuits that do not anti-concentrate but for which the truncated approximation captures the bulk in the distribution peaks. 

The core ideas behind this result are as follows:
\begin{enumerate}
	\item Approximation error bounds involving the $l_2$-norm arise from Parseval's theorem and the truncation rule lead to
	\begin{align}
	||\tilde{p} - \tilde{p}^{\T} ||_{l_2 (L^2(\mathbb{G}))} \leq  (1-\gamma)^L  ||p - p^{\T} ||_{l_2 (L^2(\mathbb{G}))}. 
	\end{align}
  	However,  we require an $l_1$ bound to obtain samples via the reduction discussed previously which uses only $O(dn)$ calls to the above algorithm computing $\tilde{p}^{\T}$. This is the only reason we require anti-concentration for the residual (quasi)-distribution. 
  	\item 
  	Let us pick $L$ to define the truncation parameter. Then pigeonhole principle implies that for any valid path there will be at least one $i \in \{1, ..., D-D/L\}$ such that the path does not branch out when applying $U(g)$ after any of the unitaries in the sequence $W_i$,... $W_{i+[D/L]}$. Assumption 3 then allows us to pick the value of $L$ that enables us to perform the \texttt{contract-full} efficiently. 
  	\item The construction directly uses Algorithm \ref{alg:prob}, with the additional modification of incorporating the error parameters in the same way as was done for the noisy expectation values.
  	
\end{enumerate}


\section{Examples of polynomial-time classical algorithms for simulating families of noisy circuits} 
Here we look at a series of noisy circuit classes to show how existing polynomial-time classical algorithms involving operator truncations can be described within our framework. We also show how our results can lead to improved algorithms that apply to a Pauli basis expansion and general Pauli channels with a spectral gap. 

\subsection{Pauli path decompositions for qubits}
Here we show how the framework applies to the situation where we consider $n$ -qubit system with the basis operators consisting of the normalised $n$-qubit Pauli operators forming the quotient group
\begin{align}
	\P_n = \{ \frac{1}{\sqrt{2^n}} P_1\otimes... \otimes P_n : P_i \in \{I, X,Y,Z\} \}
\end{align}

In particular, if the $W$-type of unitaries are Clifford then we see that the efficient contractibility assumptions will always hold as follows.  First, we have that the projector $\Pi_0$  is given by  $\Pi_0 = \sum_{P^0\in \mathcal{P}^{0}} |P^0\>\>\<\< P^0| $ where the sum is taken over all $n$-qubit Pauli operators $P^0$ that satisfy  $ U(g) P^0 U\hc(g) =  P^{0}$, which forms a subgroup $\P^0$ of $\P_n$. In particular, this will always include at least the (normalised) identity element $|\frac{I}{\sqrt{2^n}}\>\>$.  In general $\P^{0}$ will have a set of at most $k\leq 2n$ generators $\P^0  = \< P^0_1,..., P^{0}_k\>/\{\pm 1 , \pm i \}$ with a total number of  $2^k$ elements. 

The  main restriction on the circuit structure assumed by our results are as follows. First, we have that all trivial irreps in $U_i\otimes U_i^{*}$ must be spanned by Pauli operators. Then, if $W_i$ unitaries are Clifford then we see that all assumptions we outlined are guaranteed. For any Pauli operator
\begin{align}
	\Pi_0 W_i \otimes W_i^{*}  |P \>\> =  \pm |P_0\> \>   \     \delta_{P_0,  W_iPW_i\hc} 
\end{align}
can always be performed efficiently in at most $O(n^2)$.  Furthermore, this also always satisfies the controlled tree growth condition. In our notation, for Clifford operations $s=1$. Therefore, assuming look-up access $\O(1)$ to the evaluation of irreducible matrix coefficients, the results on computing expectation values with respect to a Pauli basis decomposition leads to a complexity of $\O(r(O) n^2 L d_{max}^L )$  in order to determine truncated noisy expectation values that approximate the noisy values up to a $e^{-\gamma L}$-decaying bound. Therefore, if $L = O(1)$ then we get a polynomial-time algorithm for approximating noisy expectation values and if $L=O(log(n))$ we get a polynomial-time algorithm whenever maximal irrep dimensions are independent of $n$ and quasi-polynomial when $d_{max} = O(poly(n))$.  

Second, the sampling results  may require an additional Assumption 4 which involves contractions of the form 
\begin{align}
	\label{eqn:contraction-Pauli}
	\Pi_0 W_i\otimes W_i^{*} ... \Pi_0 W_{j} \otimes W_{j}^{*} \Pi_0   = \hspace{-0.3cm} \sum_{P \in \P^{0}}  |P\>\> \<\< W_j\hc ... W_i\hc P W_i ...W_j|  
\end{align}
which can also be performed efficiently, where the sum is constrained to operators $P\in \P^0$ that also must satisfy $W_{i+s}\hc ... W_i \hc P W_i ... W_{i+s} \in \P_0$ for all $s\leq j-i$. In general these series of constraints can reduce the number of valid terms $P\in \P^0$ in the above decomposition. This can be efficiently tracked using the generators of $\P^{0}$. We have that for any $P\in \P^0$  and a Clifford unitary $W_i$, then the Pauli operator (up to quotienting a phase)  $W_i\hc P W_i$ lies in $\<W_i\hc P^{0}_1 W_i,... W_i\hc P^{0}_k W_i\>/\{\pm 1, \pm i\}$. If any of the new generators lie outside of  $\P_0$ then the effect of projecting onto it via $\Pi_0$ is to remove any such generators, thus the number of total terms halves with every generator we remove. This procedure is then repeated for all $W_i,W_{i+1} ... W_j$, for $j$ large enough so that we remain with $O(log(n))$ generators leading to  a total number of terms  $r(i) = poly(n)$ in Equation~\ref{eqn:contraction-Pauli}. This is will not be necessarily always guaranteed - and will depend on the structure of the circuits. In certain situations, as we sweep through the circuit the contraction could lead to the projector $|\frac{I}{\sqrt{2}}\>\<\frac{I}{\sqrt{2}}|$ - which trivialises any further propagation to the left or right.  

Finally, we remark that all of the above similarly holds should we have different trivial subgroups of the Pauli group for the different unitaries $U_i(g)$. This leads to different projectors $\Pi_0^{(i)}$ at each location so it simply becomes a matter of keeping track different generators and constraints with additional labels - at most this would incur an $\O(2n D)$ memory overhead.



\subsection{Noisy random circuit sampling}
In this section we describe polynomial-time algorithms for noisy random circuit sampling. This includes the  poly-time algorithm in \cite{aharonov2023polynomial} that can  be viewed as a particular example of the framework developed here, but also we present a new construction based on a different truncation method that allows generalisation to gapped two-qubit Pauli noise models. First, we show a connection between products of irreducible matrix coefficients and random circuits that may be of broader interest, and arises from the harmonic analysis we described previously. 

\subsubsection{Harmonic analysis of SU(4)-random circuits}
For random circuit sampling on $n$ qubits, the underlying group we are looking at is $SU(4)$ (see the following review of random circuit sampling \cite{hangleiter2023computational})and first we look at the defining fundamental representation $U(g) = g$ on two qubits $\h \approx {\mathbb{C}^{2}}\otimes  \mathbb{C}^{2}$ which is given by matrix multiplication for $g\in SU(4)$. The representation $U\otimes U^{*}$  we are interested in decomposes into two irreducible components  $U\otimes U^{*} = V^0\oplus V^{1}$ where $V^0$ is the trivial representation and $V^{1}$ is the adjoint irreducible representation, which has dimension $d_1 = 15$.  The $0$-irrep ITO is given by the identity $T^0:=\frac{1}{2} I\otimes I$ and for the $1$-irrep we can choose an orthonormal irreducible tensor operator basis  $\{T^{1}_{k}\}_{k=1}^{15}$ formed of products of (normalised)single qubit Pauli operators where at least one of them is not the identity. These will transform under $g \in SU(4)$ as $ g T^{1}_k g^{-1}  = \sum_{k'} v^1_{kk'}(g)T^{1}_{k'}$ and $g T^0 g^{-1} = T^0$. The irreducible matrix components $v^{1}_{kk'}(\cdot)$ can be given in terms of special functions for suitable parametrisations of group elements in $SU(4)$ or from Lie algebraic analysis as detailed in Appendix ~\ref{app:ITOU4}.

The random circuits  $C(\g)$ consist of a total of $D$ two-qubit unitaries randomly chosen with the uniform Haar measure on $SU(4)$ according to $g_1,..., g_D \in SU(4)$.  The pairs of qubits on which these unitaries act on is fixed, and this is captured by the permutations $W_i$ applied after each random unitary. 
 
In terms of the $n$-qubit Pauli operator basis we have that for any (normalised) Pauli operator $P\in \mathcal{P}_{n-2}$ on $n-2$ qubits terms of the form  $T^{\lambda}_k \otimes P$ transform as a $(\lambda,k)$ ITO and the multiplicities are captured by all $P$ so that
\begin{align}
	(U(g) \otimes \mathbb{I}) \otimes (U(g)^{*} \otimes \mathbb{I}) |T^1_k \otimes P\>\>  = \sum _{k'}  v^{1}_{k'k} (g)  |T^1_{k'} \otimes P\>\> .
\end{align}
Similarly  $|T^{0} \otimes P\>\>$ is left invariant for all $g$,  so the projector onto the subspace on which $U\otimes U^{*}$ acts trivially is given by
\begin{align}
	\Pi_0 : = \sum_{P\in \P_{n-2}} |T^{0}\otimes P \>\>\< \< T^{0} \otimes P|
\end{align}
Furthermore we can also define the projector onto the 1-irrep isotypical component
\begin{align}
\Pi_1 : = \sum_{P\in \P_{n-2}}  \sum_{k,k'} |T^{1}_k \otimes P\>\>\<\< T^{1}_{k'} \otimes P| \otimes |k'\>\<k|,
\end{align}
where we note that $\Pi_0 + \tr_{2}( \Pi_1) = \mathbb{I}$ and that the second system is a space of dimension $15$, viewed as the carrier space for the 1-irrep.  
Then the Fourier modes for the random circuit $C(\g)$ at the inequivalent irreps (i.e ``frequencies") $\boldsymbol{\lambda}  = (\lambda_1,... ,\lambda_{D})$ of $U(4)^{\times D}$ labelled by $\lambda_i \in \{0,1\}$ corresponding to the trivial and adjoint irrep respectively are 
\begin{align}
	\widehat{C\otimes C^{*} } (\boldsymbol{\lambda}) = \Pi_{\lambda_1} W_1 \otimes W_1^{*} \Pi_{\lambda_2} \  ...  \  \Pi_{\lambda_D} W_D\otimes W_D^{*}.
\end{align}
We understand the above expression as a matrix on the space representing the $n$ qubits and carrier space for the irrep $v^{\lambda_1} \otimes... \otimes v^{\lambda_D}$, specifically $\h_n\otimes \h_n^{*} \otimes \mathbb{C}^{d_1} \otimes ... \mathbb{C}^{d_{\lambda_D}}$.
Finally, the random circuit $C(\g)$ can be expressed in terms of the Fourier modes and the products of irreducible matrix coefficients  $v^{\boldsymbol{\lambda}}_{\bf{k}, \bf{k'}} (\g):= v^{\lambda_1} _{k_1,k_1'}(g_1) ... v^{\lambda_D}_{k_D, k_D'}(g_D)$ as
\begin{align}
	C(\g)\otimes C^{*}(\g)  =   \sum_{\boldsymbol{\lambda, {\bf{k}}, {\bf{k}'}} } v^{\boldsymbol{\lambda}}_{\bf{k}, \bf{k'}} (\g)  \tr_2 [\widehat{C\otimes C^{*}}(\boldsymbol{\lambda}) \mathbb{I}\otimes |\bf{k'} \> \<\bf{k} | ]
\end{align}
As we have seen this before, the harmonic decomposition allows us to separate circuit structure from any parameter-dependency.

\subsubsection{Efficient simulation algorithm for truncated Fourier series}
The number of components in the Fourier mode expansion can be controlled through the number of irreps that are non-trivial in $\boldsymbol{\lambda} = (\lambda_1,..., \lambda_D)$.  Let's define the following truncation set
\begin{align}
	\T = \{ \boldsymbol{\lambda} :  \sum_{i} \lambda_i \leq L \}.
\end{align}
This condition implies that all paths in the truncation must satisfy the following: there exists $i \in  \{1,D- [D/L]\}$ such that  $\lambda_i  = \lambda_{i+1} = ... \lambda_{i+[D/L]} = 0$.
Furthermore, we look at the object $\Pi_0 W_{i} \otimes W_{i}^{*} ... W_{i+D/L -1 } \otimes W_{i+D/L -1 }^{*} \Pi_0 =  \sum |T^0\otimes P\>\>\<\< T^{0} \otimes P' | $, where the summation involves a \emph{restricted} set of $r(i)$ pairs of  Pauli operators $P$ and $P'$.  Note that $\Pi_0$ involves $4^{n-2}$ terms and each of the $W_i$'s is a fixed permutation that swaps the first two qubits with another pair of qubits. This means that if we look at $\Pi_0 W_i \otimes W_i^{*} \Pi_0$ this will involve a number of terms $4^{n-4}$ if the two qubit pairs do not overlap or $4^{n-3}$ if the two qubit pairs overlap on one qubit.   What this means is that in general we obtain $r(i)$ by counting the number of of distinct qubits involved in the unitaries - namely if $W_i$ acts on qubit pairs $(a_i, b_i)$ and if there are $(n-c-2)$ distinct qubits involved in $W_i$,... $W_{i+D/L-1}$  for some constant or at most log-scaling $c$ then $r(i)  = 4^c$. 
In the example considered here we have $s=1$ as the $W$'s are permutations and $d_{max} = 15$, the dimension of the 1-irrep. Then, the assumptions in our framework  (in the main general sampling result)  require  $r(i) = poly(n)$ for all $i\in  \{1, D- D/L\}$ and $L $ at most $\O(log(n))$ in order to obtain a simulation algorithm that has polynomial-time scaling in the number of qubits. If these hold, then the complexity of the simulation algorithm that computes the truncated (quasi)-probability distribution and all its marginals $p^{\T}(\bf{ x})$ will be $\O(n^2) \O(L 15^{L} ) \sum_{i=1}^{D-D/L} r(i)  = \O(D)\O(poly(n)) \O(L \,15^{L})$. In the most general case, these conditions will be met whenever $D = O(n L)$ 
 which results in an overall complexity $\O(poly(n) ) \O(L15^L)  =\O(poly(n))$ when $L$ is at most  $\O(log(n))$ or equivalently for random circuits consisting of at most $\O(log(n))$ layers.
 


In the above, we have made no assumptions on the connectivity or on the repeating pattern of layers dictated by the $W$'s - generally this structure will impose a bound on the degree of the polynomial for each $r(i)$ and that can lead to improved complexity compared to the worst case analysed.

\emph{Example 2D random circuits}  

Let's consider circuits with a lattice architecture 
applying a random unitary between qubits on all the $n_{edges}$ edges. 
These are repeated for $n_{layers}$ . In the notation used in this framework, such circuits correspond to $C(\g) = U(g_1) W_2... U(g_D) W_D$  with with $D = n_{layers}n_{edges}$ and where the random unitary $U$ is applied  on the edge $(1,2)$ between first and second qubit and $W$ unitaries consist of  $SWAP_{(a, 1)}  SWAP_{(b,2)}$ for every edge $(a,b)$ ($a<b$).
We now discuss how to deal with contractions of the form
$\Pi_0 W_i \Pi_0 ... W_{i+M} \Pi_0 $.  The projector $\Pi_0$ involves the Pauli operators generated by the single qubit Pauli's $\< X_{3},... X_{n}, Z_3,..., Z_{n}\>$ The effect of the swaps removes $X_{a}, Z_a , X_b , Z_b$ from the generator set. Therefore if $W_{i},..., W_{i+M}$ involves all qubits we have $\Pi_0 W_i \otimes W_i^{*}\Pi_0 ... W_{i+M} \otimes W_{i+M}^{*}\Pi_0  \propto  |I\>\>\<\<I| $
 so it contracts trivially. Otherwise if they involve all but a set of qubits $i_1, ..., i_c$ then 
\begin{align}
\Pi_0 W_i  \otimes W_i^{*} ... \Pi_0  &= \sum_{P,P' \in \< X_{i_1}, Z_{i_1},... X_{i_c}, Z_{i_c} \> }| P\>\>\<\<P' |,
\end{align} 
which consists of $4^c$ terms.  

Just to have an example of a structure to work with we consider $n =2n_1 \times 2n_1$ qubits with each layer consisting of four depth-1 layers such that the first and second layer have the edges in the x direction and third an fourth layer have the edges in the y direction.  Each depth-1 layer involves  either all $n$ qubits and $2n_1^2$ unitaries or $n- 4n_1$ qubits and $2n_1^2-2n_1$ unitaries. Therefore, if we pick sequential $W$-unitaries with $W_i,..., W_{i+M}$ with $M = 4n_1^2-2n_1-c$ then each of the contractions, as we sweep through the circuit for  $i = 1, ..., D-M$, will reduce to at most $4^c$ terms, for a fixed constant $c$.
Now for each such Pauli $P$ we are propagating it forwards through $W_{i+M+1} ... U(\g_D)W_D$ and backwards through $ U(\g_1)W_1.... U(\g_{i-1})W_{i-1}$ to construct trees of height at most $L \geq  D/M = n_{layers} \frac{8n_1^2-4n_1}{4n_1^2-2n_1-c}$ 
which corresponds to including paths that branch-out at most $L $ times. Therefore  each of the trees will involve $15^L$ branches and each can be evaluated in $O(L)$ time. The number of such trees we need to construct is going to be at most $ 2\cdot 4^c (D-M)$, with many of the contractions being trivial for this particular connectivity.  
\subsubsection{Noisy circuits and  approximation bounds} 
Consider a noise model where the 2-qubit random unitary is followed by a 2-qubit Pauli channel $\E$ with spectral gap $\gamma$. This means that the identity on both qubits is the only basis fixed by the noise channel.  All the paths outside of the truncation set defined previously will have Fourier modes with magnitudes contracted by $(1-\gamma)^{L}$. 
This will imply that we can use the previous simulation algorithm -  where the tree structures now take into account the noisy coefficients instead given by $e_{k} v^{1}_{kk'}(\g) $ with $|e_{1,k}|\leq (1-\gamma)$ the eigenvalue corresponding to the two-qubit Pauli labelled by $k$.   This will produce a truncated noisy probability distribution $\tilde{p}^{\T}$ that approximates the noisy distribution $\tilde{p}$ within a bound that decays as  $(1-\gamma)^{L} = \frac{1}{n^{\O(1/\gamma)}}$.  
In the situation where the connectivity of the random circuits match the topology of the device then the $W$ operators are not noisy so there will be no additional decay, however generally noisy swaps are expected to further decay the Fourier modes, in which case the previous bounds will not be tight. 

\subsubsection{Connections with other poly-time algorithms for noisy random circuits}
 A poly-time algorithm for noisy  RCS is described in the seminal work \cite{aharonov2023polynomial}. Here we see how that fits into our framework, for a basis decomposition also based on the Pauli operators.  First, the noise model considered there is a single qubit depolarising channel with error $p$ and the truncated set of paths are based on the total Hamming weight of each path.   With our notation, the two-qubit noise model will have eigenvalues $|e_{0,0}|=1$ for the identity Pauli basis and $ |e_{1, k}| =  (1-3/4p)$ whenever $k$ labels  $I\otimes P$  or $P\otimes I$, with $ |e_{1, k}| =  (1-3/4p)^2$ otherwise. This means that for each path we have a bound of
\begin{align}
| e_{\lambda_0,k_0} ....  e_{\lambda_D, k_D} | \leq (1-3/4p)^{Hamming(\boldsymbol{\lambda},\bf{k})}
\end{align} 
This justifies the connection between the truncation parameter $L$ and total Hamming weight of a full Pauli path, as explored in \cite{aharonov2023polynomial}. The approximation error bounds can also be viewed as a consequence of Parseval's theorem, with the mention that the orthogonality of paths employed in this previous work can in fact be attributed to the Schur's orthgonality relations that arises as a consequences of the Peter-Weyl theorem.  
We note how the enumeration algorithm used in \cite{aharonov2023polynomial} is different than the algorithmic constructions we have here. 

\subsection{Mean values of U(1)-parametrised noisy circuits } \subsubsection{Connection with LOWESA}
We discuss next how the results and the polynomial-time algorithm for expectation values of noisy circuits from \cite{fontana2023classical} can also be viewed within the framework here, provided we extend to decompositions involving \emph{real} irreducible representations, which come with some modifications arising from the decomposition of the group algebra $\mathbb{R}[G]$ outlined in Appendix~\ref{app:realU1}.  The class of noisy circuits considered there are parametrised by angles $g_i \in [0,2\pi]$ with $\g \in [0,2\pi]^{\times D}$ and compiled into a Clifford and $Z$-rotations with the task to estimate expectation values $O$ that decompose into at most $r(O) = poly(n)$ Pauli operators on $n$ qubits.

With the notation here, the unitary $U(g)= e^{iZg}$ for $Z$ - single qubit Pauli operator can be viewed as a representation of the $U(1)$ group on single qubit $\h = \mathbb{C}^{2}$. The $0$-irrep ITOs will occur with multiplicity so that $T^0_{1} = \frac{1}{\sqrt{2}} {I}$ and $T^0_{2} = \frac{1}{\sqrt{2}} {Z}$ transform trivially  $U(g) T^{0}_{k} U\hc(g) = T^{0}_{k}$ for all $g\in U(1)$.  Finally the Pauli operators $T^{1}_1  = X/\sqrt{2}$ and $T^{1}_2 = Y/\sqrt{2}$ will form a 2-dimensional \emph{real} irreducible representation, that we label by $\lambda= 1$ with matrix coefficients given by trigonometric polynomials $\cos{2g}$ and $\sin{2g}$. On the full $n$-qubit system, the $\lambda=1$ irreps will occur with multiplicity so that $\{|X\otimes P\>\>, |Y\otimes P\>\>\}_{P\in \P_{n-1}}$, and similarly basis operators that transform as $\lambda=0$  irrep are given by $\{|I\otimes P\>\>\}_{P\in\P_{n-1}}$ and $\{|Z\otimes P\>\>\}_{P\in\P_{n-1}}$.

The class of circuits will then consist of fixed Clifford unitaries $W_i$, that will also include a potential swap operation between the first qubit and any other qubit; just as as before the swap is noise-less as it is just a convenience for notation and labelling.  However, the  remaining  Clifford operations can be noisy. 
The truncation used in \cite{fontana2023classical} relies on including paths $(\boldsymbol{\lambda},{\bf{k}})$ that branch out in at most $L$ locations, or equivalently all the valid paths that have at most $L$ irreps equivalent to $\lambda=1$. 
In our notation, it corresponds to  
\begin{align}
\T = \{ \texttt{validate}((\boldsymbol{\lambda},{\bf{k}}) )  = True  \   \   \&  \sum_{i} \lambda_i \leq L \  \},
\end{align}
where the validation comes from starting out from an output Pauli operator, that needs to be in the support for the decomposition of the observable $O$, and propagating it through the circuit backwards (Heisenberg picture). Furthermore, the Clifford unitaries $W$'s can be efficiently contracted in $O(n^2)$ time, which allows to construct trees of paths with height at most $L$. Therefore \emph{Result 1} applies leading to the polynomial-time algorithm from \cite{fontana2023classical} for expectation values of noisy parametrised circuits.


\section{Further generalisations} 
In the discussion so far we have considered circuits $\mathfrak{C}_n$ that contain unitary representations of a compact or finite group $G$, and employed the decomposition of the space of integrable functions on the group $L^2(G)$ into irreducible components to develop a harmonic analysis for these quantum circuits. The framework can also be applied even when the parametrised unitaries do not form a group but carry the structure of a homogeneous space - or equivalently when there is a transitive group action on the parameter space. The results that we have used could be applied to smaller parameter spaces such as group quotients $G/H$. In such a situation, the only change is that the irreducible matrix coefficients are replaced by zonal spherical functions \cite{klimyk1995representations} and that the average approximation errors are taken over the space of $L^2(G/H)$ with an integration measure on the homogeneous space $G/H$.  Specifically, for $H$ a closed subgroup of $G$ and $\mu$ a normalised $G$-invariant measure on $G/H$ the Parseval's theorem generalises to
\begin{align}
	\sum_{\lambda \in  \widehat{G/H}} c_{\lambda}|| \hat{f}(\lambda)||^2_2 = ||f||^2_{L^2(G/H)}
\end{align}
for complex-valued functions $f : G/H \longrightarrow \mathbb{C}$ and norm $||f||^2_{L^2(G/H)} = \int |f(gH)|^2 d\, \mu (gH)$, with the integration over coset elements $gH$.  The inequivalent irreps $\lambda\in \widehat{G/H}$ correspond to irreps $v^\lambda$ of $G$ for which $\int_H v^{\lambda}(h)dh \neq 0$. The $c_\lambda$ corresponds to the multiplicity of $\lambda$-irrep in $L^2(G/H)$; for example if $H$ is a massive subgroup then $c_\lambda=1$. This ensures that the average approximation error between the noisy and truncated noisy quantities (observables or output probability distributions) can be bounded in an analogous way for parameters $g_i \in G/H$ as we have done for compact groups. In particular, an application of this extension is that the results we have for noisy random circuits consisting of two-qubit gates drawn according to the Haar measure on $SU(4)$ also hold for $SU(4)/H$ for closed subgroups $H\subset SU(4)$. This covers practical situations where the gates are randomly chosen from a restricted gate set given by native operations of the quantum device.

Another direction in which the framework directly applies is the situation where unitaries at different locations are parameterised by different groups $G_i$ so that we will need to  consider a harmonic analysis over $\mathbb{G} =G_1\times... \times G_{D}$.  The main difference is that we would need to carefully consider operator basis expansions at different locations that are both suitable for the action of the local group $G_i$ and for which transition matrix coefficients for the fixed $W$ operators can be efficiently evaluated. 

Finally, we emphasise again that the main algorithmic construction we rely on for evaluating the sparse truncated Fourier decomposition is based on projectors onto the trivial subspaces of $U_i\otimes U_i^{*}$. However, one could have multiple irreducible subspaces that are 1-dimensional and are inequivalent to the trivial. In such cases the growth of complexity due to the $U_i$ unitaries is also controlled  because the path basis expansion does not incur additional branches. So, we can analogously consider multiple projectors onto not only the trivial subspace but also other 1-dimensional irreps $\lambda_0$  with characters $\chi_{\lambda_0}(g)$ that for the complex case correspond to roots of unity.

\section{Discussion}
There are three largely modular contributions in this work. First, we describe how classes of unitaries with a structure that encodes a group action can be viewed in terms of harmonic analysis. This provides a way to separate the group parameter dependencies that is captured by matrix coefficients of irreducible representations from the fixed structure that is captured by the generalised Fourier coefficients. Furthermore, we show this analysis is similar to decomposing the vectorised sequence of unitary channels into multiple paths with respect to a symmetry-adapted basis of operators. This provides a perspective on random circuits  that allows to separate the randomness (i.e which particular gates are chosen from a given distribution) from the fixed architecture (i.e where those gates are applied). 

Second, we discuss how truncated Fourier coefficients can be evaluated from tree structures that enumerate the valid paths. For our purposes, the tree structures were enough to establish a connection between particular truncations and the number of valid paths. In general, one cannot guarantee convergence of the truncated Fourier decomposition when increasing the truncation parameter. This presents the challenge of identifying which Fourier coefficients should be kept in the truncation to ensure a target approximation error. It also raises the question of what are the conditions on the unitaries that can lead to an approximation by a polynomial number of terms.  

Thirdly, we consider the effects of noise and show that gapped error models local to the unitaries parameterised by group elements effectively induce a truncation that allows to bound the approximation error. Combined with the tree-based algorithm that evaluates all the valid paths in the truncation this guarantees a polynomial-time simulation for sampling noisy unitaries and computing noisy expectation values.

The main results on classical simulability of tasks for noisy circuits involve bounds on average errors between noisy quantities and their truncated noisy approximations.  However, these do not necessarily carry over when considering how well the noisy values are approximated by the truncated noiseless quantities. This point also applies to previous works constructing polynomial-time algorithms for random circuit sampling \cite{aharonov2023polynomial} and expectation values of parameterised circuits \cite{fontana2023classical} covered within this framework. This means that we obtain provable guarantees on the efficiency of classical simulation algorithms in a situation where the parameters of the error model are incorporated into the operator truncation. 

Finally, we remark that the approximation guarantees we derive are with respect to mean square errors averaged over the entire parameter space. Whilst we have seen the classes of unitaries may be further restricted by generalising to parameters representing elements of a coset, the bounds arising from Parseval's theorem still involve average errors, albeit over a smaller parameter space.  For sampling tasks this can be well motivated when we select unitaries (or circuits) randomly from a given ensemble.  However, as we are considering circuits constructed from unitary representations the results here also apply to correlated parameters - as long as the different parametrised layers $U_i(g_i)$ are still independent.

We have discussed throughout how the framework here generalises the results of \cite{aharonov2023polynomial, gao2018efficient} and \cite{fontana2023classical} dealing with noisy circuits.  In the noiseless setting, several other works actively use the underlying group structures and symmetries to investigate classical simulability of evaluating expectation values when the circuits involve a symmetry principle \cite{anschuetz2023efficient}.  The simulation framework of $\mathfrak{g}$-sim  \cite{goh2023lie} uses dynamical Lie algebras to compute basis expansions of evolved observables for parametrised circuits by identifying small invariant subspaces for the simulation to be tractable.  
While no trucation and only parameterised unitaries are involved, it can be viewed in our framework as an example where $U_i(g)$ are the defining representation of the dynamical Lie group whose action transforms a fixed algebra basis in the adjoint representation.  Finally we note that previous work \cite{nemkov2023fourier} considered trigonometric Fourier series of mean values in variational algorithms in the context of classical simulation. The relation between Fourier series and parametrised circuits has also been used in the context of quantum machine learning. This work describes a formalism to generalise these connections.  

In the late stages of preparing this manuscript, related work \cite{schuster2024polynomial}, \cite{gonzalez2024pauli} was published that uses Pauli-based path decompositions for classical simulations of noisy circuits and \cite{angrisani2024classically} in the noiseless case.  It will be interesting to compare if the approximation error guarantees obtained in the noiseless case for expectation values of random circuits \cite{angrisani2024classically} can also be related to the regularity conditions that ensure convergence of the truncated Fourier approximation. 

\section{Acknowledgements}
I would like to express my gratitude and thanks to Pablo Andres-Martinez, Fred Sauvage and David Amaro for feedback that helped to improve the manuscript and to Ross Duncan for supporting this research direction.  Many thanks also to Peter Burgisser for suggesting the reference \cite{burgisser2000computational}.  I acknowledge support from the Simons Institute, to visit during the ``Summer Cluster on Quantum Computing" programme. 
\bibliography{main}

\begin{thebibliography}{10}

\bibitem{boixo2018characterizing}
Sergio Boixo, Sergei~V Isakov, Vadim~N Smelyanskiy, Ryan Babbush, Nan Ding,
  Zhang Jiang, Michael~J Bremner, John~M Martinis, and Hartmut Neven.
\newblock Characterizing quantum supremacy in near-term devices.
\newblock {\em Nature Physics}, 14(6):595--600, 2018.

\bibitem{hangleiter2023computational}
Dominik Hangleiter and Jens Eisert.
\newblock Computational advantage of quantum random sampling.
\newblock {\em Reviews of Modern Physics}, 95(3):035001, 2023.

\bibitem{morvan2023phase}
Alexis Morvan, B~Villalonga, X~Mi, S~Mandra, A~Bengtsson, PV~Klimov, Z~Chen,
  S~Hong, C~Erickson, IK~Drozdov, et~al.
\newblock Phase transition in random circuit sampling.
\newblock {\em arXiv preprint arXiv:2304.11119}, 2023.

\bibitem{decross2024computational}
Matthew DeCross, Reza Haghshenas, Minzhao Liu, Yuri Alexeev, Charles~H Baldwin,
  John~P Bartolotta, Matthew Bohn, Eli Chertkov, Jonhas Colina, Davide
  DelVento, et~al.
\newblock The computational power of random quantum circuits in arbitrary
  geometries.
\newblock {\em arXiv preprint arXiv:2406.02501}, 2024.

\bibitem{chen2023complexity}
Sitan Chen, Jordan Cotler, Hsin-Yuan Huang, and Jerry Li.
\newblock The complexity of nisq.
\newblock {\em Nature Communications}, 14(1):6001, 2023.

\bibitem{francca2021efficient}
Daniel~Stilck Fran{\c{c}}a, Sergii Strelchuk, and Micha{\l} Studzi{\'n}ski.
\newblock Efficient classical simulation and benchmarking of quantum processes
  in the weyl basis.
\newblock {\em Physical Review Letters}, 126(21):210502, 2021.

\bibitem{stilck2021limitations}
Daniel Stilck~Fran{\c{c}}a and Raul Garcia-Patron.
\newblock Limitations of optimization algorithms on noisy quantum devices.
\newblock {\em Nature Physics}, 17(11):1221--1227, 2021.

\bibitem{aharonov2023polynomial}
Dorit Aharonov, Xun Gao, Zeph Landau, Yunchao Liu, and Umesh Vazirani.
\newblock A polynomial-time classical algorithm for noisy random circuit
  sampling.
\newblock In {\em Proceedings of the 55th Annual ACM Symposium on Theory of
  Computing}, pages 945--957, 2023.

\bibitem{gao2018efficient}
Xun Gao and Luming Duan.
\newblock Efficient classical simulation of noisy quantum computation.
\newblock {\em arXiv preprint arXiv:1810.03176}, 2018.

\bibitem{fontana2023classical}
Enrico Fontana, Manuel~S Rudolph, Ross Duncan, Ivan Rungger, and Cristina
  C{\^\i}rstoiu.
\newblock Classical simulations of noisy variational quantum circuits.
\newblock {\em arXiv preprint arXiv:2306.05400}, 2023.

\bibitem{shao2023simulating}
Yuguo Shao, Fuchuan Wei, Song Cheng, and Zhengwei Liu.
\newblock Simulating quantum mean values in noisy variational quantum
  algorithms: A polynomial-scale approach.
\newblock {\em arXiv preprint arXiv:2306.05804}, 2023.

\bibitem{mele2024noise}
Antonio~Anna Mele, Armando Angrisani, Soumik Ghosh, Sumeet Khatri, Jens Eisert,
  Daniel~Stilck Fran{\c{c}}a, and Yihui Quek.
\newblock Noise-induced shallow circuits and absence of barren plateaus.
\newblock {\em arXiv preprint arXiv:2403.13927}, 2024.

\bibitem{zurek2003decoherence}
Wojciech~Hubert Zurek.
\newblock Decoherence, einselection, and the quantum origins of the classical.
\newblock {\em Reviews of modern physics}, 75(3):715, 2003.

\bibitem{ryan2024high}
C~Ryan-Anderson, NC~Brown, CH~Baldwin, JM~Dreiling, C~Foltz, JP~Gaebler,
  TM~Gatterman, N~Hewitt, C~Holliman, CV~Horst, et~al.
\newblock High-fidelity teleportation of a logical qubit using transversal
  gates and lattice surgery.
\newblock {\em Science}, 385(6715):1327--1331, 2024.

\bibitem{wang2024fault}
Yang Wang, Selwyn Simsek, Thomas~M Gatterman, Justin~A Gerber, Kevin Gilmore,
  Dan Gresh, Nathan Hewitt, Chandler~V Horst, Mitchell Matheny, Tanner Mengle,
  et~al.
\newblock Fault-tolerant one-bit addition with the smallest interesting color
  code.
\newblock {\em Science Advances}, 10(29):eado9024, 2024.

\bibitem{google2023suppressing}
Suppressing quantum errors by scaling a surface code logical qubit.
\newblock {\em Nature}, 614(7949):676--681, 2023.

\bibitem{bluvstein2024logical}
Dolev Bluvstein, Simon~J Evered, Alexandra~A Geim, Sophie~H Li, Hengyun Zhou,
  Tom Manovitz, Sepehr Ebadi, Madelyn Cain, Marcin Kalinowski, Dominik
  Hangleiter, et~al.
\newblock Logical quantum processor based on reconfigurable atom arrays.
\newblock {\em Nature}, 626(7997):58--65, 2024.

\bibitem{daher2019titchmarsh}
Radouan Daher, Julio Delgado, and Michael Ruzhansky.
\newblock Titchmarsh theorems for fourier transforms of h{\"o}lder--lipschitz
  functions on compact homogeneous manifolds.
\newblock {\em Monatshefte f{\"u}r Mathematik}, 189(1):23--49, 2019.

\bibitem{rudolph2023classical}
Manuel~S Rudolph, Enrico Fontana, Zo{\"e} Holmes, and Lukasz Cincio.
\newblock Classical surrogate simulation of quantum systems with lowesa.
\newblock {\em arXiv preprint arXiv:2308.09109}, 2023.

\bibitem{beguvsic2024fast}
Tomislav Begu{\v{s}}i{\'c}, Johnnie Gray, and Garnet Kin-Lic Chan.
\newblock Fast and converged classical simulations of evidence for the utility
  of quantum computing before fault tolerance.
\newblock {\em Science Advances}, 10(3):eadk4321, 2024.

\bibitem{coecke2011interacting}
Bob Coecke and Ross Duncan.
\newblock Interacting quantum observables: categorical algebra and
  diagrammatics.
\newblock {\em New Journal of Physics}, 13(4):043016, 2011.

\bibitem{sutcliffe2024fast}
Matthew Sutcliffe and Aleks Kissinger.
\newblock Fast classical simulation of quantum circuits via parametric
  rewriting in the zx-calculus.
\newblock {\em arXiv preprint arXiv:2403.06777}, 2024.

\bibitem{koch2023speedy}
Mark Koch, Richie Yeung, and Quanlong Wang.
\newblock Speedy contraction of zx diagrams with triangles via stabiliser
  decompositions.
\newblock {\em arXiv preprint arXiv:2307.01803}, 2023.

\bibitem{schuster2024polynomial}
Thomas Schuster, Chao Yin, Xun Gao, and Norman~Y Yao.
\newblock A polynomial-time classical algorithm for noisy quantum circuits.
\newblock {\em arXiv preprint arXiv:2407.12768}, 2024.

\bibitem{aaronson2004improved}
Scott Aaronson and Daniel Gottesman.
\newblock Improved simulation of stabilizer circuits.
\newblock {\em Physical Review A---Atomic, Molecular, and Optical Physics},
  70(5):052328, 2004.

\bibitem{gottesman1998heisenberg}
Daniel Gottesman.
\newblock The heisenberg representation of quantum computers.
\newblock {\em arXiv preprint quant-ph/9807006}, 1998.

\bibitem{terhal2002classical}
Barbara~M Terhal and David~P DiVincenzo.
\newblock Classical simulation of noninteracting-fermion quantum circuits.
\newblock {\em Physical Review A}, 65(3):032325, 2002.

\bibitem{jozsa2008matchgates}
Richard Jozsa and Akimasa Miyake.
\newblock Matchgates and classical simulation of quantum circuits.
\newblock {\em Proceedings of the Royal Society A: Mathematical, Physical and
  Engineering Sciences}, 464(2100):3089--3106, 2008.

\bibitem{bremner2017achieving}
Michael~J Bremner, Ashley Montanaro, and Dan~J Shepherd.
\newblock Achieving quantum supremacy with sparse and noisy commuting quantum
  computations.
\newblock {\em Quantum}, 1:8, 2017.

\bibitem{bravyi2022simulate}
Sergey Bravyi, David Gosset, and Yinchen Liu.
\newblock How to simulate quantum measurement without computing marginals.
\newblock {\em Physical Review Letters}, 128(22):220503, 2022.

\bibitem{mocherla2023extending}
Avinash Mocherla, Lingling Lao, and Dan~E Browne.
\newblock Extending matchgate simulation methods to universal quantum circuits.
\newblock {\em arXiv preprint arXiv:2302.02654}, 2023.

\bibitem{bravyi2016improved}
Sergey Bravyi and David Gosset.
\newblock Improved classical simulation of quantum circuits dominated by
  clifford gates.
\newblock {\em Physical review letters}, 116(25):250501, 2016.

\bibitem{ermakov2024unified}
Igor Ermakov, Oleg Lychkovskiy, and Tim Byrnes.
\newblock Unified framework for efficiently computable quantum circuits.
\newblock {\em arXiv preprint arXiv:2401.08187}, 2024.

\bibitem{goodman2009symmetry}
Roe Goodman, Nolan~R Wallach, et~al.
\newblock {\em Symmetry, representations, and invariants}, volume 255.
\newblock Springer, 2009.

\bibitem{collins2022weingarten}
Benoit Collins, Sho Matsumoto, and Jonathan Novak.
\newblock The weingarten calculus.
\newblock {\em arXiv preprint arXiv:2109.14890}, 2022.

\bibitem{klimyk1995representations}
AU~Klimyk and N~Ya Vilenkin.
\newblock Representations of lie groups and special functions.
\newblock In {\em Representation Theory and Noncommutative Harmonic Analysis
  II: Homogeneous Spaces, Representations and Special Functions}, pages
  137--259. Springer, 1995.

\bibitem{burgisser2000computational}
Peter B{\"u}rgisser.
\newblock The computational complexity to evaluate representations of general
  linear groups.
\newblock {\em SIAM Journal on Computing}, 30(3):1010--1022, 2000.

\bibitem{albert2019asymptotics}
Victor~V Albert.
\newblock Asymptotics of quantum channels: conserved quantities, an adiabatic
  limit, and matrix product states.
\newblock {\em Quantum}, 3:151, 2019.

\bibitem{anschuetz2023efficient}
Eric~R Anschuetz, Andreas Bauer, Bobak~T Kiani, and Seth Lloyd.
\newblock Efficient classical algorithms for simulating symmetric quantum
  systems.
\newblock {\em Quantum}, 7:1189, 2023.

\bibitem{goh2023lie}
Matthew~L Goh, Martin Larocca, Lukasz Cincio, M~Cerezo, and Fr{\'e}d{\'e}ric
  Sauvage.
\newblock Lie-algebraic classical simulations for variational quantum
  computing.
\newblock {\em arXiv preprint arXiv:2308.01432}, 2023.

\bibitem{nemkov2023fourier}
Nikita~A Nemkov, Evgeniy~O Kiktenko, and Aleksey~K Fedorov.
\newblock Fourier expansion in variational quantum algorithms.
\newblock {\em Physical Review A}, 108(3):032406, 2023.

\bibitem{gonzalez2024pauli}
Guillermo Gonz{\'a}lez-Garc{\'\i}a, J~Ignacio Cirac, and Rahul Trivedi.
\newblock Pauli path simulations of noisy quantum circuits beyond average case.
\newblock {\em arXiv preprint arXiv:2407.16068}, 2024.

\bibitem{angrisani2024classically}
Armando Angrisani, Alexander Schmidhuber, Manuel~S Rudolph, M~Cerezo, Zo{\"e}
  Holmes, and Hsin-Yuan Huang.
\newblock Classically estimating observables of noiseless quantum circuits.
\newblock {\em arXiv preprint arXiv:2409.01706}, 2024.

\end{thebibliography}
\bibliographystyle{unsrt}

\clearpage
\onecolumngrid
\appendix

\section{Mathematical background}
\subsection{Matrix coefficients of irreducible representations}
For any $\lambda\in \hat{G}$ with unitary irreducible representation $V^{\lambda}$, its matrix coefficients with respect to a fixed orthonormal basis $\{\ket{i}\}_{i=1}^{\rm{dim}(\lambda)}$ for the carrier space are given by:
\begin{equation}
	v^{\lambda}_{ij} (g) := \<j| V^{\lambda}(g) |i\>.
\end{equation}
These can be viewed as complex valued functions on the group $v^{\lambda}_{ij}:G\longrightarrow\mathbb{C}$. Moreover they satisfy several properties, particularly Schur orthogonality which we will use extensively.
\begin{lemma}
	Let $G$ be a compact group. Then the following hold:
	\begin{enumerate}
		\item For any $\lambda,\mu \in \hat{G}$
		\[ \int_{G}v^{\lambda}_{ij}(g) (v^{\mu}_{mn}(g))^{*} d\,g = \frac{\delta_{\mu,\lambda}\delta_{im}\delta_{jn}}{\rm{dim}(\lambda)} \]
		\item The dual $\lambda^{*}$ of an irreducible representatation $\lambda$ is irreducible and
		\[ v^{\lambda^{*}}_{mn}(g) =(v^{\lambda}_{mn}(g))^{*}\]
	\end{enumerate}
\end{lemma}

 For general groups ($GL_{n}(\mathbb{C})$, $U(n)$) explicit formulas for the matrix irreps with respect to Gelfand-Tsetlin basis for the carrier representation space can be found in Chapter 18 of \cite{klimyk1995representations}.

\subsubsection{Example: SU(2) - irreps and matrix coefficients}
SU(2) is the group of $2\times 2$ matrices with determinant one and can be characterised by:
\begin{equation}
	SU(2)=\{\left(\begin{array}{cc}\alpha&-\beta^{*}\\ \beta&\alpha^{*}\end{array}\right): |\alpha |^{2}+|\beta |^{2}=1\}.
\end{equation}

It is a compact simply-connected Lie group whose irreducible representations have a simple characterisation.
They are labelled by integer or half-integer numbers $\lambda$ -- these are the heighest weight vectors that correspond to the $d=2\lambda+1$-dimensional irrep. Each $\lambda$-irrep can be realised on the $d$-dimensional vector space of homogeneous polynomials of degree $d-1$ in two real variables $z_1$ and $z_2$ with complex coefficients.  Any element of $H^{d}$ will be a homogeneous polynomial $f(z_1,z_2)=\sum\limits a_{k}z_1^k z_2^{d-1-k}$ for some arbitrary complex coefficients $a_{k}$. $SU(2)$ acts irreducibly on $H^{d}$ via the group action:
\begin{equation}
	g\cdot f(z_1,z_2)=f(g^{T}(z_{1},z_{2}))= f(\alpha z_1+\beta z_2, -\beta^{*} z_1+\alpha^{*} z_2)
\end{equation}
for any group element $g=\left(\begin{array}{cc}\alpha&-\beta^{*}\\ \beta&\alpha^{*}\end{array}\right)$.  This explicitly gives the irreducible representation $V^{\lambda}:SU(2)\longrightarrow GL(H^{d})$ with  $V^{\lambda}(g)[f(z_1,z_2 )]  = f(\alpha z_1 + \beta z_2, -\beta^{*}z_1 +\alpha^{*}z_2)$. This allows to compute the corresponding matrix coefficients $v^{\lambda}_{kk'}$ with respect to the orthonormal basis of homogeneous monomials. In particular:
\begin{equation}
	V^{\lambda}(g)[z_1^k z_2^{d-1-k}] = (\alpha z_1+\beta z_2)^{k}(-\beta^{*} z_1+\alpha^{*} z_2)^{d-1-k}
\end{equation}

From this, one can then obtain explicit formulas for each matrix coefficient through binomial expansion. These are also related to Wigner matrices that express the representations of SO(3) in terms of the Euler rotation angles.

For example, $\lambda = 1$ has a carrier space given by homogeneous polynomials of degree two with (normalised) basis given by $\mathbf{z} = (z_1^2/\sqrt{2}, z_1z_2, z_2^{2}/\sqrt{2})^{T}$. With respect to this basis the 1-irrep matrix of $SU(2)$ takes the form
\begin{equation}
	V^{1}(g)=\left(\begin{array}{ccc}\alpha^{2}&\sqrt{2}\alpha\beta&\beta^{2}\\
		-\sqrt{2}\alpha\beta^{*}& (|\alpha|^{2}-|\beta|^{2}) & \sqrt{2}\alpha^{*}\beta\\
		(\beta^{*})^{2}& -\sqrt{2}\alpha^{*}\beta^{*}& (\alpha^{*})^{2}\end{array}\right)
	\label{Eq:1-irrep}
\end{equation}

\subsubsection{Example: SU(4) acting on 2 qubits}
Here we establish the matrix coefficients for the irreps in the decomposition of $U\otimes U^{*}$ for $U\in SU(4)$. 

We have that for $U\in SU(4)$ the following decomposition into irreps 
\begin{align}
	U\otimes U^{*} = \bf{0} \oplus \bf{1}
\end{align}
where ${\bf{0}}$ labels the trivial irrep and ${\bf{1}}$ labels the adjoint irrep of dimension $d_1 = 15$. 
There are several ways to construct the matrix coefficients depending on how we choose to describe the group elements. For example we can do so using entries of a $4\times 4$ unitary matrix and then every irrep matrix coefficient will be a homogeneous polynomial in the entries of the matrix $U$ \cite{burgisser2000computational}. Another approach is to look at generators of the Lie algebra, which is particularly suitable for describing the adjoint irrep. 

\subsection{Fourier transform on general compact and finite groups} 
\label{app:FTgroup}
We will be concerned with functions that map elements of a compact (or finite) group $G$ to complex values. Consider an irreducible representation $\lambda$  of $G$ with dimension $d_{\lambda}$ so that $\lambda: G \longrightarrow GL_{d_\lambda}(\mathbb{C})$. For a complex-valued function $f:G\longrightarrow \mathbb{C}$ its Fourier transform is given by a function that maps a representation to a complex-valued matrix of dimension equal to the representation's dimension.
\begin{definition}
	For any  compact group, the Fourier transform of $f:G\longrightarrow \mathbb{C}$ with $f\in L^2(G)$ evaluated at the irreducible unitary representation $\lambda$
	\begin{equation}
	\hat{f}(\lambda) =  \int_{g\in G} f(g) \lambda^{*}(g) d\, g,
	\end{equation}
where $d\,  g$ is the uniform Haar measure on the group $G$. A similar definition holds if $G$ is a finite group, where the measure $\int (\cdot) d\, g$  is replaced by the finite (normalised) sum $\frac{1}{|G|} \sum_{g\in G} (\cdot)$ corresponding to a discrete topology with a normalised counting measure.
\end{definition}

We will denote by $\hat{G}$ the set that indexes all irreducible representations of $G$ (up to isomorphism). We have the following properties (similar to the usual Fourier transform for periodic or real functions).
\begin{lemma}
	The convolution of $f_1,f_2 \in L^2(G)$ is defined by
	 \begin{equation}
	 	f_1\star f_2 (g) = \int f_1(h) f_2(gh^{-1} ) d\,h
	 \end{equation}
	Then, the  Fourier transform of the convolution is the product of individual Fourier transforms for any $\lambda\in\hat{G}$
		\begin{equation}
			\widehat{f_1\star f_2}(\lambda) = \hat{f_1}(\lambda) \hat{f_2}(\lambda)
		\end{equation}
\end{lemma}

The Fourier inversion formula holds
\begin{equation}
	f(g) = \sum_{\lambda\in\hat{G}} d_\lambda \< \lambda^{*}(g),  \hat{f}(\lambda)\>_{HS},
\end{equation}
where the Hilbert-Schmidt inner product is given by $ \< \lambda^{*}(g),  \hat{f}(\lambda)\>_{HS} = \tr( \lambda^{*}(g^{-1}) \hat{f}(\lambda))$ and we note that each of $\hat{f}(\lambda)$ and $\lambda(g)$ are in $GL_{d_\lambda}(\mathbb{C})$ and we use the fact that the $\lambda$'s are assumed to be unitary representations.  The inversion formula can be shown as follows

\begin{align}
	\sum_{\lambda\in \hat{G}} d_{\lambda}  \tr( \lambda^{*}(g^{-1}) \hat{f}(\lambda)) &= 	 \sum_{\lambda\in \hat{G} }d_{\lambda} \tr\left(   \lambda^{*}(g^{-1}) \int_{h\in G} f(h) \lambda^{*}(h) \, d\,h \,\right)\\
	& =  \sum_{\lambda\in \hat{G} }d_{\lambda} \tr\left(\int_{h\in G} f(h) \lambda^{*}(g^{-1}h)\  d\,h \,\right)\\
	&  = \int_{h\in G} f(h) \sum_{\lambda\in \hat{G} }d_{\lambda} \tr\left( \lambda^{*}(g^{-1}h)\  \,\right)  d\,h\\
	& =  \int_{h\in G} f(h) \sum_{\lambda\in \hat{G} }d_{\lambda}  \chi_{\lambda}^{*}(g^{-1}h)\  \, d\, h \\
	&= f(g)
\end{align}
where the last equation  arises from orthogonality of characters for irreducible representations in the form of $\sum_{\lambda\in \hat{G}} d_{\lambda} \chi_{\lambda}^{*}(g) = \delta_{eg} $ for $e\in G$, the identity element. 

Similarly one can show that the Parseval identity also holds:
\begin{lemma} {(\bf{Parseval's identity})}
	For $f\in L^2(G)$ we have that
	\begin{equation}
		||f||_{L^2(G)}^2 = \sum_{\lambda\in \hat{G}} d_{\lambda} ||\hat{f}(\lambda)||^2_{GL_{d_\lambda}(\mathbb{C})},
	\end{equation}
	where
	$||f||_{L^2(G)}^2 = \int_{G} |f(g)|^2 d\, g$ and $||\hat{f}(\lambda) ||^2 _{GL_{d_\lambda}(\mathbb{C})} = \< \hat{f}(\lambda) , \hat{f}(\lambda) \>_{HS} = \Tr (\hat{f}(\lambda)\hc\hat{f}(\lambda))$. $d_{\lambda}$ corresponds to the dimension of the $\lambda$-irrep.
\end{lemma}
\begin{proof}
	Expanding out we have
	\begin{align}
		\sum_{\lambda\in \hat{G} }  d_{\lambda}  \tr (\hat{f}(\lambda)\hc \hat{f}(\lambda)) &= \sum_{\lambda\in \hat{G}} d_{\lambda} \int_G f(g) f^{*}(h) \Tr (\lambda^{*}(h))\hc \lambda^{*}(g)) d\, g  \ d\, h\\
		& =\int_G f(g) f^{*}(h)  \sum_{\lambda \in \hat{G}}  d_{\lambda} \tr(\lambda^{*}(h^{-1}g) ) d \,g \ d\,h\\
		& =\int_G f(g) f^{*}(h)  \sum_{\lambda \in \hat{G}}  d_{\lambda} \chi_{\lambda}^{*}(h^{-1}g)d \,g \ d\,h\\
		& = \int_G |f(g)|^2 dg  =||f||_{L^2(G)}^2
	\end{align}
where we used that  $\sum_{\lambda\in \hat{G}} d_{\lambda} \chi_{\lambda}^{*}(h^{-1}g) = \delta_{g,h}$
\end{proof}

\subsection{Extension to real-valued functions on $L^2(G,\mathbb{R})$}

In certain situations,  we will see it convenient to represent Fourier coefficients with respect to real representations.  For example, where the assumptions that lead to classical simulability cannot be met by the decomposition of $L^2(G)$ over matrix irreducible representations over the complex numbers, but can be met over the reals. Many of the results from Peter-Weyl decomposition over the complex numbers  carry over, with some modifications. 

Consider $\lambda_{\mathbb{R}}:G \longrightarrow GL_{d_{\lambda}}(\mathbb{R})$ an irreducible representation over the field of real numbers.  Then, for real-valued functions $f:G \longrightarrow \mathbb{R}$  we can similarly define the real-valued Fourier transform
\begin{align}
	\hat{f}(\lambda_{\mathbb{R}})    = \int_{g\in G}  f(g) \lambda_{\mathbb{R}}(g) d\, g
\end{align}
where here $\hat{f}(\lambda_{\mathbb{R}})$ is a real-valued $d_{\lambda} \times d_{\lambda}$ matrix.

A complex irreducible representation of dimension $d_{\lambda}$ can give rise to real irreducible representations of dimensions $d_{\lambda}$ or $2d_{\lambda}$  This can be determined for instance by using the Frobenius-Schur indicator. In particular for real-valued functions on compact groups $L^2(G, \mathbb{R})$ also decomposes into irreducible components over the reals. A modified form of Parseval's theorem will apply
\begin{align}
	||f||_{L^2(G)}^2 &=  \sum_{\lambda_{\mathbb{R}}}d_{{\lambda}_{\mathbb{R}}}||\hat{f}(\lambda_{\mathbb{R}})||^2_{HS}
\end{align}
where now $d_{\lambda_{\mathbb{R}}}$ can be either $d_{\lambda}$ (when complexification of real irrep is irreducible), $d_{\lambda}/2$(complex type) or $d_{\lambda}/4$ (quaternionic type).  
\subsection{Irreducible Tensor Operators }
\label{app:ITO}

The notion of \emph{irreducible tensor operators} illustrates the structure that emerges from the symmetry lifted to the space of operators $\B(\h)$.

\begin{definition}
	Let $G$ be a compact or finite group with $U$ a unitary representation on $\h$. Then for any irrep $\lambda\in\hat{G}$ the set of operators $\{T^{\lambda}_{k}\}_{k=1}^{\rm{dim}(\lambda)}$ in $\B(\h)$ are called irreducible tensor operators if they transform under the group action as
	\begin{equation}
		U(g) T^{\lambda}_{k} U(g)\hc = \sum_{k'} v^{\lambda}_{k'k}(g) T^{\lambda}_{k'}
		\label{eq:ITO}
	\end{equation}
	for all vector label $k$ and $g \in G$.
\end{definition}

As representation spaces for the group $G$ we have that $\B(\h)\cong \h\otimes\h^{*}$ are isomorphic, where the dual space $\h^{*}$ carries the dual representation $U^{*}$. As Hilbert spaces, $\h$ and $\h^{*}$ are isomorphic so we can view $\B(\h)$ as two copies of $\h$ with the tensor product representation $U\otimes U^{*}$. In this sense, the irreducible decomposition of $U\otimes U^{*}$ will be the same as the decomposition of $\mathcal{U}$. Moreover every irreducible tensor operator set $\{T^{\lambda}_{k}\}_{k=1}^{\rm{dim}(\lambda)}\subset \B(\h)$ gives an orthogonal basis for a $\lambda$-irrep subspace of $\B(\h)$. Indeed, from Schur orthogonality of matrix coefficients it follows that any irreducible tensor operators $\{T^{\lambda}_k\}$ and $\{T^{\mu}_j\}$ transforming under the $\lambda$-irrep and $\mu$-irrep respectively are orthogonal with respect to the Hilbert-Schmidt inner product
\begin{equation}
	\Tr((T^{\lambda}_k)\hc T^{\mu}_{j}) = N_{\lambda}\delta_{\lambda,\mu} \delta_{kj} 
\end{equation}
where $N$ is a normalisation factor independent on the vector component. Typically, for convenience we may assume that $N=1$.

The irreducible tensor operators give the decomposition of $\B(\h)$ into irreducible subspaces such that:
\begin{equation}
	\B(\h) = \bigoplus_{\lambda, {m_{\lambda}}} \{ T^{\lambda}_{k} : 1\leq k \leq \rm{dim}(\lambda)\}
\end{equation}
where the direct sum is over all irrep $\lambda$ in $U\otimes U^{*}$ of multiplicity $m_{\lambda}$. Conversely, any orthogonal basis for each irreducible subspace of $\B(\h)$ must transform according to equation \ref{eq:ITO} so that it forms a set of irreducible tensor operators.

For the case of compact Lie groups, there is an alternative definition of irreducible tensor operators in terms of the generators $\{J_1,..., J_n\}$ of the Lie algebra representation on $\B(\h)$. Then the set of operators $\{T^{\lambda}_{k}\}_{k=1}^{\rm{dim}(\lambda)}$ in $\B(\h)$ transform as:
\begin{equation}
	[J_m, T^{\lambda}_{k}] = \sum_{k'} \<\lambda,k|J_{m}|\lambda,k'\> T^{\lambda}_{k'}
\end{equation}
for all generators $J_{m}$ and all vector labels $k$, where the $\<\lambda,k|J_{m}|\lambda,k'\>$ are given by the matrix coefficients of the generators with respect to an eigenstate basis $|\lambda,k\>$ corresponding to weight vectors.

In order to explicitly construct an ITO basis for $\B(\h)$ we need to specify a basis for the $\lambda$-irrep carrier space with respect to which the matrix coefficients $v^{\lambda}_{k'k}$ is constructed.  Once we have fixed that, we can vectorise the transformation to get an equivalent definition 
\begin{align}
	U(g) \otimes U^{*} (g) |T^{\lambda}_k\>\> =  \sum_{k'} v^{\lambda}_{k'k}(g )|T^{\lambda}_{k'} \>\>.
\end{align}
If $\lambda$ is an irrep in the decomposition of $\mu\otimes \nu$ for $\mu$, $\nu$ irreps then given the basis for the carrier irrep of  $\mu$, $|e^{\mu}_m \>$  and of $\nu$, $|e^{\nu}_n\>$ then a basis for the $\lambda$-irrep subspace is given by
\begin{align}
	|e^{\lambda}_ k\> := \sum_{m,n} \<\mu,m;\nu,n |\lambda, k\> |e^{\mu}_m\> \otimes |e^{\nu}_n\>
	\label{eqn:CG-coupling}
\end{align}
where $\<\mu,m;\nu, n|\lambda,l\>$ are the Clebsch-Gordan coefficients, which are also evaluated with respect to the same fixed basis for the carrier irreducible representations. In our case, we are interested in the decomposition of $U\otimes U^{*}$ into irreducible components. If $U$ itself is irreducible then so is its dual  in which case we can employ Equation~\eqref{eqn:CG-coupling} to construct the ITO. We identify the the basis constructed from coupling via Clebsch-Gordan coefficients and the vectorised form of the operators  $|T^{\lambda}_k\>\> = |e^{\lambda}_k\>$. 
In general we may have that a representation $U$ that splits into irreps as $U  = \oplus_{\mu} \mu $, which can appear with multiplicities. 

\subsubsection{Example: defining  SU(2) representation on single qubit} 
Consider a single qubit system with the defining action of $g\in SU(2) $ given by $U(g) = g$. This corresponds to the $1/2$ - irrep of $SU(2)$ and we have the decomposition  $U\otimes U^* = 1 \oplus 0$. We choose the computational basis states for the qubit given by $|0\> = |1/2,-1/2\>$  and $|1\> = |1/2, 1/2\>$. In terms of Gelfand-Tsetlin patterns these are given respectively by $\left(\begin{array}{ccc} 1 & \hspace{-1cm}& 0 \\ &0	&\end{array}\right)$ and $\left(\begin{array}{ccc} 1 & \hspace{-1cm}& 0 \\ &1	&\end{array}\right)$. 
If we were to label the ITOs via the GT-patterns these would be given by
\begin{align}
	\left(\begin{array}{ccc} 2 & \hspace{-1cm}& 0 \\ &2&\end{array}\right),  \  \ 	 \left(\begin{array}{ccc} 2 & \hspace{-1cm}& 0 \\ &1&\end{array}\right),  \  \	\left(\begin{array}{ccc} 2 & \hspace{-1cm}& 0 \\ &0&\end{array}\right)  \  \  \text{and} \ \
	\left(\begin{array}{ccc} 0 & \hspace{-1cm}& 0 \\ &0&\end{array}\right)  \  \
\end{align}

With respect to the standard computational basis we get the irreducible tensor operators that span these irreducible subspaces in $U\otimes U^{*}$ 
\begin{equation}
	\begin{split}
		T^{1}_{-1}&=\left(\begin{array}{cc} 0 & -1 \\ 0 &0 \end{array} \right) \  \  \
		T^{1}_{0}=\frac{1}{\sqrt{2}}\left(\begin{array}{cc} 1 &0 \\0  &-1\end{array} \right)\\
		T^{1}_{1}&=\left(\begin{array}{cc}  0&0 \\ 1 &0\end{array} \right) \  \  \ T^{0}=\frac{1}{\sqrt{2}}\iden .
	\end{split}
	\label{eqn:ITO}
\end{equation}

These will transform according to the 1-irrep of $SU(2)$ and with respect to the carrier space and matrix irrep coefficients in Equation~\ref{Eq:1-irrep}.
Although the Pauli matrices form an orthogonal basis for the $1$-irrep subspace of $\B(\h)$, they will transform according to \emph{different} matrix coefficients given by $RV^{1} R\hc$ where explicitly have the following transformation between the column vectors of Pauli operators $\boldsymbol{\sigma}^1 = (X,Y,Z)$ on single qubits and the set of ITOs above ${\bf{T}}^1 = (T_{-1}^1, T_0^, T_1^1)$ 
\begin{equation}
	\boldsymbol{\sigma}^{1}=R\mathbf{T^1}:=\left(\begin{array}{ccc}\frac{1}{\sqrt{2}}&0&-\frac{1}{\sqrt{2}}\\0&1&0\\\frac{i}{\sqrt{2}}&0&\frac{i}{\sqrt{2}}\end{array}\right) \mathbf{T^{1}}.
\end{equation}

\subsubsection{Example: ITOs for U(4) action on two qubits}
\label{app:ITOU4}
For the adjoint irrep the generators of a Lie algebra can be used to construct corresponding irreducible tensor operators. We illustrate this for the $SU(4)$. 
For a Lie algebra the adjoint representation is given by the action on itself via the Lie bracket
\begin{align}
	&ad: \mathfrak{g} \longrightarrow GL(\mathfrak{g})\\
	&ad_A[B] = [A,B].
\end{align}

For $\mathfrak{su(4)}$ we take products of Pauli operators $I\otimes \sigma_i$, $\sigma_i\otimes I$ and $\sigma_i\otimes \sigma_j$ where $\boldsymbol{\sigma} = (X,Y,Z)$. We denote their normalised form by $\{T^{1}_{k}\}_{k=1}^{15}$. These form a basis for trace-less hermitian matrices that generate $\mathfrak{su(4)}$ (using a physics convention that maps the Lie algebra to the matrix Lie group $exp:\mathfrak{g} \rightarrow  G$ via $A \rightarrow e^{iA} $).
Lie brackets of pairs of generators expand in terms of the structure constants as
\begin{align}
	[T^1_i, T^1_j] = \sum_k f_{i,j}^{k} T^1_k
\end{align}
And furthermore any $A\in \mathfrak{su(4)}$ expands as $A = \sum_{i} A_i T^{1}_i$ so that
\begin{align}
	[A, T^{1}_j] = \sum_k A_i f^k_{i,j} T^1_k
\end{align}

\subsection{Liouville representation of quantum channels}
We will consider CPTP maps $\mathcal{E}: \B(\h) \longrightarrow \B(\h)$. Suppose that $\B(\h)$ has an operator basis given by $\{B_i\}_{i=1}^{d^2}$  where $d = dim(\H)$.  Furthermore assume this is orthonormal with respect to the Hilbert-Schmidt inner product $\Tr(B_i^{\dagger} B_j) = \delta_{ij}$. Then we define the Liouville representation as the matrix $\boldsymbol{\mathcal{E}}$ corresponding to $\mathcal{E}$ and given by
\begin{align}
	\boldsymbol{\mathcal{E}}_{ij} = \Tr(B_i^{\dagger} \mathcal{E}(B_j))
\end{align}
Similarly as $B_i \in \B(\h)$ form operator basis, then in vectorised form $|B_i\kv =  \sum_{k,k'} B_{kk'}^{(i))} |e_k\>|e_{k'}^{*}\>$, for a basis $\{|e_k\>\}_{k}$ for $\h$ and $\{|e_k^{*}\>\}_{k}$ for $\h^{*}$ and $B_{ki}\in\mathbb{C}$ given by $B_{kk'} =\<e_k| B_i |e_{k'}\>$
\begin{lemma}
	Consider a unitary channel $\U$ given by $\U(\rho) = U\rho U^{\dagger}$ for all $\rho \in \B(\h)$. Then its Liouville matrix  $\boldsymbol{\mathcal{U}}$ is
	\begin{equation}
		\boldsymbol{\mathcal{U}} = U\otimes U^{*}
	\end{equation}
\end{lemma}
\begin{proof} We have $\boldsymbol{\mathcal{U}}_{ij} = \bv B_i| \mathcal{U}(B_j) \kv$. We need to show that $|\mathcal{U}(B_j) \kv = U\otimes U^{*} |B_j\kv$ for any $j$. This follows from  $|\mathcal{U}(B_j) \kv = \sum_{k,k'} \<e_k| UB_j U^{\dagger} |e_{k'}\> | e_k\>|e_{k'}^{*}\> = \sum_{m,n} \<e_m| B_j |e_n\> (\sum_{k,k'} \<e_k|U|e_m\>\<e_n|U^{\dagger}|e_{k'}\> |e_k\> |e_{k'}^{*}  \> = \sum_{m,n} \<e_m| B_j |e_n\>  U|e_m\> U^{*}|e_n\> $
\end{proof}

\subsection{Fourier analysis of $U(g)\otimes U^{*}(g)$ type operators}
In this section we will look at expanding Fourier analysis to operators, and we can do that for example by considering matrix entries of said operator as functions on the group $G$. Suppose $O:G \longrightarrow M_{d\times d} (\mathbb{C})$ with each entry an element of $L^2(G)$.  We have 
\begin{align}
\hat{O}(\lambda) = \int_{g\in G} O(g) \otimes \lambda(g)  d\, g
\end{align}
Similarly we have
\begin{align}
	O(g) = \sum_{\lambda\in \hat{G}} d_\lambda \Tr_{2} (\hat{O}(\lambda) \iden\otimes \lambda(g))
\end{align}
In particular we will be interested in $O (g) = U(g) \otimes U^{*}(g)$ and $\reallywidehat{U\otimes U^{*}} (\lambda)$. Suppose $|B_i\kv$ basis for $\h\otimes \h^{*}$.

For each $\lambda$ we will label $'A'$ the operator space and $'B'$ the irrep space, then we have a basis decomposition (This is like the Schmidt decomposition)
\begin{align}
	\hat{O}(\lambda) = \sum_{i=1}^{max(d^2, d_\lambda^2)} o_i(\lambda) B_i \otimes \lambda_i
\end{align}

\subsection{Connection between ITOs and Fourier transform}
Consider the evolved operator $ O(g)= U(g) O U^{\dagger} (g) = \sum_{\lambda,k,k'} O_{\lambda,k} v^{\lambda}_{kk'}(g) T^{\lambda}_{k'}$ 
The Fourier transform is given by
\begin{equation}
	\hat{O}(\lambda) = \int O(g) \otimes (v^{\lambda} (g) )^{*} d\, g
\end{equation}
Plugging in the thing before we get
\begin{align}
	\hat{O}(\lambda) =\sum_{kk'} O_{\lambda,k} T^{\lambda}_{k'} \otimes \int  v^{\lambda}_{kk'}(g) (v^{\lambda}(g))^{*} d\, g
\end{align}
Orthogonality relations give that
\begin{align}
	\hat{O}(\lambda) =\sum_{kk'} O_{\lambda,k} T^{\lambda}_{k'} \otimes \frac{1}{d_{\lambda}}\delta_{k,i} \delta_{k',j} 
\end{align}
where the $\delta_{k,i} \delta_{k',j}$ is taken to represent a matrix of $d_{\lambda}\times d_{\lambda}$ with non-zero coefficient in the $k,k'$ entry. 
Now let's look at the Parseval's relation in this setting, which we see how it arises from Schur orthogonality of irreducible matrix elements
\begin{align}
	&||O(g)||_{L^2(G)}^2 = \int \Tr(O(g)^{\dagger} O(g)) d\, g   \\
	& = \sum_{\lambda,k,k', \mu, m,m'} \hspace{-0.5cm}O_{\lambda,k} O_{\mu,m}^{*} \Tr((T^{\mu}_{m'}) ^{\dagger}T^{\lambda}_{k'}) \int v^{\lambda}_{kk'}(g) (v^{\mu}_{mm'}(g))^{*} d\, g\\
	& = \sum_{\lambda, k} |O_{\lambda,k}|^2
\end{align}
We can even do the same for inner products 
\begin{align}
&	\<O(g) , \tilde{O}(g) \>_{L^2(G) } = \int \Tr((O(g))^{\dagger} \tilde{O}(g) ) d\, g\\
	& =\hspace{-0.5 cm} \sum_{\lambda,k,k', \mu, m,m'} \hspace{-0.7cm} \tilde{O}_{\lambda,k}  (O_{\mu,m})^{*} \Tr((T^{\mu}_{m'}) ^{\dagger}T^{\lambda}_{k'}) \int v^{\lambda}_{kk'}(g) (v^{\mu}_{mm'}(g))^{*} d\, g\\
	& = \sum_{\lambda,k} O_{\lambda,k }^{*}\tilde{O}_{\lambda,k }
\end{align}

\section{Path Decomposition}
\subsection{Symmetry-adapted path decompositions}
\label{app:pathdecomp}
We describe in detail the notation of the path decompositions through the circuit class $\mathfrak{C}_n$ given by unitaries of the form
\begin{equation}
	C(\g) = U(g_1) W_1 ... U(g_D) W_D.
\end{equation}

In terms of the Liouville representation it takes the following form with respect to the vectorised orthonormal basis  $\{|T^{\lambda}_k\>\> \}$  for the space $\h\otimes \h^{*}$ 
\begin{equation}
	C(\g) \otimes C^{*}(\g) = U(g_1) \otimes U^{*}(g_1)  W_1\otimes W_1^{*}  ... U(g_D) \otimes U^{*}(g_D )W_D\otimes W_D^{*}.
\end{equation}
Furthermore we have the resolution of identity operator in terms of the vectorised basis as $\sum_{\lambda,k} |T^{\lambda}_k \>\>\<\< T^{\lambda}_k|  = \mathbb{I}\otimes \mathbb{I} $, which follows from orthonormality, and where the summation is taken over the full basis.  
\begin{align}
	C(\g) \otimes C^{*}(\g) = \sum_{all} |T^{\lambda_0}_{k_0}\>\>\<\<T^{\lambda_0}_{k_0}|U(g_1) \otimes U^{*}(g_1)  W_1\otimes W_1^{*} |T^{\lambda_1}_{k_1}\>\>\<\<T^{\lambda_1}_{k_1}| ... U(g_D) \otimes U^{*}(g_D )W_D\otimes W_D^{*} |T^{\lambda_D}_{k_D}\>\>\<\<T^{\lambda_D}_{k_D}|.
\end{align}
Since the basis transforms irreducibly according to the representation $U(\cdot)$ we have that  $\<\<T^{\lambda_0}_{k_0} |  U(g_1) \otimes U^{*}(g_1) = \sum_{k_0'} v^{\lambda_0}_{k_0',k_0} (g_1^{-1}) \<\< T^{\lambda_0}_{k_0'}|$, where the summation ranges only over the vector component $k_0'$ of the irrep $\lambda_0$, so it is restricted to the $\lambda_0$ -block of dimension $d_{\lambda_0}$ equal to the dimension of the irrep.  Using this, we have that 
\begin{align}
	C(\g) \otimes C^{*}(\g) = \sum_{all} v^{\lambda_0}_{k_0',k_0}(g^{-1}_1)... v^{\lambda_{D-1}}_{k'_{D-1},k_D}(g_D^{-1}) |T^{\lambda_0}_{k_0}\>\>\<\<T^{\lambda_0}_{k_0'}| W_1\otimes W_1^{*} |T^{\lambda_1}_{k_1}\>\> ...
	\<\< T^{\lambda_{D-1}}_{k_{D-1}'}|W_D\otimes W_D^{*} |T^{\lambda_D}_{k_D}\>\>\<\<T^{\lambda_D}_{k_D}|.
\end{align}
Using the compact notation we introduce in the main text given by $v^{\boldsymbol{\lambda}}_{{\bf{k'}, {\bf{k}}}} (\g^{-1}) :=  v^{\lambda_0}_{k_0',k_0}(g^{-1}_1)... v^{\lambda_{D-1}}_{k'_{D-1},k_D}(g_D^{-1})$, where the  $\g$  is an element of the direct product $\mathbb{G}= G^{\times D}$,  we can see that $v^{\boldsymbol{\lambda}}_{\bf{k'}, \bf{k}}(\g)$ represents matrix coefficients of the irreducible representation of $\mathbb{G}$, as these are given precisely by tensor products of irreps of $G$.  Re-writing the above we get
\begin{align}
	C(\g) \otimes C^{*}(\g) = \sum_{all}   v^{\boldsymbol{\lambda}}_{{\bf{k'}, {\bf{k}}}} (\g^{-1}) \,  \<\<T^{\lambda_0}_{k_0'}| W_1\otimes W_1^{*} |T^{\lambda_1}_{k_1}\>\>   \<\<T^{\lambda_1}_{k_1'}| W_2\otimes W_2^{*} |T^{\lambda_2}_{k_2}\>\> ...
	\<\< T^{\lambda_{D-1}}_{k_{D-1}'}|W_D\otimes W_D^{*} |T^{\lambda_D}_{k_D}\>\>  \, \,|T^{\lambda_0}_{k_0}\>\>\<\<T^{\lambda_D}_{k_D}|.
\end{align}

\section{Main framework for operator-truncated classical simulation algorithms}

\subsection{General description - truncated operator paths} 
\label{sec:results}
\begin{definition}
	Given an operator basis $\B = \{B_i\}$ for $\B(\h)$  and an observable $O \in \B(\h)$  the rank of $O$ with respect to the basis $\B$ is given by the size of the set
	\begin{equation}
		r(O) =  |\{ B_i \in \B : Tr(B_i\hc O) \neq 0 \}|
	\end{equation}
\end{definition}

\begin{definition}{\bf{Operator truncated classical algorithm}}
	Given an operator basis $\B = \{B_i\}$ for $\B(\h)$ an operator truncated weak classical algorithm $\A$ takes as input an observable $O$,  a unitary circuit $C = V_1 V_2 ... V_D$  and a truncation function $\texttt{\textup{truncate}}$ to output the truncated operator
	\begin{align}
		O^{\T} = \sum_{ \texttt{\textup{truncate}}[{\bf{i}}] = True} O_{i_0} v^{(1)}_{i_0, i_1} ... v_{i_{D-1}, i_D}^{(D) } B_{i_D}
 	\end{align}
	where ${\bf{i}}$ corresponds to a particular path decomposition $B_{i_1} \rightarrow B_{i_2} ... \rightarrow B_{i_{D-1}}$.  
\end{definition}

\subsection{Noise models}
\label{app:noise}
In the main text, we assume an error model that is diagonal in the symmetry-adapted basis. Specifically that the Liouville channel $\mathbf{E}$ satisfies
\begin{align}
	\mathbf{E}|T^{\lambda}_{k} \>\>  = e_{\lambda,k} |T^{\lambda}_{k} \>\>
\end{align}
and furthermore for any basis operators $T^{\lambda_0}_0$ that transform trivially (i.e $U(g) T^{\lambda_0}_0  U\hc(g) = T^{\lambda_0}_0$)  then $e_{\lambda_0, 0} = 1$ while any other error eigenvalues are strictly less than 1. One consequence of this assumption is that the noise model will be unital and will have no effect on the trivial subspace that $\Pi_0$ projects into, but at the same time will penalise any occurrence of a basis operator that transforms non-trivially within a multi-path decomposition by a factor of at most $(1-\gamma)$.

For a general error channel $\N$ we can consider a dephasing map along a particular basis, in this case along $\{T^{\lambda}_k\}$. In the Liouville decomposition the resulting map $\bar{\N}$  is given by
\begin{align}
	\bar{\mathbf{N}} = \sum  |T^{\lambda}_k\>\>\<\< T^{\lambda}_k|  \< \< T^{\lambda}_k|\mathbf{N} |T^{\lambda}_k\>\>.
\end{align}
For example, if we were considering a Pauli basis, the above dephasing map corresponds to twirling the general error channel into a Pauli channel. The dephased map will be diagonal in the operator basis, by construction. Furthermore if we can ensure that it is a close approximation of $\N$ such that $|| \N- \bar{\N}||_{\Diamond}\leq \epsilon'$, then we have that the noisy circuits satisfy 
\begin{align}
	 || \N \U_1 (g_1)\mathcal{W}_1  ... \N \U_D(g_D)\mathcal{W}_D -  \bar{\N} \U_1(g_1) \mathcal{W}_1 ... \bar{\N} \U_D(g_D)\mathcal{W}_D||_{\Diamond}\leq \epsilon' D
\end{align}
from sub-multiplicativity of the diamond norm.  This leads to 
\begin{align}
	&||\tilde{o}(\g) - \tilde{o}^{\T}_{\bar{\N}}(\g)||_{L^2(\mathbb{G})} \leq \epsilon' D + ||\tilde{o}_{\bar{\N}}(\g) - \tilde{o}^{\T}_{\bar{\N}}(\g)||_{L^2(\mathbb{G})} 
\end{align}
whereby $\tilde{o}(\g)$ we denote expectation values with respect to a noise model given by $\N$ and by $\tilde{o}_{\bar{\N}}(\g)$ we denote the expectation values given by using the dephased map $\bar{\N}$.

\subsection{Framework for group-parametrised noisy circuit ensembles} 
We consider families of unitaries $C(\g)$ acting on the $n$-qudit system $\h_n$, indexed by elements of a compact or finite group $g_i \in G$ 
\begin{align}
	C(\g) = U_1(g_1) W_1 U_2(g_2) W_2 \, ...  \, U_D(g_D) W_D
\end{align}
where $\g: = (g_1,..., g_D)$ is an element of a direct group product $\mathbb{G} = G^{\times D}$,  $U_i$ are representations of the group $G$ on $\h_n$.  We will refer to unitaries that have this fixed structure as "circuits", although technically speaking we are not necessarily discussing qubit systems nor assuming that the individual unitaries $W_i$ and $U_i$ correspond to primitive gates.

The results we have rely on one or more assumptions related to the circuit structure, decomposition basis or noise model. \emph{We will make it explicit on which of the assumptions we rely on for each result but emphasise that we don't always assume all of them are met, as certain situation require weaker conditions.} 

The first type of assumption helps to identify convenient choices of operator basis with respect to which we can construct a truncated set of the full path decomposition that can then be evaluated. 

{\bf{\emph{Assumption 0} }}\emph{ Each $U_i$ forms a representation of $G$ and the maximal dimension of inequivalent irreps $\lambda_i$ in $U_i\otimes U_i^{*}$ that are also allowed within at least one truncated path is given by $d_{max} = poly(n)$. Furthermore, for each such irrep $\lambda_i$ we can efficiently compute matrix coefficients $v^{\lambda_i}_{kk'}(g)$ for any $g\in G$. 
}

The second type of assumption involves only the $W$-type operators - and are constructed in such a way as to constrain the complexity growth due to these sub-circuits when we essentially remove the parameterised $U_i$ unitaries.  In a way, we can think of the $W$'s as being the ``free" unitaries although we make no assumptions that they form an efficient classical sub-theory in any strict sense. In fact, the assumption is much weaker because we only require efficient evaluation on the restricted subspaces arising from paths allowed in the truncation set. However, with appropriate choice of basis, when the $W$-type unitaries do belong to classes that form efficient classical sub-theories such as Clifford or matchgate unitaries these kind of assumptions are met a-priori and do not impose additional constraints.  
 
{\bf{\emph{Assumption 1} }}
{\bf{\emph{(Controlled tree growth)} }} \emph{For all $i\in \{1, D\}$, each unitary $W_i$ maps any given basis operator $|T^{\lambda}_{k}\>\>$ entirely into the trivial subspace or is orthogonal to the trivial subspace.  Equivalently, 
	\begin{equation}
		\Pi_0 W_{i}\otimes W_i^{*} |T^{\lambda}_k \>\> = \begin{cases}
			0 &  \ \  or \\
			W_{i} \otimes W_{i} ^{*}|T^{\lambda}_k \>\> .
		\end{cases}
	\end{equation} 
	{\bf{\emph{(Path contractions and evaluation)} }}The number basis operators $(\mu, m)$ for which the contraction is non-zero $\<\< T^{\mu}_m|W_{i+J}\otimes W_{i+J}^{*} W_{i+J-1}\otimes W_{i+J-1}^{*} ... W_i\otimes W_i^{*}|T^{\lambda}_k\>\> \neq 0$ is at most $s = poly(n)$ and can be evaluated efficiently in time at most $cost(\texttt{\textup{contract}}) = poly(n)$.}

{\bf{\emph{Assumption 2 (Noise penalisation)}}}. \emph{For the noise model, we assume every unitary is subject to a channel $\E$ with spectral gap $\gamma$ and that the only fixed or rotating points correspond to the trivial irreps which means that the eigenvalues $|e_{\lambda,k}| <(1-\gamma)$ whenever $\lambda$ is not equivalent to a trivial irrep.}

\begin{theorem} {\bf{(Main result for mean values)}} 
	\label{thm:mainmean}
	Let $O$ be an observable on an $n$-qudit system $\h_n$. Denote its expectation value with respect to unitaries $C(\g) \in \mathfrak{C}_n$ and input state $|0^n\> \in \h_n$  by $o(\g) =  Tr( O \, C(\g) |0^n\>\<0^n| C^{\dagger}(\g))$  and by $\tilde{o}(\g)$ the corresponding noisy expectation value with noise model satisfying Assumption 2.  Suppose that the rank of $O$ with respect to a suitable basis is $r(O) = poly(n)$. Then, if with respect to that basis the Assumption 1 is satisfied, then  there is a classical algorithm $\A$ with a truncation rule given by a controllable parameter $L$ that outputs $\tilde{o}^{\T}(\g)$ such that   
		\begin{align}
			||\tilde{o}(\g) - \tilde{o}^{\T}(\g) ||_{L^2(\mathbb{G})}  \leq e^{-\gamma L} ||O||_{HS},
		\end{align}
		and has run-time $\O(r(O)(s\,d_{max})^L cost(\texttt{\textup{evaluate}})) $, where $d_{max}$ is the maximal dimension of the irreps in the decomposition of $U\otimes U^{*}$. 
\end{theorem}
For example, if  the $U_i$ are $k$-local  (i.e acts non-trivially on $k$ qudits) and $s = poly(n)$ then the run-time is $poly(n)^L \, d^{2kL}$ where $d$ is the qudit dimension (e.g $d=2$ for qubits).
\begin{proof}
	Denote by $\{T^{\lambda}_k\}$ the orthonormal irreducible tensor operator basis. First we explain how the noise model induces an effective truncation that allows us to apply the Algorithm 1 for sparse truncated Fourier expansions, starting from each of the basis operators in the support of $O$ as root that we propagate in the Heisenberg picture.

(Polynomial - number of tree-like structures) The observable decomposes into a sum of $r(O)$ basis operators as $O = \sum_{ \lambda_0,k_0} O_{\lambda_0,k_0} T^{\lambda_0}_{k_0}$ where the coefficients are $O_{\lambda_0,k_0} = \tr[(T^{\lambda_0}_{k_0})\hc O]$.  This  means that if we apply an operator truncated algorithm to each of the basis operators in the decomposition of $O$, then we will have an $\O(r(O)) = poly(n)$ overhead.

(Noise model induced truncation) For a noise channel  $\E$ diagonal in the ITO basis and acting after each $U(g)$ with a spectral gap $\gamma$. We denote the channel's eigenvalues by $e_{\lambda,k}$ such that in the Liouville representation $\boldsymbol{\E}|T^{\lambda}_k\>\> = e_{\lambda,k} |T^{\lambda}_k\>\>$ and we have $|e_{\lambda,k}| \le (1-\gamma)$ whenever $\lambda$ is not equivalent to the trivial irrep. 	

Suppose that for each valid path $(\boldsymbol{\lambda, {\bf{k}}})$ we look at the number of locations in a given path $(\lambda_0,..., \lambda_D)$ for which at most $L$ of them correspond to irreps that are not-equivalent to the trivial one (i.e $\lambda_i\neq 0$). In other words, every path in the truncation set can branch out in at most $L$ locations.   Therefore the full truncation set $\T$  is given by

\begin{align}
	\T:= \{(\boldsymbol{\lambda, {\bf{k}},{\bf{k'}}}) : \texttt{validate}[(\boldsymbol{\lambda, {\bf{k}},{\bf{k'}}}) ]= True   \& \ \boldsymbol{\lambda}=(\lambda_0,.., \lambda_D) \  {\rm{satisfies}} \   \#|\{ \lambda_i : \lambda_i = 0,  i\in \{1,D\} \}| \leq L \}.
\end{align}

and has the property that any valid path outside of it $(\boldsymbol{\lambda, {\bf{k}}}) \notin \T$ must obey  $|e_{\lambda_0,k_0}... e_{\lambda_D,k_D}| \leq (1-\gamma)^L$, which follows from the fact that the channel has a spectral gap and that fixed and rotating points lie within the trivial subspace.

\emph{Step 1 - Operator-truncated algorithm to evaluate the truncated noisy expectation value $\tilde{o}^{\T} (\g)$}

For each $T^{\lambda_0}_{k_0}$ in the decomposition of the target observable $O$ we construct the truncated noisy approximation of $\<0^n| C\hc(\g) T^{\lambda_0}_{k_0}C(\g)|0^n\>$  with respect to the paths in $\T$.  In particular the tree generated by Algorithm 1 $\texttt{GTT}(L,T^{\lambda_0}_{k_0}, W_1,..., W_D)$ will give the sparse Fourier path expansion truncated to include all the Fourier modes $\boldsymbol{\lambda}$ with at most $L$ non-trivial irreps.


(Efficient path evaluation) Coefficients of all the unitaries $W_i$ are given by $Tr((T^{\mu}_m) \hc W_i \,T^{\lambda}_k \, W_i\hc) $ can be computed and stored efficiently. This assumption can be restricted to the basis operators that are retained in the valid truncation set $\T$.  We used $s =poly(n)$ to denote the maximal rank $s = max_{\lambda,k} max_{ i }r(W_i T^{\lambda}_k W_i\hc)$ over all input basis for all $W$'s. The assumption that $s$ is polynomial in the number of qubits ensures a level of control on the spread of the support as we propagate each $T^{\lambda_0}_{k_0}$ in the Heisenberg picture.

(Efficient irreducible matrix coefficient evaluation) There exists an efficient algorithm to compute matrix coefficients of irreducible representation $v^{\lambda} _{kk'}(\cdot)$ for all the $\lambda$-irreps in the decomposition of $U_i\otimes U^{*}_i$ with dimension at most $d_{max}$. For a $k$ - local unitary $U$ this includes irreps with dimensions of at most $d_{max} \leq d^{2k} - 1$ (where $d$ denotes the qudit dimension)   
Ultimately, Assumption 1 also implies we can efficiently contract parts of the paths so that we obtain a tree of height $L$, with the degree of each node at most  $sd_{max}$.

\emph{Step 2 Proof of the approximation error bound } Now, let's show the approximation error bound in the truncated expectation value $\tilde{o}^{\T}(\g)$, where only the contributions from noisy paths in $\T$ are taken into account. The approximation bounds rely only on (i) the fact that the circuit class is parametrised by group elements $\g\in \mathbb{G}$ and (ii) \emph{Assumption 3} and its consequence on the truncation rule defining the set $\T$. We have from Parseval's theorem that
	\begin{align}
		||\tilde{o}(\g) - \tilde{o}^{\T} (\g) ||_{L^2(\mathbb{G})}^2  &= \sum_{\boldsymbol{\lambda}} d_{\boldsymbol{\lambda}} ||\widehat{\tilde{o} - \tilde{o}^{\T}} \  \  (\boldsymbol{\lambda}) ||_{HS}^2\\ 
		& = \sum_{(\boldsymbol{\lambda},{\bf{k}}) \notin \T} d_{\boldsymbol{\lambda}}  |e_{\lambda_0,k_0}... e_{\lambda_D, k_D}|^2 |[\widehat{ o-o^{\T}}(\boldsymbol{\lambda}) ]_{{\bf{k'}}, {\bf{k}}}|^2 \\
		&\leq (1-\gamma)^{2L}  \sum_{(\boldsymbol{\lambda},{\bf{k}}) \notin \T} d_{\boldsymbol{\lambda}}   |[\widehat{ o-o^{\T}}(\boldsymbol{\lambda}) ]_{{\bf{k'}}, {\bf{k}}}|^2 \\
		&\leq (1-\gamma)^{2L}  \sum_{(\boldsymbol{\lambda},{\bf{k}}) } d_{\boldsymbol{\lambda}}   |[\widehat{ o-o^{\T}}(\boldsymbol{\lambda}) ]_{{\bf{k'}}, {\bf{k}}}|^2  \stackrel{\rm{Parseval}}{=} (1-\gamma)^{2L} ||o(\g)||_{L^2(G)}^2\\
		&\leq e^{-2\gamma L} ||O||_{HS}^2.
	\end{align}
	In the above, we used that the noisy Fourier coefficient contract the true Fourier coefficient as $[\widehat{\tilde{o}- \tilde{o}^{\T}} \  \  (\boldsymbol{\lambda})]_{{\bf{k'}}, {\bf{k}}}  = e_{\lambda_0,k_0} ... e_{\lambda_D,k_D}[\widehat{o-o^{\T}} \  \  (\boldsymbol{\lambda})]_{{\bf{k'}}, {\bf{k}}}$.  Then, the first inequality comes from the truncation rule, the second inequality comes from including the positive contributions from the  coefficients involving paths within the truncation set. The latter allows us to relate it back to an average norm $||o(\g)||_{L^2(\mathbb{G})}$ via Parseval's. Finally, the last inquality comes from $(1-\gamma)^{2L} \leq e^{-\gamma L} $ and that $||o(\g)||_{L^2(\mathbb{G})}^2 = \int | \<0^n| C(\g)\hc OC(\g)  |0^n\>|^2 d\,\g \leq ||O||_{HS}^2 $ arises from the Cauchy-Schwarz inequality.

	\emph{Step 3 -  Complexity analysis of OTA}
	The total number of paths in each tree will be $(sd_{max})^L$, as each tree has height $L$ and degree at most $sd_{max}$. Along each path at most $L$ contractions are performed, the cost of which are at most $cost(\texttt{contract})L$ for each path. A single path evaluation within the contracted tree will take $cost(\texttt{evaluate})$, which includes the cost of evaluating the matrix irrep coefficients at specified group parameter values and the cost of all contractions within each path.  Overall, evaluating the forest (i.e each of the individual trees for all root operator basis) will produce $\tilde{o}^{\T}(\g )$ as described in Algorithm 1 in main text, and this will require $r(O) (sd_{max})^{L}  cost(\texttt{evaluate})$.
\end{proof}

The following assumption on the efficient contraction of the $W$-unitaries restricted to the trivial isotypical component allows for an improved simulation of the sampling task.

 {{\bf{\emph{Assumption 3 (Efficient contractions)} }}  There exists and $M$ such that for all $i \in \{1, D-M\}$ the following operator
		\begin{align}
			\Pi_0 W_i\otimes W_i^{*} \Pi_0 ... W_{i+M}\otimes W_{i+M}^{*} \Pi_0 
		\end{align}
		can be efficiently contracted to a support containing a polynomial number of  terms of the form $|T^{\lambda_0} \>\> \<\< T^{\lambda_0'}|$, where $\lambda_0$ and $\lambda_0'$ correspond to irreducible subspaces equivalent to the trivial irrep. The time complexity of the contraction is $cost(\texttt{contract-full})$ and $cost(\texttt{rank}) = poly(n)$ denotes the maximal number of terms in the support. 
	}

\begin{theorem} {\bf{(Main result for sampling task)}}
	\label{thm:mainsampling}
	Suppose that the marginals $\sum_{x_{i_1},... x_{i_j}} \<{\bf{x}}|T^{\lambda}_{k} |\bf{x}\>$ can be efficiently evaluated for all basis operators. 
	Given Assumptions $0-3$,  there exists an explicit classical algorithm 
	which outputs a function $\tilde{p}^{\T}({\bf{x}})$ with $l_1$-approximation error $ ||\tilde{p}- \tilde{p}^{\T} ||_1 \leq  (1-\gamma)^{L}$, provided that $\tilde{p} -\tilde{p}^{\T}$ anti-concentrates. The run-time of the algorithm is given by at most $cost(\texttt{\textup{evaluate}}) (sd_{max})^{2L} D cost(\texttt{\textup{rank}})  \,cost(\texttt{\textup{contract-full}})  = \O(D) poly(n) \O((sd_{max})^{2L} )$.
\end{theorem}
\begin{proof}
	From the noise model assumption we have that it induces the truncation rule (see proof of Theorem \ref{thm:mainmean}  for the discussion of how it arises) 
	\begin{align}
		\T:= \{(\boldsymbol{\lambda, {\bf{k}},{\bf{k'}}}) : \texttt{validate}[(\boldsymbol{\lambda, {\bf{k}},{\bf{k'}}}) ]= True   \& \ \boldsymbol{\lambda}=(\lambda_0,.., \lambda_D) \  {\rm{satisfies}} \   \#|\{ \lambda_i : \lambda_i = 0,  i\in \{1,D\} \}| \leq L \}.
	\end{align}

	{\emph{Step 1  - Operator Truncated Classical Algorithm that outputs $\tilde{p}^{T}(\bf{x}, \g)$ } } This follows the Algorithm 2  in the main text. The only modification in the noisy setting is that the path evaluation will incorporate the noise parameters. 
	
	From the assumption on efficient contractions of  $W$-type operator we pick the truncation parameter $L$  to ensure that  $\Pi_0W_i\otimes W_i^{*} \Pi_0 ... W_{i+len}\otimes W_{i+len}^{*} \Pi_0 $ can be efficiently contracted  to a polynomial number of terms where $len = [D/L]$.
	  	
	  We run the following procedure for each $i \in \{1, ..., D-len\}$ 
\begin{enumerate}[(i)]
		\item Contract  $\Pi_0 W_i\otimes W_i^{*} \Pi_0  W_{i+1} \otimes W_{i+1}^{*}  \, ...  \, W_{i+len} \otimes W_{i+len}^{*} \Pi_0$ to the $r_i = poly(n)$ terms of \emph{Assumption 3}. This requires time complexity  of  $cost(\texttt{contract-full}) $
		\item For each term $|T^{\mu_0}\>\>\<\< T^{\mu_0'}|$ in the contraction set previously determined we consider the object
		\begin{align}
			U_1 \otimes U_1^{*} W_1 \otimes W_1^{*}... U_{i-1}W_{i-1}\otimes W_{i-1}^{*} |T^{\mu_0}\>\>\<\< T^{\mu_0'}| W_{i+len+1}\otimes W^{*}_{i+len+1}. ...  U_D\otimes U_D^{*}W_D \otimes W_{D}^{*}
		\end{align}
		\item Propagate $|T^{\mu_0}\>\>$ backwards with a tree of height at most $L$ first and prune branches that have already been incorporated for previous values of $i$. This condition checks that  any consecutive applications of $\Pi_0$ are applied for less $len$ times. If that is not true then the branch gets pruned (removed).  This ensures we do not double-count contributions for these paths. Then we propagate  $\<\<T^{\mu_0'}|$ forwards with tree constructions of height at most $L$ in either directions. The maximal number of paths will be $(sd_{max})^{2L}$, which is ensured by \emph{Assumption 2}. This leads to bi-directional tree constructions that are given by $\texttt{GTT}(L, T^{\lambda_0},  W_{i-1}\hc, ... W_1\hc)$ and $\texttt{GTT}( T^{\lambda_0'},  W_{i+M+1}, ... W_D)$. 
		\item  The final propagated basis on either ends of the bi-directional trees can be efficiently contracted against both computational basis elements and marginals, so that each valid truncated path can be efficiently evaluated. Evaluate contribution to the $p^{\T}({\bf{x}})$ for all the paths with at most $L$ branches along a path.  The total cost will be at most $(sd_{max})^{2L} L cost(\texttt{contract}) cost(\texttt{matrix-coeff}) cost(\texttt{evaluate}) $, where the $cost(\texttt{evaluate})$ counts the complexity of contracting the ITO basis operators with the computational basis.  
\end{enumerate}

	In total there will be at most $ D\, {\rm{max}}_{i} r_i =  D\, cost(\texttt{rank})$ bi-directional trees to be constructed and evaluated, with each having at most $(sd_{max})^{2L}$ paths.  The cost to evaluate each path is given by $cost( \texttt{evaluate})$ which is the time complexity of the $\texttt{evaluate}$ function in the main text. This incorporates the at most $L$ contractions within each path involving W-type unitaries as in \emph{Assumption 2}, the evaluation of matrix coefficients  $cost(\texttt{matrix-coeff} )$ and the evaluation of basis operators with respect to the computational basis.
Overall we obtain a time complexity of $ cost(\texttt{contract-full}) D\, cost(\texttt{rank})  (sd_{max})^{2L} cost(\texttt{evaluate})  = \O(D)poly(n) \O(L (sd_{max})^{2L})$.
	
{\emph{Step 2 - Approximation bounds} } For any $\g$ the noiseless probability distribution is given by
\begin{align}
	p({\bf{x}}, \g) = \<\< {\bf{x}}|  C(\g) \otimes C^{*}(\g) |{\bf{0}}\>\> .
\end{align}

The truncated probability distribution will be similarly defined by including all the Fourier terms that satisfy the truncation condition and are validated. Specifically 
\begin{align}
	p^{\T}({\bf{x}}, \g) =  \sum_{\boldsymbol{\lambda} \in \T} d_{\boldsymbol{\lambda}}\<\< {\bf{x}}|  \Tr_{2}[ \widehat{C\otimes C^{*}} \  (\boldsymbol{\lambda}) \mathbb{I}\otimes \boldsymbol{\lambda}^{*}(\g^{-1}) ] |{\bf{0}}\>\> 
\end{align}
and the noisy truncated expectation value will be given by 
\begin{align}
	\tilde{p}^{\T}({\bf{x}}, \g) =  \sum_{(\boldsymbol{\lambda},{\bf{k}}) \in \T} d_{\boldsymbol{\lambda}} e_{\lambda_0,k_0} ... e_{\lambda_D, k_D}  v^{\boldsymbol{\lambda}}_{\bf{k'},\bf{k}} (\g^{-1})   \<\< {\bf{x}}|[\widehat{C\otimes C^{*}} \  (\boldsymbol{\lambda}) ]_{\bf{k'}, {\bf{k}}}   |{\bf{0}}\>\> 
\end{align}

We have the following bounds with respect to the norm $l^2(L^2(\mathbb{G}))$ 
\begin{align}
	|| \tilde{p} - \tilde{p}^{\T} ||_{l^2(L^2(\mathbb{G}))}^2 &= \int  ||\tilde{p} -\tilde{p}^{\T} ||^2_{l^2} d\g = \sum_{{\bf{x}}} \int |\tilde{p}({\bf{x}}, \g) - \tilde{p}^{\T}({\bf{x}}, \g) |^2  d\, \g\\
	& \stackrel{\rm{Parseval}}{=} \sum_{\bf{x}}  \sum_{\bf{\boldsymbol{\lambda}}}    d_{\boldsymbol{\lambda}} || \widehat{\tilde{p} - \tilde{p}^{\T} } \,({\bf{x}})  \, \,   (\boldsymbol{\lambda}) ||_{HS}^2		\\	
	& = \sum_{ \bf{x}} \sum_{ \boldsymbol{\lambda} }   d_{\boldsymbol{\lambda}}  \sum_{ \bf{k,k'} }  |e_{\lambda_0, k_0} ... e_{\lambda_D, k_D} |^2  | [\widehat{p -  p^{\T}} (\bf{x}) (\boldsymbol{\lambda})]_{\bf{k'}, \bf{k}} |^2 \\
	& =  \sum_{ \bf{x}} \sum_{ (\boldsymbol{\lambda}, \bf{k})\notin \T }   d_{\boldsymbol{\lambda}}   |e_{\lambda_0, k_0} ... e_{\lambda_D, k_D} |^2  | [\widehat{p -  p^{\T}} (\bf{x}) (\boldsymbol{\lambda})]_{\bf{k'}, \bf{k}} |^2 \\
	& \leq  (1- \gamma)^{2L}   \sum_{ \bf{x}} \sum_{ (\boldsymbol{\lambda}, \bf{k})\notin \T }   d_{\boldsymbol{\lambda}}    | [\widehat{p -  p^{\T}} (\bf{x}) (\boldsymbol{\lambda})]_{\bf{k'}, \bf{k}} |^2 \\
	& \stackrel{\rm{Parseval}}{=}     (1- \gamma)^{2L} ||p - p^{\T}||_{l^2(L^2(\mathbb{G}))}^2 
\end{align}
However, the sampling to compute reduction involves a $l_1$ norm, so in order to connect with the above bounds we invoke Cauchy-Schwartz so that 
\begin{align}
	||\tilde{p} - \tilde{p}^{\T} ||_{l_1(L^2(\mathbb{G}))} \leq  \sqrt{d^n} ||\tilde{p} - \tilde{p}^{\T} ||_{l_2 (L^2(\mathbb{G}))},  
\end{align}
where we recall that $d$ corresponds to the qudit dimension (e.g $d=2$ for qubits). Combining the above relations we obtain that
\begin{align}
	||\tilde{p} - \tilde{p}^{\T} ||_{l_1(L^2(\mathbb{G}))} \leq \sqrt{d^n} ||\tilde{p} - \tilde{p}^{\T} ||_{l_2 (L^2(\mathbb{G}))} \leq  (1-\gamma)^L  \sqrt{d^n} ||p - p^{\T} ||_{l_2 (L^2(\mathbb{G}))}
\end{align}
We've assumed that the residual quasi-probability distribution anti-concentrates, which means that
\begin{align}
	\sqrt{d^n} ||p - p^{\T} ||_{l_2 (L^2(\mathbb{G}))} = O(1).
\end{align}
Therefore we obtain an approximation error that decays exponentially with $L$ as $||\tilde{p}- \tilde{p}^{T} || _{l_2 (L^2(\mathbb{G}))} \leq O(1) e^{-\gamma L}$.
\end{proof}

The following lemma from \cite{bremner2017achieving} shows that  sampling can be reduced to computing probability distributions and its marginals.
\begin{lemma} {\bf{Sampling to computing reduction}} Let  $\tilde{p}$ be a probability distribution on $\mathbb{Z}_d^{\times n}$.  Suppose that there exists an oracle computing the function $\bar{q}:  \mathbb{Z}_d^{\times n} \rightarrow \mathbb{R}$  and its marginals $\bar{q}(x_1,...,x_k) = \sum_{x_{k+1},...x_{n} \in \mathbb{Z}_d} \bar{q}(x_1,.., x_n)$ such that $||\tilde{p}-\bar{q}||_1\leq \epsilon$. If $\sum_{x} \bar{q}(x) =1$ then there is an algorithm that uses $O(dn)$ calls to the oracle to sample from a probability distribution $q$ with $||\tilde{p}-q||_1 \leq 2\epsilon$. 
	\label{lemma:samplingtocompute}
\end{lemma}
\begin{corollary}
		\label{thm:sampling-corrolary}
			Suppose the 
		probability distribution marginals $\sum_{x_{i_1},... x_{i_j}} \<{\bf{x}}|T^{\lambda}_{k} |\bf{x}\>$ can be efficiently evaluated for all basis operators. 
		Given Assumptions 0-3 then there exists an explicit classical algorithm which outputs a function $\tilde{p}^{\T}({\bf{x}})$ with $l_1$-approximation error $ ||\tilde{p}- \tilde{p}^{\T} ||_1 \leq \epsilon$, provided that the noiseless residual (quasi)-distribution $p -p^{\T}$ anti-concentrates and with runtime either
		\begin{enumerate}
			\item  constant-degree polynomial scaling in $n$ and exponential scaling in the inverse spectral gap $\gamma^{-1}$ with  $\O(poly(n)) \left(\frac{\O(1)}{\epsilon} \right)^{\O(1/\gamma)}$ if $s, d_{max} = \O(1)$ 
			\item quasi-polynomial in $n$ (for inverse polynomial precision $\epsilon$) and exponential in the inverse spectral gap $\gamma^{-1}$ with  $\O(poly(n)) n^{\O(\log{\epsilon^{-1}}/\gamma)}$ if $sd_{max} = poly(n)$  and  $\O(poly(n)) \log{n}^{\O(\log{\epsilon^{-1}}/\gamma)}$ if $sd_{max} = \O(log(n))$ 
		\end{enumerate}
	This can be slightly tightened by replacing $\gamma \rightarrow \log{(1-\gamma)^{-1}}$.
\end{corollary}
\begin{proof}
	We directly use Theorem~\ref{thm:mainsampling} such that we have a method to output a function $\tilde{p}^{\T}$ with  $||\tilde{p} - \tilde{p}^{\T}||\leq  (1-\gamma)^{L}$. Then to get the target approximation error we ask $(1-\gamma)^L \leq  \epsilon$, which implies that $L log(1-\gamma) \leq log(\epsilon)$ so $L \geq log(\epsilon)/ log(1-\gamma) = log(1/\epsilon) \frac{1}{log((1-\gamma)^{-1})} $  Therefore choosing a value of $L \approx \frac{1}{\gamma} log(1/\epsilon)$ satisfies the target approximation error. Turning now to the complexity of the algorithm we have that the total runtime is going to be at most $\O(\D) poly(n) (L (sd_{max})^L)$ where we emphasise the degree of the polynomial does not depend on $\gamma$ or  $\epsilon$. We have then
	\begin{align}
		\O(D) poly(n)(L (sd_{max})^L)  &= \O(D) poly(n) \gamma^{-1}\log{\epsilon^{-1}}  \  \  (sd_{max})^{\frac{1}{\gamma} \log{\epsilon^{-1}}}\\
		& = \O(D) poly(n)   \gamma^{-1}\log{\epsilon^{-1}}  \left(\frac{1}{ \epsilon}\right)^{ \frac{\log{sd_{max}}}{\gamma} }\\
		& = \begin{cases}
			 \O(poly (n)) \gamma^{-1}\log{\epsilon^{-1}}  \left( \frac{1}{\epsilon}\right)^{ \O(\frac{\log(n)}{\gamma})}  \hspace{1cm}  {\rm{if}}   \hspace{1cm}   sd_{max} = poly(n) \\
			 \O(poly (n)) \gamma^{-1} log \epsilon^{-1} \left( \frac{1}{\epsilon}\right)^{ \O(\frac{\log{\log{n}}}{\gamma})}   \hspace{0.9cm}  {\rm{if}}   \hspace{1cm}    sd_{max} = \O(\log{n}) \\
			\O(poly (n)) \gamma^{-1} log \epsilon^{-1} \left( \frac{1}{\epsilon}\right)^{ \O(\frac{1}{\gamma})}   \hspace{1.6cm}  {\rm{if}}   \hspace{1cm}    sd_{max} = 
			\O(1) 
			\end{cases}
	\end{align}
where we have used that the circuit depth is $D = poly(n)$.
\end{proof}

\subsection{Operator truncated classical simulation framework - for $W$'s permutation of operator basis}
In this section we will describe how the framework simplifies whenever the $W$-type operators in the ensemble of circuits act as permutation of the operator basis elements considered.  

\subsubsection{Operator truncated expansion for a restricted class of unitaries}

\label{sec:paths}
In here we describe in more detail how the framework is applied to unitary classes where the $W$ operators are a permutation of the ITO basis.

The observable $\O$ can be decomposed in terms of the ITO basis as $O = \sum_{\lambda,k} O_{\lambda,k}  T^{\lambda}_k$ for some complex coefficients given by $O_{\lambda,k }= \tr [(T^{\lambda}_k)^{\dagger} O]$.  Therefore we have by linearity that the mean value of $O$ with respect to the state $C(\g)|0^n\>$ is given by
\begin{align}
o({\g}) = \sum_{\lambda,k} O_{\lambda,k} \< 0^{\otimes n}| C^{\dagger}(\g) T^{\lambda}_k C(\g) |0^{\otimes n} \> 
\end{align}
Now our restriction that $W$ is a permutation of the irreducible tensor operators means that $W^{\dagger} T^{\lambda}_{k} W = T^{w(\lambda)}_{w(k')}$  for $w$ a permutation of the irrep labels and components with $w(\lambda)$ denoting a potentially distinct irrep and $w(k')$ the corresponding component.  This makes it easier to see how each ITO $T^{\lambda}_k$ propagates through the circuit. Specifically, in the Heisenberg picture after applying the first layer the evolved operator becomes

\begin{align}
C^{\dagger} (\g) T^{\lambda}_k C (\g) = \sum_{k'} v^{\lambda}_{k'k}(g^{-1}_{1})   [W^{\dagger} U^{\dagger}(g_D)...WU^{\dagger}(g_2)] T^{w(\lambda)}_{w(k')} [U(g_2) W .... U(g_D) W].
\end{align}

We continue this to evolve through the entire circuit and obtain the full decomposition
\
\begin{align}
C^{\dagger} (\g) T^{\lambda}_k C (\g) =
\sum_{(\boldsymbol{\lambda},\boldsymbol{k})} v^{\lambda}_{k_1k}(g^{-1}_{1}) v^{\lambda_2}_{k_2 k_1'}(g^{-1}_{2}) ... v^{\lambda_D}_{k_{D} k_{D-1}'} (g^{-1}_{D}) T^{\lambda_D}_{k_D} .
\end{align}
In the above, it's worth emphasising that the summation should be viewed  not over all possible \emph{paths} choices of irreps and components $(\lambda,k)\rightarrow (\lambda_2,k_2)\rightarrow ...\rightarrow (\lambda_D,k_D)$, but only over those \emph{valid paths} that give a non-zero contribution to the backwards-evolved operator. The circuit structure rules out certain paths, giving us the following constraint.  Any subsequent possible choices of irrep and component $(\lambda_i, k_i)$  in the path depend on the component and irrep of the previous step, and this dependency is dictated by $W$. For simplicity, in this overview section we assumed that $W$ is a permutation of the basis elements. It implies that propagating the fixed ITO $T^{\lambda}_k$ through a single layer gives us a combination of at most $d_{\lambda}$ basis elements.  Starting from an operator $T^{\lambda}_k$  - the fact that these form a symmetry-adapted basis means that  the action of the unitary channel $U(g)$ does not mix basis operators associated to different $\lambda$-irrep subspaces.  Therefore there will be at most $d_{\lambda}$ possible branches.  On each of these branches the input operator $T^{\lambda}_k$ is mapped to some other basis operator $T^{\lambda}_{k'}$ with a particular weighting given by the corresponding matrix coefficient $v^{\lambda}_{k'k}(g)$. We can associate these weights to each edge starting at $T^{\lambda}_k$ and ending at one of the allowed operators. The branching process itself is independent of which group elements are applied at each layer - the only thing that changes is the value of the weights on the edges.  Subsequently we apply the unitaries $W$ which permute the symmetry-adapted operator basis. This means there's no further branching out, but the output operators on each of the $d_{\lambda}$ branches no longer correspond necessarily to the same irrep $\lambda$.  The operators at the end the branches now become input operators for the next layer, and the whole process continues as described in Fig.~\ref{paths}.

In light of these constraints on the valid paths that occur in the expansion of $C^{\dagger} (\g)T^{\lambda}_k C(\g)$ we can simplify notation for the paths and describe two types of functions, a path validation and path evaluation:

\begin{align}
(\boldsymbol{\lambda},\boldsymbol{k}) &:= (\lambda,k) \rightarrow (\lambda_2,k_2) \rightarrow ... \rightarrow (\lambda_D, k_D)  \  \ {\rm{with} } \ \boldsymbol{\lambda} =  [ \lambda, \lambda_2, ...\lambda_D]  \ \  {\rm{and} }  \  \  \boldsymbol{k}   =  [k,k_2,..., k_D]\\
\texttt{validate}[(\boldsymbol{\lambda}, \boldsymbol{k}), w] &= True   \  \  {\rm{iff}} \  \  \exists!  k',k_2',..., k_D'   \  \ {\rm{s.t}  } \  \ (\lambda_i,k_i') = w(\lambda_{i-1}, k_i)  \\
\texttt{matrix-coeff}[(\boldsymbol{\lambda}, \boldsymbol{k}), \g, w]  &= v^{\lambda}_{k_1k}(g^{-1}_{1}) v^{\lambda_2}_{k_2 k_1'}(g^{-1}_{2}) ... v^{\lambda_D}_{k_{D} k_{D-1}'} (g^{-1}_{D}) =: v^{\boldsymbol{\lambda}}_{\boldsymbol{k',k}}(\g^{-1}) 
\end{align}


Finally, using the above notation we have that the operator $O$ fully evolved in the Heisenberg picture decomposes 
\begin{align}
\O_{\g} := C^{\dagger}(\g) O C(\g) =   \sum_{\substack{(\boldsymbol{\lambda},\boldsymbol{k}) \\ \texttt{validate}[(\boldsymbol{\lambda}, \boldsymbol{k})] = True}} v^{\boldsymbol{\lambda}}_{\boldsymbol{k'},\boldsymbol{k}}(\g^{-1}) O_{\lambda,k}  T^{\lambda_D}_{k_D}.
\end{align}
In general, there will be multiple paths that start at  $(\lambda,k)$ and end up at $(\lambda_D, k_D)$ that correspond to different branching decisions. Counting the number of total valid paths depends on the dimension of the allowable irreps (e.g those that appear in the decomposition of $U\otimes U^{*}$) but also on how the basis elements get mixed up by the fixed unitaries such as $W$.   In the restricted situation we have considered here where $W$ is a permutation unitary in the ITO basis, it is fairly straightforward to compute a loose upper bound for the size $|\P|$ of all valid paths. This would be given by  $  |\P| \leq  d_{max}^D \, |\Lambda| $  where we denoted by $d_{max} := \max_{ \lambda\in U\otimes U^{*}\cap \hat{G} }{\rm{dim}}(\lambda)$  the maximal dimension of the irrep in the decomposition of $U\otimes U^{*}$ and by $\Lambda = \{\lambda\in \hat{G} :   \exists k  \  {\rm{s.t}} \   \Tr(T^{\lambda}_k O) \neq 0  \}$ the set of all irreps of $G$ on which $O$ has non-zero support. However, this is a very generic upper bound, so it is to be expected that specific circuit classes would lead to smaller bounds.

\subsubsection{Classical simulability of noisy circuits and average errors in the  truncated approximation for observables}

Once a truncation set $\mathcal{T}$ has been established, either by enumerating paths in $\mathcal{T}$ or defining a rule that tells us if a given path belongs to the set, then we want to check if it enables an efficient classical computation of the truncated operator. There are three ingredients that makes this possible: (i) The size of the truncation set (ii) Efficiency in the evaluation of each path (iii) Approximation error that is bounded. 

The truncated operator evolution for an observable $O$ will be given by
\begin{equation}
\O_{\g}^{\T} = \sum_{\substack{(\boldsymbol{\lambda},\boldsymbol{k}) \in \mathcal{T} }} v^{\boldsymbol{\lambda}}_{\boldsymbol{k'},\boldsymbol{k}}(\g^{-1}) O_{\lambda,k}  T^{\lambda_D}_{k_D},
\end{equation}
where we use the truncation set defined in~\eqref{eqn:truncation} determined by paths that satisfy $||\boldsymbol{\lambda}||_0 \leq L$ for a fixed cut-off parameter $L$.

{\bf{Size of the truncation set}} Recall we have previously defined $|\Lambda|$ as the size of the support of the given observable $O$ on the ITO basis. In here, we will assume this has a polynomial size in $n$, the number of qudits. Therefore, there will be $poly(n)$ number of starting operators $T^{\lambda}_k$ for which the coefficient $O_{\lambda,k } = Tr((T^{\lambda}_k)^{\dagger} O) \neq 0$ is non-zero. For each of these, we construct the valid paths $(\boldsymbol{\lambda}, \boldsymbol{k}) $ that start at $(\lambda,k)$ and satisfy the truncation rule. Starting with $T^{\lambda}_k$  and propagating through the circuit as we did in Section~\ref{sec:paths} we notice that we encounter multiple branching paths only when the propagated operator that becomes an input to the next layer does not lie in the trivial subspace. In other words, within a path $(\lambda_0,k_0 ) \longrightarrow (\lambda_1, k_1) \longrightarrow ... \longrightarrow (\lambda_{D-1},k_{D-1})$, whenever $(\lambda_i,k_i)$ correspond to a trivial irrep (we can view the component here as a multiplicity label), then this means knowledge of the circuit, in particular of $W$ will be enough to completely determine $(\lambda_{i+1}, k_{i+1})$. This is a result of  $W$ being a permutation in the ITO basis so then we deterministically move from $(\lambda_i, k_i)$ to $(\lambda_{i+1}, k_{i+1})$ and there are no other possible choices to propagate the $T^{\lambda_i}_{k_i}$ operator through the next layer.  Because of our constraint $||\boldsymbol{\lambda}||_0\leq L $ along each valid path in the truncation set we can encounter at most $L$ positions where the tree starting at $T^{\lambda}_k$ branches out into multiple paths. Therefore all the possible valid paths in $\T$ that start at $T^{\lambda}_k$ are described by a tree with height $L$.  The tree will have at each node a degree that's at most $d_{max}$, which is the maximal dimension of the irreps in the decomposition of $U\otimes U^{*}$.  Putting it all together, this gives us the size of the truncation set $\T$ as $|\T| \leq |\Lambda| d_{\rm{max}}^{L}$. Whenever $d_{max}$ is polynomial in the number of qudits then this leads to a polynomial-sized truncation set $|\T| = poly(n)$. 

{\bf{Path evaluation}} The weight corresponding to a valid path in the truncation set $(\boldsymbol{\lambda}, \boldsymbol{k}) \in \T$ will have at most $L$ non-trivial irreps, and therefore the corresponding coefficient  $v^{\boldsymbol{\lambda}}_{\boldsymbol{k',k}}(\g^{-1})$ will be a product of at most $L$ terms. If the individual irreducible matrix coefficients can be evaluated efficiently (i.e polynomial-time in the number of qudits) for all the irreducible representations of $G$ that appear in the decomposition of $U\otimes U^{*} $ then we can compute $v^{\boldsymbol{\lambda}}_{\boldsymbol{k',k}}(\g^{-1})$ in $O(L)$ time.  Similarly, to compute the expectation values for the truncated observables, we also need to be able to efficiently compute $O_{\lambda, k} $ and $Tr(T^{\lambda_D}_{k_D} |0^{n}\>\<0^n|)$.  

{\bf{Average approximation bounds}}  Finally, we need to quantify how well the truncated operator approximates the fully evolved observable when we include the effects of noise. To that aim we will employ an \emph{average} approximation error given by the $L_2$ norm average over all the independent group elements $\g = (g_1,..., g_D)$. This allows us to connect the approximation error with the contributions of paths not included in the truncation set. As before, we achieve this using a generalised form of Parseval's theorem which we state here in its simplified form for the particular circuits considered here where the $W$ operators are permutations of the basis operators.

\begin{lemma} 
	Let's define functions $f: G^{\times D} \longrightarrow \mathbb{C} $ via $f(\g) =  \<\psi| \O_{\g} |\phi\>$, then we have
	\begin{align}
		||f||_{L^2(G)}^2 & = \int |\<\psi|\O_{\g}|\phi\>|^2 d\, \g  =  \sum_{\substack{(\boldsymbol{\lambda},\boldsymbol{k}) \\ \texttt{\textup{validate}}[(\boldsymbol{\lambda}, \boldsymbol{k})] = True}}  |O_{\lambda,k} |^2 \frac{1}{d_{\boldsymbol{\lambda}}} |\<\psi|T^{\lambda_D}_{k_D} |\phi\>|^2. 
	\end{align} 
	where $d_{\boldsymbol{\lambda}} = d_{\lambda} d_{\lambda_1} ... d_{\lambda_D}$, with each $d_{\lambda_i}$ corresponding to the dimension of the $\lambda_i$ irrep.
\end{lemma}

\begin{proof}
	We have that $O = \sum_{\lambda, k}  O_{\lambda, k}  T^{\lambda}_k$ with $O_{\lambda,k } = \Tr( O (T^{\lambda}_l )^{\dagger})$ since $\{T^{\lambda}_k\}$ form an orthonormal operator basis. Furthermore, the evolved operator $\O_{\g} = C(\g)^{\dagger} O  C(\g)$ expands into valid paths as
	\begin{equation}
		\O_{\g}   =   \sum_{\substack{(\boldsymbol{\lambda},\boldsymbol{k}) \\ \texttt{validate}[(\boldsymbol{\lambda}, \boldsymbol{k})] = True}}   v^{\boldsymbol{\lambda}}_{\boldsymbol{k',k}} (\g^{-1} )  T^{\lambda_D}_{k_D}  O_{\lambda, k}
	\end{equation}
	
	Therefore we have that
	\begin{align}
		||f||^2_{L^2(G)} & =  \int   \sum_{\substack{(\boldsymbol{\lambda},\boldsymbol{k}), (\boldsymbol{\mu},\boldsymbol{m}) \\ \texttt{validate}[(\boldsymbol{\lambda}, \boldsymbol{k})] = True  \\ \texttt{validate}[(\boldsymbol{\mu}, \boldsymbol{m})] = True}}    v^{\boldsymbol{\lambda}}_{\boldsymbol{k',k}} (\g^{-1} )  \<\psi|T^{\lambda_D}_{k_D} |\phi\> O_{\lambda, k}     (v^{\boldsymbol{\mu}}_{\boldsymbol{m',m}} (\g^{-1} ))^{*}  \<\psi|T^{\mu_D}_{m_D} |\phi\>^{*} O_{\mu, m} ^{*}d\, \g \\
		& =     \sum_{\substack{(\boldsymbol{\lambda},\boldsymbol{k}), (\boldsymbol{\mu},\boldsymbol{m}) \\ \texttt{validate}[(\boldsymbol{\lambda}, \boldsymbol{k})] = True  \\ \texttt{validate}[(\boldsymbol{\mu}, \boldsymbol{m})] = True}}  \<\psi|T^{\lambda_D}_{k_D} |\phi\> O_{\lambda, k}       \<\psi|T^{\mu_D}_{m_D} |\phi\>^{*} O_{\mu, m} ^{*}    \int v^{\boldsymbol{\lambda}}_{\boldsymbol{k',k}} (\g^{-1} ) (v^{\boldsymbol{\mu}}_{\boldsymbol{m',m}} (\g^{-1} ))^{*} d\, \g  
	\end{align}
	However from Schur orthogonality theorem we have that
	\begin{align}
		\int v^{\boldsymbol{\lambda}}_{\boldsymbol{k',k}} (\g^{-1} ) (v^{\boldsymbol{\mu}}_{\boldsymbol{m',m}} (\g^{-1} ))^{*} d\, \g  = \frac{1}{d_{\lambda_0}} \frac{1}{d_{\lambda_1}}... \frac{1}{d_{\lambda_{D-1}}}  \delta_{\boldsymbol{\lambda,\mu}} \delta_{\boldsymbol{k,m}}\delta_{\boldsymbol{k',m'}} \
	\end{align}
\end{proof}

Now we want to see the approximation error if we were to truncate the noisy operator decomposition. Keeping with the notation in the previous section and adding tilde's to denote the noisy version of the operator $\tilde{\O}_{\g} = \tilde{\C}(\g^{-1}) O \, \tilde{\C}(\g) $ and its truncated form $\tilde{\O}_{\g}^{\T}$ over a subset, we then obtain the average approximation error
\begin{align}
(\tilde{\Delta}^{\T} )^2:&= || \tilde{o}(\g)  -  \tilde{o}^{\T}(\g) ||_{L^2(G)} ^2 \\ &=  \sum_{(\boldsymbol{\lambda},\boldsymbol{k}) \notin \T } |O_{\lambda_0,k_0}|^2 \frac{| e_{\lambda} e_{\lambda_2} ... e_{\lambda_D}|^2}{d_{\lambda_0} d_{\lambda_1} ... d_{\lambda_{D-1}}} |\<0^n| T^{\lambda_D}_{k_D} |0^n\>|^2
\end{align}
which follows again from the application of Parseval's theorem. Since we have chosen the truncation set in such a way that all remaining contributions are contracted by the effect of the noise channel, then we have that $|e_\lambda ... e_{\lambda_D} |\leq E^{||\boldsymbol{\lambda}||_0} $. Furthermore since $E<1$ and then for every $(\boldsymbol{\lambda}, \boldsymbol{k})\notin \mathcal{T}$ we have $||\boldsymbol{\lambda}||_0 > L$, this implies that $|e_\lambda ... e_{\lambda_D} | \leq E^{L}$. Therefore we can bound the average approximation error as 
\begin{align}
(\tilde{\Delta}^{\T})^2 &= || \tilde{o}(\g)  -  \tilde{o}^{\T}(\g) ||_{L^2(G)}^2   \leq  E^{2L} ||o(\g)  -  o^{\T}(\g) ||_{L^2(G)} ^2  \\  \nonumber & \leq E^{2L} ||o(\g)||_{L^2(G)}^2
\end{align}

Now, recall that we defined $E$ as the largest eigenvalue of the channel that has magnitude strictly less than 1.  Therefore the \emph{spectral gap} of the channel can be defined as $\gamma =1 - E$.  This allows us to have an exponential decaying bound as
\begin{align}
(\tilde{\Delta^{\T} })^2 \leq (1-\gamma)^{2L} ||o(\g) ||_{L^2(G)} ^2\leq e^{-2\gamma L } ||o(\g)||^2_{L^2(G)}.
\end{align}

Furthermore we can use Cauchy-Schwartz inequality for the Hilbert-Schmidt norm $||o(g) ||^2_{L^2(G)} = \int |\Tr[C(\g) |0\>\<0| C^{\dagger}(\g) O]|^2 d\, \g  \leq  ||O||_{HS}^2 $  where $||O||_{HS}^2  = Tr(O^{\dagger} O)$. 

\section{Extended analysis related to the examples in main text}

\subsection{U(1)  group action on circuit ensembles}
\label{app:realU1}
\subsubsection{Example: U(1) complex vs real representations and the corresponding Fourier expansion}
$U(1)$ corresponds to the circle group $\mathbb{T}$ (unit complex numbers under multiplication) and therefore can be parametrised by angle  $\theta$ via the exponential map $\theta \rightarrow e^{i\theta}$.  As $U(1)$ is abelian all its complex irreducible representations are one-dimensional and indexed by integers $m\in \mathbb{Z}$ so that $V_{m} (\theta) : v \rightarrow e^{im\theta}v$ for $v\in \mathbb{C}$ with character $\chi_m(\theta) = e^{im\theta}$. However, the \emph{real} irreducible representations, with the exception of the trivial representation $\chi_0(\theta) = 1$, are two-dimensional, \footnote{Irreducible representations on complex spaces, together with the Frobenius-Schur indicator given by $\int_{G} \chi(g^2) d\,g$ tells us how to construct the real representations.  Conversely one can obtain the complex irreps by complexification.} with characters $\chi_{m} ^{(real)}(\theta)= 2\cos{m\theta}$ for $m\in \mathbb{N}^{*}$ and given by elements of $SO(2)$ with
\begin{equation}
	V_m^{(real)}(\theta) = \left( \begin{array}{cc} \cos{m\theta} & -\sin {m\theta} \\ \sin{m\theta} &\cos{m\theta} \end{array}\right)
\end{equation}
Just like in the complex case, the real characters are orthogonal, however they are no longer normalised.  Namely, 
\begin{align}
	\<\chi_m, \chi_n\> &= \int_{0}^{2\pi} \chi_m(\theta) \chi_n(-\theta) \frac{d\theta}{2\pi} = \delta_{m,n}\\
	\<\chi_m^{(real)}, \chi_n^{(real)}\>  &= 4\int_{0}^{2\pi} \cos{m\theta} \cos{n\theta} \frac{d\theta}{2\pi} = 2 \delta_{mn}, \   \ m\in \mathbb{N}^{*}
\end{align}

Now, consider a representation of $U(1)$ on a single qubit $\H \cong \mathbb{C}^2 $ given by $U_1 (\theta) = e^{iZ\theta}$ where $Z$ is the Pauli operator. The object of interest $U_1\otimes U_1^{*}$  also forms a representation of $U(1)$ that decomposes into complex irreps with $m=\pm 2, 0$ with the space decomposition $\h\otimes \h^{*} \cong 0\oplus 0\oplus 2\oplus -2$ and as matrices
\begin{align}
	U_1(\theta)\otimes U_1^{*} (\theta) = \left( \begin{array}{cccc}
		1 & 0&0 &0  \\
		0 & e^{2i\theta}& 0 & 0\\
		0 & 0& e^{-2i\theta}& 0 \\
		0 & 0&0 &1  \\
	\end{array} \right)
\end{align}

Now the character of $\chi_{U_1\otimes U_1^{*}} (\theta) = |\Tr(U_1(\theta))|^2$ is real, which means that we can decompose it also in terms of $\emph{irreducible real}$ representations as $\h\otimes \h^{*} \cong V_0^{(real)} \oplus V_0^{(real)}  \oplus V_2^{(real)}$.

\begin{align}
	U_1(\theta)\otimes U_1^{*} (\theta) = C \left( \begin{array}{cccc}
		1 & 0&0 &0  \\
		0 &\cos{2\theta}& -\sin{2\theta} & 0\\
		0 & \sin{2\theta}& \cos{2\theta}& 0 \\
		0 & 0&0 &1  \\
	\end{array} \right) C^{\dagger}
\end{align}

With the change of basis unitary 
\begin{align}
	C  = \left( \begin{array}{cccc}
		1 & 0&0 &0  \\
		0 &1/\sqrt{2}& i/\sqrt{2}& 0\\
		0 & 1/\sqrt{2}&-i/\sqrt{2}& 0 \\
		0 & 0&0 &1  \\
	\end{array} \right) 
\end{align}

Now let's express $U_1\otimes U_1^{*}$ in the PTM representation $\{\mathbb{I}/\sqrt{2}, X/\sqrt{2}, Y/\sqrt{2}, Z/\sqrt{2}\}$ 
\begin{align}
	U_1\otimes U_1^{*}  \stackrel{PTM}{=} \left( \begin{array}{cccc}
		1 & 0&0 &0  \\
		0 &\cos{2\theta}& -\sin{2\theta} & 0\\
		0 & \sin{2\theta}& \cos{2\theta}& 0 \\
		0 & 0&0 &1  \\
	\end{array} \right) 
\end{align}
where we have used that $e^{iP\theta} P' e^{-iP\theta} = cos{2\theta} P' +  i \sin{2\theta} PP' $ for anti-commuting Pauli operators $P$ and $P'$. 

What does the Fourier analysis using the complex representation looks like:
\begin{align}
	\reallywidehat{U_1\otimes U_1^{*}} (0) & =   \left( \begin{array}{cccc}
		1 & 0&0 &0  \\
		0 &0& 0& 0\\
		0 & 0&0& 0 \\
		0 & 0&0 &1  \\
	\end{array} \right)  \\
	\reallywidehat{U_1\otimes U_1^{*}} (2) & =     \left( \begin{array}{cccc}
		0 & 0&0 &0  \\
		0 &0& 0& 0\\
		0 & 0&1 & 0 \\
		0 & 0&0 &0 \\
	\end{array} \right)  \\
	\reallywidehat{U_1\otimes U_1^{*}} (-2) & =     \left( \begin{array}{cccc}
		0 & 0&0 &0  \\
		0 &1& 0& 0\\
		0 & 0&0& 0 \\
		0 & 0&0 &0 \\
	\end{array} \right)  \\
\end{align}
with the inverse formula
\begin{align}
	U_1\otimes U_1^{*} (\theta ) = \reallywidehat{U_1\otimes U_1^{*}}(0) + \reallywidehat{U_1\otimes U_1^{*}}(2) e^{-2i\theta}  + \reallywidehat{U_1\otimes U_1^{*}}(-2)e^{2i\theta}
\end{align}
and equivalently in the PTM representation 
\begin{align}
	\reallywidehat{U_1\otimes U_1^{*}} (2) & =     \left( \begin{array}{cccc}
		0 & 0&0 &0  \\
		0 &1/2& -i/2& 0\\
		0 & i/2&1/2 & 0 \\
		0 & 0&0 &0 \\
	\end{array} \right)  \\
	\reallywidehat{U_1\otimes U_1^{*}} (-2) & =     \left( \begin{array}{cccc}
		0 & 0&0 &0  \\
		0 &1/2& i/2& 0\\
		0 & -i/2&1/2& 0 \\
		0 & 0&0 &0 \\
	\end{array} \right)  \\
\end{align}

What does the Fourier analysis using the real representation looks like:
\begin{align}
	\reallywidehat{U_1\otimes U_1^{*}} (0) & =   \left( \begin{array}{cccc}
		1 & 0&0 &0  \\
		0 &0& 0& 0\\
		0 & 0&0& 0 \\
		0 & 0&0 &1  \\
	\end{array} \right)  \\
	\reallywidehat{U_1\otimes U_1^{*}} (2) & =     \left( \begin{array}{cccc}
		0 & 0&0 &0  \\
		0 &1/2 & 0& 0\\
		0 & 0&1/2 & 0 \\
		0 & 0&0 &0 \\
	\end{array} \right)  \otimes \left(\begin{array}{cc} 1& 0 \\ 0& 1 \end{array} \right) +  \left( \begin{array}{cccc}
		0 & 0&0 &0  \\
		0 &i/2 & 0& 0\\
		0 & 0&-i/2 & 0 \\
		0 & 0&0 &0 \\
	\end{array} \right)  \otimes \left(\begin{array}{cc} 0& -1 \\ 1& 0 \end{array} \right)
\end{align}
And similarly in the PTM representation
\begin{align}
	\reallywidehat{U_1\otimes U_1^{*}} (2) & =     \left( \begin{array}{cccc}
		0 & 0&0 &0  \\
		0 &1/2 & 0& 0\\
		0 & 0&1/2 & 0 \\
		0 & 0&0 &0 \\
	\end{array} \right)  \otimes \left(\begin{array}{cc} 1& 0 \\ 0& 1 \end{array} \right) +  \left( \begin{array}{cccc}
		0 & 0&0 &0  \\
		0 &0 & -1/2& 0\\
		0 & 1/2& 0 & 0 \\
		0 & 0&0 &0 \\
	\end{array} \right)  \otimes \left(\begin{array}{cc} 0& -1 \\ 1& 0 \end{array} \right)
\end{align}
with the inverse formula 
\begin{align}
	U_1\otimes U_1^{*} (\theta) = \reallywidehat{U_1\otimes U_1^{*}} (0) + \Tr_{2}[\reallywidehat{U_1\otimes U_1^{*}}(2) \mathbb{I} \otimes (V_2^{(real)}(\theta))^{T}]
\end{align}
Note this is the inversion formula for \emph{real} representations. 

\subsection{Noisy Random Circuit Sampling}
In this section we describe an algorithm for noisy random circuit sampling that does not rely on \emph{Assumption 3} regarding efficient contractions of sequences of $W_i$ operators on the trivial subspace. This uses the $\texttt{GTT}$ constructions and a layer-based circuit structure that differs from the analysis in the main text.  We emphasise here that $D$ is the number of layers in the circuit as opposed to the number of gates (as used in the main text).

We look at the family of random circuit unitaries $\mathfrak{C}_n$ acting on $n$-qubits. In line with our notation every unitary $C(\g) \in \mathfrak{R}_n$ is indexed by an element $\g \in G^{\times D}$. Here $G$ itself is a product group given by $U(4)^{\times  q}$ where $q$ is the number of pairs of qubits in each layer on which we act non-trivially with an element of $U(4)$, drawn at random according to the uniform Haar measure.  In particular, a layer $U(g_i)W_{\sigma_i}$ will be composed of a fixed permutation $W_{\sigma}$ of the $n$ qubits and unitary representation $U(\g_i)$ given by $q$ tensor products of the fundamental representation 
\begin{equation}
	U(\g_i) = U(g_{i1})\otimes ...\otimes U(g_{iq})\otimes \I
\end{equation}
where  each $U(g_{ij})$ acts on disjoint pairs of two qubits and the identity $\I$ acts on the remaining Hilbert space. Notably, tensor products of irreps form an irreducible representation of the group product $U(4)^{\times q}$. For any Pauli operator (normalised w.r.t Hilbert-Schmidt inner product) $P_{i_1} \otimes ... \otimes P_{i_n}$ there exists $ (\lambda_i, k_i) = [(\lambda_{i_1}, k_{i_1}),..., (\lambda_{i_q}, k_q) ]$ such that $P_{i_1} \otimes ... \otimes P_{i_n} = T^{\lambda_{i_1}}_{k_{i_1}} \otimes ... T^{\lambda_{i_q}}_{k_{i_q}} \otimes P_{i_{2q+1}}  \otimes ... \otimes P_{i_n}$. This transforms as
g
	\begin{align}
		U(\g_i)^{\dagger}  P_{i_1} \otimes .... \otimes P_{i_n} U(\g_i)  =  \sum_{k'_{i_1},... k'_{i_q} }  v^{\lambda_{i_1}}_{k_{i_1} k'_{i_1}} (g_{i1}^{-1}) ...  v^{\lambda_{i_q}}_{k_{q_1} k'_{i_q}} (g_{iq}^{-1}) \, T^{\lambda_{i_1}}_{k'_{i_1}} \otimes ... T^{\lambda_{i_q}}_{k'_{i_q}} \otimes P_{i_{2q+1}}  \otimes ... \otimes P_{i_n}
		\label{eq:randomlayer1}
	\end{align}
Since group elements are chosen independently, every $\lambda$-irrep can only be the 0 or 1 - irrep and thus the number of terms in this expansion will generally be $15^{||\lambda_i||_0}$, where $||\lambda_i||_0$ denotes the number of non-trivial irreps in $(\lambda_{i_1}, ... \lambda_{i_q})$.  We keep in mind that this is completely determined by the starting Pauli operator that we propagate through the circuit in the Heisenberg picture.  Furthermore, as   $W_{\sigma_i}$ is a permutation of the $n$ qubits then the number of terms in the expansion in Eqn.~\eqref{eq:randomlayer1} remains the same; however the support of those operators in the expansion can correspond to different irreps. 
\subsubsection{Truncated path decomposition for estimating probabilities} 
Suppose we truncate paths according to the following rule  $\sum_{i=1}^{D} ||\lambda_i||_0 \leq L$ for some \emph{fixed} truncation value $L$, where $||\lambda_i||_0$ denotes the number of non-trivial irreps in $(\lambda_{i_1}, ... \lambda_{i_q})$. Then the number of valid paths starting from a given Pauli operator $P_{1_1} \otimes...P_{1_n}$, will be at most $15^L$.  Now if we were to go over all possible input Pauli operators that satisfy the truncation rule we would end up with $\binom{q}{k}15^k$ where $q \leq [n/2]$ and $k\leq L$.  But we can do slightly more clever and notice that the constraint imposes that in any valid paths at least one layer $i$ satisfies $||\lambda_i||_0\leq L/D$  (as we have at least one non-trivial irrep at each layer). The second useful observation is that the group structure involved implies that $||\lambda_i||_0/2 \leq ||\lambda_{i+1}||_0 \leq 2||\lambda_i||_0$.  Therefore we can start from Pauli operators that satisfy $||\lambda_1||_0 \leq  L/D$, and for each of them we evolve backwards in  the Heisenberg picture until we reach the truncation constraint  - overall the number of all such valid paths ranging over all Pauli's that satisfy the constrain will be at most $\sum_{k=1}^{L/D} \binom{[n/2]}{k}15^{L}$. This is not enough to range over all all valid paths that satisfy the truncation rule, however any other valid paths must satisfy $||\lambda_1||_0\ge L/D$.  So we turn to the next layer and we seek those irreps that satisfy $||\lambda_2||_0 \leq L/D$, however this will lead to valid paths when the condition $L/2D \leq ||\lambda_2||_0\leq L/D$ is met.  For any given choice of irrep at this layer, there are $15^{||\lambda_2||_0}$ Paulis  which we propagate forwards and backwards for a total of at most $15^{L}$ valid paths per choice of irrep. Overall the number of valid paths for which either $||\lambda_1||_0\leq L/D$ or $||\lambda_2||_0\leq L/D$ will be at most  $ \sum_{k=1}^{L/D} \binom{[n/2]}{k} 15^{L}  + \sum_{k=L/2D}^{L/D} \binom{[n/2]}{k} 15^{L}$.  Then we continue this procedure for all $D$ layers.  In total, we obtain $\sum_{k=1}^{L/D} \binom{[n/2]}{k} 15^L  + (D-1) \sum_{k=L/2D}^{L/D}\binom{[n/2]}{k} 15^L \leq (n/2)^{L/D} 15^{L} (L/D + (D-1) L/2D)  \approx  L \, 15^L n^{L/D}$ - a similar complexity to that obtained in ~\cite{aharonov2023polynomial}.


The full circuit $C(\g) \in \mathfrak{R}_n$ with $D$ layers is given by
\begin{equation}
	C(\g) = U(\g_1) W_{\sigma_1} ... U(\g_D)W_{\sigma_D}.
\end{equation}

Let's denote $p(C(\g),x) = |\<x| C(\g) | 0\>|^2$  the ideal output distribution of outcome $x\in \{0,1\}^{n}$ given a random circuit $C(\g) \in \mathfrak{R}_n$ and by $\tilde{p}(C(\g)),x)$ the output distribution of the noisy circuit. First we look at the expansion of the observable $|x\>\<x|$ in terms of the irreducible operator basis  $\{ T^{\lambda}_k \}$ such that $|x\>\<x| = \sum_{\lambda,k}  O_{\lambda,k} (x) T^{\lambda}_k$  where orthonormality of ITOs implies that the coefficients in the expansion are given by $O_{\lambda,k}=Tr((T^{\lambda}_k)^{\dagger} |x\>\<x|) $ and satisfy the following $|\<x_1 |x_2\>|^2 = \sum_{\lambda,k}  O_{\lambda,k}^{*}(x_1) O_{\lambda,k}(x_2) $.  

\subsubsection{Approximation bounds}
Bounds on square of the total variation distance $||\tilde{p}-\bar{q}||_1 = \sum_{x} |\tilde{p}(C(\g),x) -\bar{q}(C(\g),x)|$  between distributions $\tilde{p}$ and $\bar{q}$ averaged over the group is given by

\begin{align}
	\int &||\tilde{p} - \bar{q} ||_{1} ^2 d\, \g  \leq 2^n \sum_{x} \int |\tilde{p}(C(\g), x) -\bar{q}(C(\g),x )|^2 d\, g  \\ &= 2^n \sum_{x} || \<0| \tilde{O}_{\g} (x) - \tilde{O}_{\g}^{\I} (x) |0\>||_{L^2(G)}\\
	& = 2^n \sum_{x}  \sum_{\boldsymbol{\lambda} \notin \mathcal{I}} |O_{\lambda,k}(x)|^2  \frac{|e_{\boldsymbol{\lambda}}|^2}{d_{\boldsymbol{\lambda}}} |\<0| T^{\lambda_D}_{k_D} |0\>|^2\\
	& =  2^{2n} \sum_{{\boldsymbol{\lambda}} \notin \I}    \frac{|e_{\boldsymbol{\lambda}}|^2}{d_{\boldsymbol{\lambda}}} |\<0| T^{\lambda_D}_{k_D} |0\>|^2 |O_{\lambda,k} (0) |^2 \\ & \leq |e_{\boldsymbol{\lambda}_{max}}|^2  \left(   2^{2n} \sum_{{\boldsymbol{\lambda}} \notin \I}    \frac{1}{d_{\boldsymbol{\lambda}}} |\<0| T^{\lambda_D}_{k_D} |0\>|^2 |O_{\lambda,k} (0) |^2\right) 
\end{align}
where in the first line we used  the Parseval's theorem and in the second line we used that $|O_{\lambda,k} (x) |^2 = |O_{\lambda,k}(0)|^2$. 
What does anti-concentration tell us in addition? It allows us to put an appropriate bound on the above quantity, in particular it is equivalent to saying that the term in the bracket is bounded.

Let's se what it gives us:
\begin{align}
	O(1) & =  2^n \sum_x \int |p(C(\g),x) |^2 d\, \g \\
	& = 2^n \sum_x || \<0| \O_{\g}(x) |0\>||_{L^2(G)}\\
	& = 2^n \sum_x \sum_{\boldsymbol{\lambda}} |O_{\lambda,k}(x)|^2 \frac{1}{d_{\boldsymbol{\lambda}}} |\<0|T^{\lambda_D}_{k_D} |0\>|^2\\
	& = 2^{2n} \sum_{\boldsymbol{\lambda}} \frac{1}{d_{\boldsymbol{\lambda}}} |\<0|T^{\lambda_D}_{k_D} |0\>|^2 |O_{\lambda,k} (0)|^2
\end{align}

\end{document}